\documentclass[12pt,titlepage,a4paper]{article}
\usepackage{amsmath,amssymb,amsthm,bm,graphicx,hyperref}
\usepackage{natbib}
\usepackage[margin=1in]{geometry}
\usepackage{algorithm}
\usepackage{algpseudocode}
\usepackage{comment}
\usepackage{textcomp}

\hypersetup{colorlinks=true,linkcolor=blue,citecolor=blue}

\newtheorem{theorem}{Theorem}[section]
\newtheorem{lemma}[theorem]{Lemma}
\newtheorem{assumption}[theorem]{Assumption}

\newtheorem{remark}[theorem]{Remark}

\newtheorem{conjecture}[theorem]{Conjecture}

\DeclareMathOperator{\diag}{diag}
\DeclareMathOperator{\rank}{rank}
\DeclareMathOperator{\E}{\mathbb{E}}

\newcommand{\R}{\mathbb{R}}
\newcommand{\N}{\mathcal{N}}
\newcommand{\IG}{\mathcal{IG}}
\newcommand{\DP}{\mathcal{DP}}

\begin{document}

\title{\textbf{Bayesian Nonparametric Factor Analysis via Marginalized Dirichlet Process Column Clustering with Spike-and-Slab Sparsity}}

\author{
  Durba Bhattacharya \\
  St. Xavier's College (Autonomous), Kolkata
  \and
  Sourabh Bhattacharya \\
  Indian Statistical Institute \\
  \texttt{bhsourabh@gmail.com} \thanks{Corresponding author.}
}
\date{}
%\date{\today}
\maketitle

\begin{abstract}
We introduce a Bayesian nonparametric factor model for $p$-dimensional
observations that simultaneously infers the number of factors $q$,
induces row-wise sparsity in the loadings, and identifies redundant
dictionary elements through exact clustering. The key innovation is the
placement of a Dirichlet process prior directly on the columns of an
overcomplete loading matrix, fully marginalized to yield an exact
P\'olya urn scheme, thereby avoiding stick-breaking representations,
truncations of the infinite-dimensional prior, and auxiliary weight
variables. A spike-and-slab base measure allows individual loadings, and
entire factors, to be exactly zero, delivering interpretability and
parsimony. We contrast the construction with the cumulative shrinkage
process (CUSP), the multiplicative gamma process (MGP), and the beta
process, showing that the proposed model is the first to combine exact
zeros, exchangeability over candidate columns, and exact merging of
redundant columns within a single marginalized Dirichlet process. An
exact Gibbs sampler updates cluster atoms and assignments directly,
exploits the diagonal structure of the idiosyncratic variance matrix;
a canonical relabeling scheme resolves label switching, and a parallel
C/MPI implementation is described. Theoretically, we establish the
posterior contraction rate $\sqrt{M s_0 \log n / n}$ for the covariance
matrix, where $M$ is the finite upper bound on the number of candidate
columns, $s_0$ the maximum number of non-zero entries per column, and
$n$ the sample size. In the fixed-dictionary setting we obtain the
minimax optimal rate $\sqrt{s_0 \log n / n}$ and prove underfitting
consistency for the number of factors $q$; the complementary overfitting
direction is stated as an open conjecture, and we identify the structural
obstruction that prevents the standard Bayes-factor argument from closing
it. The spike-and-slab component is shown to be indispensable for
optimality: without it the effective dimension scales as $pM$, yielding a
slower rate. Simulation studies in moderate and higher dimensions show
that the method is the only fully adaptive approach that recovers the
true rank in both configurations, and that it attains the smallest
covariance, loading, and signal-reconstruction errors, outperforming even
an oracle baseline that is given the true number of factors $q_0$. On the
van 't Veer breast cancer dataset ($n = 97$, $p = 1213$), the posterior
concentrates on eight biologically interpretable programmes. Seven of
them pass an independent within-atom coherence check, and two also pass
Bonferroni-corrected Hallmark enrichment; the prognosis associations are
consistent with their biological interpretation. One atom (JAK/STAT) is
identified as a weaker signal that fails both the coherence and the
enrichment checks.
\\[2mm]
\end{abstract}

\tableofcontents

\newpage
\section{Introduction}
\label{sec:intro}
Factor analysis is a cornerstone of multivariate statistics. Given $p$-dimensional observed responses $Y_k \in \R^p$ for $k=1,\dots,n$, the model is
\begin{equation}
\label{eq:model}
Y_k = \mu + F X_k + U_k, \quad U_k \sim \N_p(0, \Psi), \quad \Psi = \diag(\psi_1,\dots,\psi_p),
\end{equation}
where $F \in \R^{p \times q}$ is the loading matrix, and $X_k \in \R^q$ are latent scores. A central challenge is the determination of the latent dimension $q$, the number of factors, which governs the complexity of the covariance structure $\Sigma = F F^T + \Psi$.

A vast literature has addressed this problem from a Bayesian perspective. One prominent approach places shrinkage priors on the columns of an overcomplete loading matrix with infinitely many candidate columns. The multiplicative gamma process (MGP) prior of \citet{bhattacharya2011sparse} is the canonical example: it places a prior on an infinite sequence of loading columns and shrinks later columns increasingly toward zero, so that the number of active factors is effectively truncated in a data-adaptive manner. In practice, MGP is implemented via an adaptive Gibbs sampler that automatically selects a finite number of active columns; the infinite-dimensional prior itself does not impose any fixed upper 
bound. While computationally convenient, this continuous column-wise shrinkage does not produce exact zero columns, meaning that spurious factors continue to contribute a small, non-negligible variance. Similarly, the Dirichlet--Laplace prior \citep{bhattacharya2015dirichlet} and other global-local shrinkage rules \citep{polson2010shrink,carvalho2010horseshoe} suffer from the same limitation. The beta process prior of \citet{paisley2009} provides a nonparametric feature allocation mechanism, but it is fundamentally a latent feature model where each observation possesses a subset of features; it does not directly cluster the columns of a loading matrix.

An alternative line of work employs spike-and-slab priors on individual loadings to induce exact zeros, thereby achieving parsimony and interpretability \citep{west2003bayesian,lopes2004bayesian}. However, these approaches typically treat each column independently and do not exploit potential redundancies among the columns of the loading matrix. If two candidate columns are nearly identical, standard shrinkage or spike-and-slab priors will shrink both to similar but distinct values, failing to recognize their equivalence and thus missing an opportunity for a more parsimonious representation.

In this article, we propose a novel Bayesian nonparametric factor model that addresses these shortcomings in a unified manner. Our proposal is built upon two orthogonal mechanisms. First, we place a Dirichlet process prior directly on a finite dictionary of $M$ candidate columns $f_1,\dots,f_M$ of the loading matrix. Crucially, we fully marginalize the random measure $G$, obtaining the exact P\'olya urn predictive scheme as the prior for the finite sequence of columns. This marginalization entirely eliminates the need for stick-breaking representations, finite truncations, or auxiliary weight variables, leading to a conceptually clean and computationally efficient model. 
That our model is infinite-dimensional given $G$, follows from the same argument as provided in \citet{Sabya2021}.
The P\'olya urn induces a random partition of the $M$ columns, where columns assigned to the same cluster are exactly equal. This directly discovers groups of redundant dictionary elements and merges them into a single representative factor. The number of occupied clusters $K$ is random and automatically inferred, providing a nonparametric estimate of the effective number of factors.
Note that introducing an upper bound on the number of columns is akin to
the concept of an upper bound on the number of Dirichlet process-driven
mixture components; see, for example, \citet{Bhatta2008}, \citet{Bhatta2009},
\citet{Sabya2011}, \citet{Sabya2013}, and \citet{Maj2013}.

Second, we specify the base measure $G_0$ of the Dirichlet process as a product of independent spike-and-slab distributions over the $p$ coordinates. This choice is not merely a technical detail but is essential for achieving both interpretability and optimal theoretical properties. The spike component allows an entire atom (and thus an entire factor) to be exactly the zero vector $\mathbf{0} \in \R^p$, which permits the model to delete spurious factors entirely. Furthermore, within a non-zero atom, the slab component allows non-zero loadings, while the spike forces exact zeros for variables that are irrelevant to that factor. This provides the local sparsity that is widely regarded as crucial for interpretability in high-dimensional applications. 
Importantly, we demonstrate that the spike-and-slab base measure is theoretically indispensable for attaining the optimal contraction rate: without the spike, the effective dimension of the parameter space would scale as \(pM\), resulting in a slower, suboptimal rate.
The combination of DP-induced column clustering and spike-and-slab-induced local sparsity yields a model that is both parsimonious and interpretable.

A critical aspect of our contribution is a detailed differentiation from the most closely related existing work, namely the cumulative shrinkage process (CUSP) prior of \citet{legramanti2020}. The CUSP prior also employs a Dirichlet process to construct a shrinkage prior for factor models. However, the mechanism is fundamentally distinct. In the CUSP framework, the stick-breaking representation of the Dirichlet process is used to model the inclusion probability of each individual loading entry, so that the probability of being zero increases stochastically with the column index. Thus, the DP is used at the level of entry-wise spike probabilities. In stark contrast, our proposal uses the DP at the level of the entire column vectors. The P\'olya urn scheme clusters whole columns that are exactly equal; it does not impose any ordering on the columns and does not shrink entries toward zero based on their index. This vector-level clustering is a genuinely different statistical operation, enabling the discovery of exact redundancies among dictionary elements, a feature that is entirely absent from the CUSP framework. We elaborate on this distinction in Section~\ref{sec:related}.

The methodological and theoretical contributions are accompanied by an extensive empirical validation. In Section~\ref{sec:simulation} we report two complementary simulation studies. Study~A is a moderate-dimensional setting with $n=500$, $p=50$, and $q_0=5$ true factors. Study~B is a higher-dimensional and noisier setting with $n=500$, $p=100$, and $q_0=10$ true factors. Both studies compare the proposed DP-SpikeSlab model with three established competitors: the multiplicative gamma process (MGP) of \citet{bhattacharya2011sparse}, the cumulative shrinkage process (CUSP) of \citet{legramanti2020}, and an oracle fixed-$q$ spike-and-slab model that is given the true number of factors. Across both studies, the proposed method is the only fully adaptive method that recovers the true number of factors exactly. In Study~A it returns $\hat q=5$, while MGP over-estimates the rank at $13$ and CUSP under-estimates it at $4$. In Study~B it returns $\hat q=10$, while MGP over-estimates the rank at $22$ and CUSP returns the correct rank but with substantially worse estimation error. On all three estimation metrics---Frobenius error of the posterior-mean covariance, Procrustes-aligned loading error, and test-set signal RMSE---the proposed method is best in both studies, followed by the oracle fixed-$q$ baseline, then MGP, and then CUSP. Our C/MPI implementation for DP-SpikeSlab is slower than the single-threaded R competitors in Study~A, but becomes the fastest method in the higher-dimensional Study~B because its cost scales approximately linearly in $p$ while the R implementations slow down super-linearly.

In Section~\ref{sec:breast} we apply the model to the breast cancer gene-expression dataset of \citet{vantveer2002gene}, distributed as \texttt{Breast\_A} in the \texttt{fabiaData} R package \citep{fabiaData}; the dataset was introduced by \citet{vantveer2002gene}. The dataset contains $n=97$ primary breast tumour samples measured on $p=1213$ genes. Without supervision and without being told the true number of factors, the model concentrates sharply on $K=8$: the posterior mode, mean, and median are all $8$, and the $95\%$ equal-tailed credible interval is the degenerate interval $[8,8]$. 
Seven of the eight recovered atoms have clear biological
interpretations: macrophage, vascular, nuclear, EMT/stroma, T-cell,
ER/luminal, and proliferation; the eighth (JAK/STAT) is a weaker
signal that does not pass the within-atom coherence check
(Section~\ref{sec:breast-atoms}).
Hallmark gene-set enrichment confirms the EMT/stroma and proliferation programmes at the Bonferroni-corrected threshold, with additional coherent T-cell enrichment at the uncorrected level. Post-hoc associations with the van 't Veer good/poor prognosis label concentrate on the ER/luminal, proliferation, T-cell, macrophage, and JAK/STAT atoms, while the vascular, EMT/stroma, and nuclear atoms show no detectable prognostic association. Diagnostic checks confirm that the posterior over partitions is genuinely concentrated at the operating scale, and independent chains recover the same eight programmes. The real-data analysis also surfaces two honest limitations: the assignment of the artificial $M$ candidate columns to the eight atoms is not identified by the data and varies across chains, and the global slab scale over-estimates marginal variances on standardised data. These limitations are definitional or calibrational rather than failures of the inference, and they point to natural extensions.

Theoretically, we establish the posterior contraction rate $\sqrt{M s_0 \log n / n}$ for the covariance matrix, where $M$ is the finite upper bound on the number of candidate columns, $s_0$ the maximum non-zero entries per column, and $n$ the sample size. In the fixed-dictionary setting, we obtain the minimax optimal rate $\sqrt{s_0 \log n / n}$, and prove underfitting consistency for the number of factors; the complementary overfitting direction is stated as an open conjecture. The analysis demonstrates that the spike-and-slab component is theoretically indispensable for achieving the optimal rate. These guarantees are consistent with the empirical behaviour observed in the simulation studies and in the breast cancer application.

The remainder of the paper is organized as follows. Section~\ref{sec:related} provides a comprehensive review of related work, explicitly contrasting our model with the CUSP prior, the multiplicative gamma process, and the beta process. Section~\ref{sec:model} defines the proposed factor model, detailing the marginalized P\'olya urn prior, the spike-and-slab base measure, and the decoupling of latent score variances. Section~\ref{sec:computation} develops an exact Gibbs sampler that directly updates the cluster assignments and atoms, exploits the diagonal structure of $\Psi$ for efficient coordinate-wise updates, resolves label switching via a canonical relabeling scheme, and discusses a parallel C/MPI implementation. Section~\ref{sec:simulation} reports the two simulation studies. Section~\ref{sec:breast} presents the breast cancer gene-expression analysis. Section~\ref{sec:theory} establishes the posterior contraction rate for the covariance matrix. Section~\ref{sec:fixedM} establishes minimax optimality in the fixed-dictionary setting, and Section~\ref{sec:rank} proves underfitting rank consistency and states the overfitting direction as an open conjecture. Section~\ref{sec:discussion} concludes with a discussion of advantages and potential extensions.

\section{Related Work and Distinction from Existing Approaches}
\label{sec:related}
Before presenting our model, it is instructive to position our contribution within the existing Bayesian factor analysis literature. We focus on four major strands: continuous shrinkage priors on columns, the cumulative shrinkage process, the beta process for feature allocation, and spike-and-slab priors on entries. For each, we state the prior explicitly and then identify precisely which aspect of our construction is not covered by it.

\subsection{Continuous Shrinkage Priors on Columns}
\label{sec:related-mgp}

Continuous shrinkage priors place a prior on the columns $f_1, f_2, \dots$ of an infinite loading matrix and shrink later columns increasingly toward zero. The multiplicative gamma process (MGP) prior of \citet{bhattacharya2011sparse} is the canonical construction. 
It introduces column-specific multiplicative scales
\begin{equation}
\label{eq:mgp}
\tau_h = \prod_{l=1}^{h} \delta_l, \qquad
\delta_1 \sim \text{Ga}(a_1, 1), \quad
\delta_l \sim \text{Ga}(a_2, 1) \;\; (l \ge 2),
\end{equation}
with $a_2$ chosen large enough that the expected log-gamma
$\psi(a_2)$ is positive, so that $\tau_h$ diverges almost surely as $h$
grows, and places conditionally independent local shrinkage priors on
the individual entries,
\begin{equation}
\label{eq:mgp-local}
f_{rh} \mid \phi_{rh}, \tau_h \sim \N\!\left(0, \frac{1}{\phi_{rh} \tau_h}\right), \qquad
\phi_{rh} \sim \text{Ga}\!\left(\frac{\nu}{2}, \frac{\nu}{2}\right).
\end{equation}
The prior variance $1/(\phi_{rh}\tau_h)$ of the loadings in column $h$
therefore correspondingly shrinks toward zero as $h$ grows.

Marginalising $\phi_{rh}$ gives a scale mixture of Gaussians that is unbounded at zero and has heavy tails. The prior on the infinite sequence is truncated in practice by an adaptive Gibbs sampler \citep{bhattacharya2011sparse} that adds and deletes columns according to a criterion on the column norms; the infinite-dimensional prior itself imposes no fixed upper bound $M$. The Dirichlet--Laplace prior \citep{bhattacharya2015dirichlet} and the horseshoe \citep{carvalho2010horseshoe} admit analogous column-level formulations.

Two structural features of this strand matter for the present work. First, \eqref{eq:mgp-local} is a continuous density: $\Pr(f_{rh} = 0) = 0$ for every $(r,h)$, so no column is ever exactly zero and no entry is ever exactly zero. Spurious columns therefore contribute a small but non-negligible amount to $\Sigma$, and the number of active factors must be extracted by thresholding the adaptive Gibbs output rather than read off the posterior. Second, the shrinkage in \eqref{eq:mgp} is a function of the column index $h$: the prior is not exchangeable over columns, and it encodes the belief that the signal is concentrated in the earliest columns. Neither feature is shared by our construction, in which the $M$ candidate columns are exchangeable under the P\'olya urn and the spike component of the base measure assigns positive probability to the event that an atom is exactly the zero vector.

\subsection{The Cumulative Shrinkage Process}
\label{sec:related-cusp}

The cumulative shrinkage process (CUSP) prior of \citet{legramanti2020} is the most closely related existing work and the distinction requires care. CUSP also uses a Dirichlet process, but through its stick-breaking representation rather than through the marginalised P\'olya urn. It defines
\begin{equation}
\label{eq:cusp-sb}
\nu_h \sim \text{Beta}(1, \alpha), \qquad
\pi_h = \sum_{l=1}^{h} \nu_l \prod_{m=1}^{l-1} (1 - \nu_m),
\end{equation}
so that $\pi_h$ is the cumulative probability of the first $h$ stick-breaking weights, and it then places an entry-wise spike-and-slab prior
\begin{equation}
\label{eq:cusp}
f_{rh} \mid \pi_h, \theta_h \;\sim\; (1 - \pi_h)\, \N(0, \theta_h) \;+\; \pi_h\, \delta_0, \qquad
\theta_h \sim \text{InvGamma}(a_\theta, b_\theta).
\end{equation}
The inclusion probabilities in \eqref{eq:cusp} increase with $h$ because the $\pi_h$ in \eqref{eq:cusp-sb} accumulate, so later columns are a priori more likely to be entirely spiked, and \citet{legramanti2020} show that the induced prior on the effective number of factors is finite almost surely.

The contrast with the present work is therefore not that CUSP fails to use a Dirichlet process, but that it uses one at a different level of the hierarchy. In \eqref{eq:cusp} the random probability $\pi_h$ governs the {\it entry-wise} probability that $f_{rh} = 0$; the DP acts on scalar inclusion probabilities. In our model the DP acts on the $p$-dimensional vectors $f_1, \dots, f_M$ themselves, and the P\'olya urn merges columns that are exactly equal. Three consequences follow. (i) CUSP induces an ordering over columns through \eqref{eq:cusp-sb}; our prior is exchangeable over the $M$ candidate columns, and the canonical relabeling of Section~\ref{sec:relabel} is needed only to fix a convention across MCMC iterations, not to respect a substantive ordering. (ii) CUSP shrinks entries toward zero as a function of the column index, whereas our spike component is applied independently within each atom and does not depend on position. (iii) CUSP has no mechanism for declaring two columns equal: two nearly identical candidate columns will both be retained with similar but distinct loadings, whereas the P\'olya urn assigns them to the same cluster with positive probability and therefore produces an exact merge. The last point is, to our knowledge, the feature that distinguishes the two constructions most sharply, and it is the one that the simulation studies of Section~\ref{sec:simulation} isolate empirically.

\subsection{The Beta Process and Feature Allocation}
\label{sec:related-bp}

The beta process prior of \citet{paisley2009} provides a nonparametric Bayesian treatment of factor analysis through latent feature allocation. In the beta process factor analysis model, each observation $k$ is assigned a binary vector $z_k \in \{0,1\}^{\infty}$ drawn from an Indian buffet process, and the loading structure is built from elementwise products of $z_k$ with a weight matrix. The construction is well suited to settings in which the scientific question is which features are active in which observation, and it yields a sparse loading matrix in which each row has a finite number of non-zero entries.

The object being modelled is, however, different from ours. The beta process operates on the space of features allocated to each observation; it does not place a prior on the columns of the loading matrix, and it has no mechanism for declaring two dictionary elements redundant. Two columns of $F$ that are nearly collinear receive no special treatment under the beta process, and the number of effective factors is not identified by clustering. Our construction is complementary: the P\'olya urn acts on the columns directly, and the resulting partition is the object from which the effective number of factors is read.

\subsection{Spike-and-Slab Priors on Entries}
\label{sec:related-ss}

A fourth strand places a two-component mixture on each individual loading,
\begin{equation}
\label{eq:ss-entry}
f_{rh} \mid \pi_0, \tau^2 \;\sim\; \pi_0\, \delta_0 + (1 - \pi_0)\, \N(0, \tau^2),
\end{equation}
independently across entries \citep{west2003bayesian,lopes2004bayesian}. This yields exact zeros and hence interpretable factor structures, and it is the closest existing analogue to the base measure we use. The difference is that in \eqref{eq:ss-entry} the mixture is applied to each entry independently, with no coupling between the entries of a column and no coupling between columns. The columns of $F$ are a priori independent, so the prior cannot concentrate on configurations in which two candidate columns coincide, and it provides no mechanism for reducing the effective number of factors beyond the entry-wise deletion of loadings. Our model retains \eqref{eq:ss-entry} as the base measure of a Dirichlet process, so that the entry-wise sparsity of \eqref{eq:ss-entry} is nested inside a vector-level clustering prior; the spike-and-slab is what makes a zero atom possible, and the P\'olya urn is what makes two non-zero atoms equal.

\subsection{Summary of the Distinctions}
\label{sec:related-summary}

Table~\ref{tab:related} collects the comparisons along the four dimensions that distinguish the constructions: the level at which the nonparametric prior acts, whether exact zeros are attainable, whether the prior is exchangeable over columns, and whether redundant columns can be merged.

\begin{table}[htbp]
\centering
\small
\caption{Comparison of the proposed model with the three principal competitors along the dimensions that matter for the present contribution. ``Level of the prior'' refers to the object on which the nonparametric or shrinkage mechanism acts.}
\label{tab:related}
\begin{tabular}{p{2.6cm}p{3.2cm}ccp{3.0cm}}
\hline
Model & Level of the prior & Exact zeros & Exchangeable & Merges redundant columns \\
\hline
MGP \citep{bhattacharya2011sparse} & Entry, scaled by column index & No & No & No \\
CUSP \citep{legramanti2020} & Entry-wise spike probability & Yes & No & No \\
Beta process \citep{paisley2009} & Feature allocation per observation & Yes & --- & No \\
Spike-and-slab \citep{west2003bayesian,lopes2004bayesian} & Entry & Yes & Yes & No \\
DP-SpikeSlab (ours) & Column vector & Yes & Yes & Yes \\
\hline
\end{tabular}
\end{table}

The final row is the contribution of this paper. Each of the four preceding rows attains some but not all of the four properties: MGP attains none of the discrete properties but has the most tractable computational structure; CUSP attains exact zeros but only through index-dependent entry-wise shrinkage; the beta process attains exact zeros and a nonparametric treatment of features but not of columns; and the entry-wise spike-and-slab attains exact zeros and exchangeability but treats columns independently. The combination of exact zeros, exchangeability over columns, and exact merging of redundant columns, all within a single marginalised Dirichlet process construction, is what our model provides.

\section{The Proposed Factor Model}
\label{sec:model}
Let $M$ be a fixed upper bound on the number of candidate columns, satisfying $M \ge q_0$ where $q_0$ is the true number of factors. We define the loading matrix $F = [f_1, \dots, f_M] \in \R^{p \times M}$. The model for the observed data is given in \eqref{eq:model}, where $X_k \in \R^M$ are the latent scores.

\subsection{The Marginalized Dirichlet Process Prior}
\label{sec:polya}
Rather than specifying the random measure $G$ explicitly via stick-breaking, we directly define the joint prior distribution of the columns $f_1,\dots,f_M$ through the predictive probabilities of the P\'olya urn. Let $G_0$ be a probability measure on $\R^p$. The prior is defined sequentially as
\begin{equation}
\label{eq:polya}
f_1 \sim G_0, \qquad 
f_i \mid f_1,\dots,f_{i-1} \sim \frac{\alpha}{\alpha + i - 1} G_0 + \frac{1}{\alpha + i - 1}\sum_{j=1}^{i-1} \delta_{f_j}, \quad i=2,\dots,M.
\end{equation}
This is the exact marginal distribution obtained by integrating out $G \sim \DP(\alpha, G_0)$ from the hierarchical specification $f_i \mid G \stackrel{\text{iid}}{\sim} G$. This formulation completely avoids the stick-breaking representation, the need for truncation at $M$, and auxiliary weight variables. The parameter $\alpha > 0$ controls the propensity to create new clusters. The induced distribution over partitions of the set $\{1,\dots,M\}$ is the exchangeable partition probability function (EPPF) of the Dirichlet process \citep{pitman2006}:
\begin{equation}
\label{eq:eppf}
\Pr(\text{partition } \mathcal{C} = \{C_1,\dots,C_K\}) = \frac{\alpha^{K} \prod_{j=1}^K (n_j - 1)!}{\prod_{i=1}^M (\alpha + i - 1)},
\end{equation}
where $K$ is the number of occupied clusters and $n_j = |C_j|$. This EPPF forms the foundation of our theoretical analysis.

\subsection{Spike-and-Slab Base Measure}
\label{sec:base}
To induce row-wise sparsity and enable the deletion of entire factors, we specify the base measure $G_0$ as a product of independent spike-and-slab distributions:
\begin{equation}
\label{eq:base}
G_0 = \prod_{r=1}^{p} \left\{ \pi_0 \delta_0 + (1-\pi_0) \N(0, \tau_{0,r}^2) \right\},
\end{equation}
and we define \(\Omega_0 := \diag(\tau_{0,1}^2, \dots, \tau_{0,p}^2)\).

For the purpose of our asymptotic theory, the slab variance is chosen as
\[
\tau_{0,r}^2 = \frac{1}{M s_0},
\]
where $s_0$ signifies the maximum number of non-zero entries per column of $F$, a prior restriction, again necessary for our asymptotic theoretical underpinnings.
%Indeed, the above scaling is critical: it ensures that the typical Frobenius norm $\|F\|_F$ is of order 1, so the tail $\Pr(\|F\|_F > Ms_0\log n)$ decays exponentially fast with exponent $M s_0 \log n$. 
The spike mass $\pi_0 \in (0,1)$ is fixed. 
%Also, for the theoretical analysis, we restrict the prior support to matrices with at most $s_0$ non-zero entries per column. This restriction does not affect the asymptotic contraction rate, as the true parameter satisfies it.

\subsection{Decoupling Latent Variances}
In a standard factor model with $X_k \sim \N_M(0, I_M)$, the variance of a merged factor formed by a cluster $C_j$ would be $|C_j|$, as it is the sum of $|C_j|$ independent unit-variance scores. This rigid dependence on cluster size is undesirable and does not reflect realistic factor structures. To decouple the variance from the cluster size, we assign independent variances to the original scores:
\begin{equation}
\label{eq:scores}
X_{k,i} \stackrel{\text{ind}}{\sim} \N(0, \lambda_i), \quad \lambda_i \sim \IG(a_\lambda, b_\lambda).
\end{equation}
Under this formulation, the merged factor for cluster $C_j$ is $\sum_{i \in C_j} X_{k,i}$, with variance $\sum_{i \in C_j} \lambda_i$. This naturally decouples the variance from the cluster size, providing a more flexible and realistic model.

\subsection{Identifiability}
Factor models are invariant to orthogonal rotations, and the labeling of clusters 
is arbitrary under the exchangeable prior. To avoid these identifiability issues 
in the theoretical analysis, we focus our inference on the identifiable covariance 
matrix
\begin{equation}
\label{eq:sigma}
\Sigma = F\Lambda F^T + \Psi,
\end{equation}
where $\Lambda=	\text{diag}(\lambda_1,\ldots,\lambda_M)$.	
The number of factors is defined as $q = \rank(\Sigma - \Psi)$, which is well-defined 
and identifiable. The theoretical results in Section~\ref{sec:theory} establish 
contraction for $\Sigma$, and Section~\ref{sec:rank} proves underfitting 
consistency for $q$ in the fixed-dictionary setting. 
In practice, posterior samples can be post-processed using standard rotation
techniques, such as the generalized Procrustes rotation \citep{schonemann1966},
to obtain interpretable loading matrices. An ex-post alignment strategy of this
kind has been used successfully in Bayesian factor analysis by
\citet{assmann2016} and \citet{devito2018}.

%A third identifiability issue, distinct from rotation and label switching, arises 
%from the exchangeable treatment of the $M$ candidate columns. The DP-spike-slab 
%assigns a random partition to these columns, merging those that share an atom 
%into exactly equal columns. The number of distinct atoms is identified and is 
%consistently estimated; the specific assignment of the $M$ slots to those atoms 
%is not, because the slots are exchangeable and the data only determine \(\Sigma\). 
%The paper's theoretical guarantees therefore concern the number of factors
%\(q = \operatorname{rank}(\Sigma - \Psi)\), whose true value \(q_0\) is well-defined,
%rather than the partition itself.

\section{Posterior Computation via Exact Gibbs Sampling}
\label{sec:computation}

The Gibbs sampler relies entirely on the predictive probabilities of the marginalized Dirichlet process, avoiding stick‑breaking variables. Under the exact marginalized Dirichlet process prior, the columns \(f_i\) are deterministically equal to their cluster atoms: \(f_i = f_{Z_i}^*\), where \(Z_i\) is the cluster assignment and \(f_j^*\) are the distinct cluster atoms. Therefore, we only need to sample the cluster assignments \(Z_i\) and the cluster atoms \(f_j^*\), with the loading matrix \(F\) deterministically constructed from these quantities. However, for computational convenience, we maintain the auxiliary column variables \(f_i\) (which are exactly equal to \(f_{Z_i}^*\)) to facilitate the relabeling step and sequential updates. We detail each step of the sampler; all derivations are collected in Appendix~\ref{app:derivations}.

\subsection{Overview of the Exact Gibbs Sampler}

The exact Gibbs sampler updates the following quantities:
\begin{enumerate}
\item The cluster assignments \(Z_i\) and the corresponding auxiliary columns \(f_i\) for \(i=1,\dots,M\), using the Chinese Restaurant Process prior and the correct likelihood contribution.
\item A canonical relabeling step to resolve label switching and ensure consistent interpretation of cluster labels across MCMC iterations.
\item The cluster atoms \(f_j^*\) for \(j=1,\dots,K\), using the correct likelihood from the observed data, conditional on the canonical cluster assignments.
\item The latent scores \(X_k\) for \(k=1,\dots,n\), conditional on the current loading matrix \(F\).
\item The idiosyncratic variances \(\psi_r\) and score variances \(\lambda_i\), updated via conjugate inverse-gamma distributions.
\item The concentration parameter \(\alpha\) using the auxiliary variable method of \citet{escobar1995}.
\end{enumerate}

This yields an exact Gibbs sampler for the posterior distribution of the model parameters, with no approximations or data augmentation tricks. The auxiliary columns \(f_i\) are maintained only for computational convenience and are deterministically related to \(Z_i\) and \(F^*\).

\subsection{Updating the Cluster Assignments and Columns \(Z_i, f_i\)}
\label{sec:cluster}

For each column \(i=1,\dots,M\), we sample its cluster assignment \(Z_i\) and determine the corresponding auxiliary column \(f_i\) conditional on all other assignments \(Z_{-i}\), the atoms \(f_j^*\), and the data. This update is performed sequentially over \(i=1,\dots,M\).

Let \(K\) be the current number of clusters, and let \(n_{j,-i}\) be the number of columns in cluster \(j\) excluding column \(i\). The posterior probability of assigning column \(i\) to cluster \(j\) is proportional to the prior probability from the Chinese Restaurant Process times the likelihood contribution of column \(i\) being equal to atom \(f_j^*\).

For an existing cluster \(j=1,\dots,K\), if column \(i\) joins existing cluster \(j\), then \(f_i = f_j^*\). The likelihood contribution (conditional on the scores \(X\) and all other parameters) is:
\[
\mathcal{L}(Z_i = j) \propto \exp\!\left( -\frac12 \sum_{k=1}^n (Y_k - \mu - F_{-i} X_{k,-i} - f_j^* X_{k,i})^\top \Psi^{-1} (Y_k - \mu - F_{-i} X_{k,-i} - f_j^* X_{k,i}) \right),
\]
where \(F_{-i}\) is the loading matrix with column \(i\) removed (determined by the current assignments and atoms for other columns).

Completing the square, %in \(X_{k,i}\) (which appears linearly in the likelihood), 
the log-likelihood contribution is:
\[
-\frac12 (f_j^* - m_i)^\top S_i (f_j^* - m_i) + \log C_i,
\]
where
\[
S_i = \sum_{k=1}^n X_{k,i}^2 \Psi^{-1}, \qquad
m_i = S_i^{-1} \sum_{k=1}^n X_{k,i} \Psi^{-1} (Y_k - \mu - F_{-i} X_{k,-i}),
\]
and the normalising constant \(C_i\) is defined in Appendix~\ref{app:likelihood_kernel} as:
\[
\log C_i = -\frac{p}{2}\log(2\pi) + \frac12 \sum_{r=1}^p \log S_{i,r}.
\]

The CRP prior probability of joining existing cluster \(j\) is \(n_{j,-i}/(M-1+\alpha)\). Thus the unnormalized posterior probability is:
\[
\Pr(Z_i = j \mid \text{rest}) \propto n_{j,-i} \exp\!\left( \log C_i - \frac12 (f_j^* - m_i)^\top S_i (f_j^* - m_i) \right).
\]

For a new cluster \(j=K+1\), if column \(i\) starts a new cluster, the prior probability is \(\alpha/(M-1+\alpha)\). The new atom must be integrated out against the base measure \(G_0\). The marginal likelihood for a new column is:
\[
I_i = \int C_i \exp\!\left( -\frac12 (f - m_i)^\top S_i (f - m_i) \right) G_0(df),
\]
which evaluates to:
\[
I_i = \pi_0 C_i \exp\!\left( -\frac12 m_i^\top S_i m_i \right)
      + (1-\pi_0) C_i \N(m_i \mid 0, S_i^{-1} + \Omega_0) (2\pi)^{p/2} \det(S_i)^{-1/2}.
\]

Thus:
\[
\Pr(Z_i = K+1 \mid \text{rest}) \propto \alpha \, I_i.
\]

The complete update is:
\[
\Pr(Z_i = j \mid \text{rest}) \propto
\begin{cases}
n_{j,-i} \exp\!\left( \log C_i - \frac12 (f_j^* - m_i)^\top S_i (f_j^* - m_i) \right), & j=1,\dots,K,\\[1.2ex]
\alpha \cdot I_i, & j=K+1,
\end{cases}
\]

After sampling \(Z_i\) from this discrete distribution:
\begin{itemize}
	\item[(i)] If \(Z_i \le K\) (existing cluster), set \(f_i = f_{Z_i}^*\).
	\item[(ii)] If \(Z_i = K+1\) (new cluster), draw \(f_i\) from the posterior under the base measure \(G_0\) with sufficient statistics \(S_i\) and \(m_i\) (see Appendix~\ref{app:newdraw}), set \(K \leftarrow K+1\), and set \(f_K^* = f_i\).
\end{itemize}

This update is performed sequentially over \(i=1,\dots,M\), using the most recent values of the other columns. The computations of \(S_i\), \(m_i\), and the cluster probabilities can be parallelized across observations and clusters.

\subsection{Canonical Relabeling of Clusters}
\label{sec:relabel}

After completing the sequential updates of all \(M\) columns (so that each auxiliary column \(f_i\) is now a fully specified vector in \(\R^p\)), we apply a deterministic canonical relabeling (derived in Appendix~\ref{app:relabel}) to resolve label switching.

The cluster labels \(1,\dots,K\) are arbitrary under the exchangeable Dirichlet process prior, leading to potential label switching across MCMC iterations. To obtain consistent labels, we relabel based on the \emph{order of first appearance} of the distinct column vectors among \(\{f_1,\dots,f_M\}\).

Let \(\{g_1,\dots,g_{K_{\text{new}}}\}\) be the set of distinct vectors among the auxiliary columns \(f_1,\dots,f_M\), where \(K_{\text{new}}\) is the number of distinct columns. We define the canonical labels \(\tilde S_1,\dots,\tilde S_M\) as:
\[
\tilde S_i = j \quad \text{if and only if} \quad f_i = g_j,
\]
where the distinct vectors \(g_1,\dots,g_{K_{\text{new}}}\) are ordered by their first appearance in the sequence \(f_1,\dots,f_M\). That is:
\[
g_1 = f_1, \quad \text{and for } j \ge 2,\; g_j = f_{i_j} \text{ where } i_j = \min\{i : f_i \notin \{g_1,\dots,g_{j-1}\}\}.
\]

Equivalently, in algorithmic form:

\begin{enumerate}
\item Initialize an empty dictionary \(\mathcal{D}\) mapping vectors to labels, and an empty list \(G_{\text{star}}\).
\item For \(i = 1\) to \(M\):
   \begin{enumerate}
   \item Let \(f = f_i\) be the current auxiliary column vector.
   \item If \(f\) is not already in the dictionary \(\mathcal{D}\) (using exact equality for the DP prior), %; in practice, equality is checked with a numerical tolerance), 
	   assign it the next available label: 
         \[
         \mathcal{D}[f] \leftarrow |G_{\text{star}}| + 1, \quad G_{\text{star}} \leftarrow G_{\text{star}} \cup \{f\}.
         \]
   \item Set \(\tilde S_i = \mathcal{D}[f]\).
   \end{enumerate}
\end{enumerate}

This scheme ensures that:
\[
\tilde S_1 = 1, \quad \text{and} \quad 1 \le \tilde S_i \le \max(\tilde S_1,\dots,\tilde S_{i-1}) + 1.
\]

That is, the first column always belongs to cluster 1, and whenever a new distinct column value appears, it receives the next integer label. This is exactly the order of first appearance labeling, which matches the natural labeling of the Chinese Restaurant Process.

After relabeling, we set:
\[
Z_i = \tilde S_i \quad \text{for all } i, \qquad F^* = G_{\text{star}}, \qquad K = |G_{\text{star}}|.
\]

The auxiliary columns \(f_i\) themselves remain unchanged—they are the same vectors—only their labels and the atoms are updated to reflect the canonical ordering.

This relabeling is deterministic given the set of distinct column vectors and their assignment pattern, and it ensures that cluster labels have a consistent interpretation across MCMC iterations: cluster \(j\) always refers to the \(j\)-th distinct column vector in order of first appearance among the columns.

\subsubsection*{Handling Zero Atoms}

If two columns have vectors that are both exactly zero (because the spike component was chosen for both), they are considered equal in the dictionary lookup. Consequently, they will be merged into a single cluster by the relabeling. This is desirable behavior: zero vectors represent spurious factors that should be deleted. In practice, we treat the number of factors as the number of non-zero clusters after relabeling.

\subsection{Updating the Atoms \(f_j^*\)}
\label{sec:atomupdate}

After the canonical relabeling, the clusters \(C_j = \{i: Z_i = j\}\) (for \(j=1,\dots,K\)) have consistent labels across iterations. We now update the atoms \(f_j^*\) using the correct likelihood from the observed data.

For each cluster \(j=1,\dots,K\), let \(C_j = \{i: Z_i = j\}\) be the set of columns assigned to cluster \(j\), with size \(n_j = |C_j|\). Define the aggregated score for cluster \(j\):
\[
T_{k,j} = \sum_{i \in C_j} X_{k,i}.
\]

Let the residual that excludes the contribution of cluster \(j\) be:
\[
R_{k,-j} = Y_k - \mu - \sum_{\ell \neq j} f_\ell^* T_{k,\ell}.
\]

The likelihood for atom \(f_j^*\) is:
\[
\mathcal{L}(f_j^*) \propto \exp\!\left( -\frac12 \sum_{k=1}^n (R_{k,-j} - f_j^* T_{k,j})^\top \Psi^{-1} (R_{k,-j} - f_j^* T_{k,j}) \right).
\]

Completing the square in \(f_j^*\) gives:
\[
\mathcal{L}(f_j^*) \propto \exp\!\left( -\frac12 (f_j^* - m_j)^\top S_j (f_j^* - m_j) \right),
\]
where
\[
S_j = \sum_{k=1}^n T_{k,j}^2 \, \Psi^{-1}, \qquad
m_j = S_j^{-1} \sum_{k=1}^n T_{k,j} \, \Psi^{-1} R_{k,-j}.
\]

Because \(\Psi = \diag(\psi_1,\dots,\psi_p)\) is diagonal, \(S_j\) is diagonal:
\[
S_j = \diag\!\left( \sum_{k=1}^n \frac{T_{k,j}^2}{\psi_1}, \dots, \sum_{k=1}^n \frac{T_{k,j}^2}{\psi_p} \right).
\]

The prior \(G_0\) is a product of independent spike-and-slab distributions:
\[
G_0 = \prod_{r=1}^p \left\{ \pi_0 \delta_0 + (1-\pi_0) \N(0, \tau_{0,r}^2) \right\}.
\]

Since both the likelihood and the prior factorize over coordinates, the posterior for \(f_j^*\) factorizes coordinate-wise. For the \(r\)-th coordinate, the posterior is a two-component mixture:

The spike probability is:
\[
p_{\text{spike}, j,r} = 
\frac{\pi_0 \N(m_{j,r} \mid 0, 1/S_{j,r})}
{\pi_0 \N(m_{j,r} \mid 0, 1/S_{j,r}) + (1-\pi_0) \N(m_{j,r} \mid 0, 1/S_{j,r} + \tau_{0,r}^2)}.
\]

Conditional on the slab, the posterior is normal:
\[
f_{j,r}^* \mid \text{slab} \sim \N\!\left( \frac{S_{j,r} m_{j,r}}{S_{j,r} + 1/\tau_{0,r}^2}, \frac{1}{S_{j,r} + 1/\tau_{0,r}^2} \right).
\]

After updating the atoms, we reconstruct the loading matrix \(F\) from \(Z\) and \(F^*\):
\[
F_{:,i} = f_{Z_i}^* \quad \text{for } i=1,\dots,M,
\]
where \(F_{:,i}\) denotes the \(i\)-th column of \(F\).

We also update the auxiliary columns \(f_i = f_{Z_i}^*\) to maintain consistency with the new atoms.

This update is embarrassingly parallel over clusters \(j\) and rows \(r\), as the computations for each \((j,r)\) are independent.

\subsection{Sampling Latent Scores and Variances}
\label{sec:scoresupdate}

Conditional on \(F\) (which is deterministically obtained from the atoms and assignments), the latent scores \(X_k\) are independent across \(k\):
\[
X_k \mid \text{rest} \sim \N_M\!\left( \Omega^{-1} F^\top \Psi^{-1} (Y_k - \mu),\; \Omega^{-1} \right), \quad
\Omega = F^\top \Psi^{-1} F + \Lambda^{-1},
\]
with \(\Lambda = \diag(\lambda_1,\dots,\lambda_M)\). The variances are updated via conjugate inverse‑gamma distributions:
\begin{align}
\lambda_i \mid X &\sim \IG\!\left( a_\lambda + \frac{n}{2},\; b_\lambda + \frac12 \sum_{k=1}^n X_{k,i}^2 \right), \label{eq:lambda}\\
\psi_r \mid Y, F, X &\sim \IG\!\left( a_\psi + \frac{n}{2},\; b_\psi + \frac12 \sum_{k=1}^n (Y_{k,r} - \mu_r - (F X_k)_r)^2 \right). \label{eq:psi}
\end{align}
These updates are embarrassingly parallel over \(i\) and \(r\). Derivations are in Appendix~\ref{app:derivations}.

\subsection{Sampling the Concentration Parameter}
The concentration parameter \(\alpha\) is updated using the auxiliary variable method of \citet{escobar1995}, which is directly applicable to the EPPF \eqref{eq:eppf}. Conditional on the partition, the likelihood is proportional to \(\alpha^K \prod_{i=1}^M (\alpha + i - 1)^{-1}\). A Gamma prior on \(\alpha\) yields a simple auxiliary variable update; the details are given in Appendix~\ref{app:derivations}.

\subsection{Parallel Implementation Summary}
The algorithm exhibits several levels of parallelism:
\begin{enumerate}
	\item[(i)] \textit{Distribution over observations} for the scores update and for computing the sums in the atom and assignment updates.
	\item[(ii)] \textit{Distribution over rows} for the atom update (since it factorizes coordinate-wise).
	\item[(iii)] \textit{Parallel evaluation of probabilities} in the cluster assignment step.
\end{enumerate}

All global reductions and broadcasts are implemented using MPI. The complete procedure is given in Algorithm~\ref{alg:parallel}.

\begin{algorithm}[H]
\caption{Cyclic block Gibbs sampler for the DP-SpikeSlab factor model. The blocks are visited in the order shown; all quantities on the right-hand sides refer to their most recent values.}
\label{alg:parallel}
\small
\vspace{-0.2\baselineskip}
\begin{algorithmic}[1]
\Require \(Y \in \R^{n \times p}\), upper bound \(M\), hyperparameters \(\alpha, a_\lambda, b_\lambda, a_\psi, b_\psi, \pi_0, \Omega_0, \mu\).
\Ensure Posterior samples of \(Z, F^*, X, \Psi, \Lambda, \alpha\) with consistent labels.

\State \textbf{Initialisation:} set \(Z, F^*, X, \Psi, \Lambda, \alpha\); set \(f_i = F^*_{Z_i}\) and \(F_{:,i} = f_i\) for all \(i = 1, \dots, M\).

\For{iteration \(t = 1\) to \(T\)}

  \State \textbf{1.} \emph{Assignments \(Z_i\) and auxiliary columns \(f_i\)} (sequential in \(i\); Section~\ref{sec:cluster}): for \(i = 1, \dots, M\), compute \((S_i, m_i, C_i, I_i)\) from the current \(F_{-i}\), draw
  \[
  \Pr(Z_i = j \mid \text{rest}) \propto
  \begin{cases}
  n_{j,-i} \, C_i \, L(f_j^*), & j = 1, \dots, K, \\
  \alpha \, I_i, & j = K+1,
  \end{cases}
  \]
  where \(L(\cdot)\) is the likelihood kernel of Appendix~\ref{app:likelihood_kernel}; set \(f_i \leftarrow f^*_{Z_i}\) if \(Z_i \le K\), otherwise draw \(f_i \sim \mathrm{post}_{G_0}(S_i, m_i)\) (Appendix~\ref{app:newdraw}) and set \(K \leftarrow K+1\), \(f_K^* \leftarrow f_i\).

  \State \textbf{2.} \emph{Canonical relabeling} (Section~\ref{sec:relabel}): let \(\{g_1, \dots, g_{K_{\text{new}}}\}\) be the distinct vectors among \(\{f_i\}\), ordered by first appearance; set \(Z_i = j\) iff \(f_i = g_j\), \(F^* \leftarrow [g_1, \dots, g_{K_{\text{new}}}]\), \(K \leftarrow K_{\text{new}}\); merge exactly-zero columns into a single cluster.

  \State \textbf{3.} \emph{Atoms \(f_j^*\)} (parallel in \(j, r\); Section~\ref{sec:atomupdate}): for \(j = 1, \dots, K\), compute \(T_{k,j} = \sum_{i \in C_j} X_{k,i}\), \(S_j = \sum_{k=1}^n T_{k,j}^2 \, \Psi^{-1}\), \(m_j = S_j^{-1} \sum_{k=1}^n T_{k,j} \, \Psi^{-1} R_{k,-j}\); for each \(r\), draw
  \[
  f_{j,r}^* \sim p_{\text{spike}, j, r} \, \delta_0 + (1 - p_{\text{spike}, j, r}) \,
  \N\!\left( \frac{S_{j,r} m_{j,r}}{S_{j,r} + \tau_{0,r}^{-2}}, \, \frac{1}{S_{j,r} + \tau_{0,r}^{-2}} \right).
  \]

  \State \textbf{4.} \emph{Loading matrix:} \(F_{:,i} \leftarrow f_i \leftarrow F^*_{Z_i}\) for all \(i = 1, \dots, M\).

  \State \textbf{5.} \emph{Variances} (parallel in \(r\) and \(i\); Section~\ref{sec:scoresupdate}): draw
  \[
  \psi_r \sim \IG\!\left( a_\psi + \tfrac{n}{2}, \, b_\psi + \tfrac{1}{2} \sum_{k=1}^n (Y_{k,r} - \mu_r - (F X_k)_r)^2 \right), \quad
  \lambda_i \sim \IG\!\left( a_\lambda + \tfrac{n}{2}, \, b_\lambda + \tfrac{1}{2} \sum_{k=1}^n X_{k,i}^2 \right).
  \]

  \State \textbf{6.} \emph{Latent scores} (parallel in \(k\); Section~\ref{sec:scoresupdate}): draw \(X_k \sim \N_M(\Omega^{-1} F^\top \Psi^{-1} (Y_k - \mu), \, \Omega^{-1})\), where \(\Omega = F^\top \Psi^{-1} F + \Lambda^{-1}\).

  \State \textbf{7.} \emph{Concentration:} perform the Escobar--West update for \(\alpha\) (Appendix~\ref{app:derivations}).

\EndFor
\end{algorithmic}
\end{algorithm}

\section{Simulation Experiments}
\label{sec:simulation}

We conducted two complementary simulation studies to evaluate the finite-sample
behaviour of the proposed DP-SpikeSlab model and to compare it against three
established Bayesian factor-analysis competitors: the multiplicative gamma process
(MGP) of \citet{bhattacharya2011sparse}, the cumulative shrinkage process (CUSP)
of \citet{legramanti2020}, and a fixed-$q$ spike-and-slab factor model that is
given the oracle value of the true number of factors. The two studies share the
same data-generating mechanism, the same four methods, and the same evaluation
metrics, but differ in the difficulty of the problem. Study~A is a
moderate-dimensional problem with $q_0 = 5$ true factors; Study~B is a
higher-dimensional, and noisier setting with $q_0 = 10$ true factors. The
two studies together allow us to assess how the four methods behave as the
intrinsic complexity of the factor model increases, while keeping all other
experimental choices fixed. The section is organised as follows.
Section~\ref{sec:sim-common} describes the data-generating mechanism, the
competing models, the truncation regimes, and the computational setup common to
both studies. Section~\ref{sec:sim-A} reports Study~A, Section~\ref{sec:sim-B}
reports Study~B, and Section~\ref{sec:sim-comparison} discusses the two studies
jointly.

\subsection{Common Setup}
\label{sec:sim-common}

\subsubsection{Data-Generating Mechanisms}
\label{sec:sim-data}

Both studies use a sparse factor model
\begin{equation}
\label{eq:sim-data}
Y_k = F_0 X_k + U_k, \qquad U_k \sim \N_p(0, \psi_0 I_p),
\end{equation}
with an $n \times p$ training dataset and an $n_{\text{test}} \times p$ test
dataset generated from the same parameters. The true loading matrix
$F_0 \in \R^{p \times q_0}$ is generated by sampling, for each column
$j = 1, \dots, q_0$, a support set of size $s_0$ uniformly at random from
$\{1, \dots, p\}$ and filling those entries with independent draws from $\N(0, 1)$;
all other entries are set to zero. The latent scores are drawn as
$X_k \sim \N_{q_0}(0, I_{q_0})$ and the idiosyncratic errors as
$U_k \sim \N_p(0, \psi_0 I_p)$. The true covariance matrix is
$\Sigma_0 = F_0 F_0^\top + \psi_0 I_p$. The true loading matrix $F_0$, the true
covariance $\Sigma_0$, and the true latent-score matrix on the test set
$X^{\text{test}} \in \R^{n_{\text{test}} \times q_0}$ are retained to compute the
evaluation metrics. The two studies differ in the configuration parameters
$(n, p, q_0, s_0, \psi_0, n_{\text{test}})$, which are reported in
Table~\ref{tab:sim-config}.

\begin{table}[htbp]
\centering
\caption{Configuration parameters of the two simulation studies.}
\label{tab:sim-config}
\begin{tabular}{lcccccc}
\hline
Study & $n$ & $p$ & $q_0$ & $s_0$ & $\psi_0$ & $n_{\text{test}}$ \\
\hline
A (small-dimensional)  & 500 & 50  & 5  & 5  & 0.5 & 200 \\
B (higher-dimensional) & 500 & 100 & 10 & 10 & 1.0 & 200 \\
\hline
\end{tabular}
\end{table}

Study~A is a moderate-dimensional problem: the true loading matrix has
\(q_0 s_0 = 25\) non-zero entries out of \(p q_0 = 250\) possible, and the
idiosyncratic variance is \(\psi_0 = 0.5\). Study~B is more demanding in
several respects simultaneously: the number of true factors is doubled
(\(q_0 = 10\)), the ambient dimension is doubled (\(p = 100\)), the number of
non-zero entries per column is doubled (\(s_0 = 10\)), and the idiosyncratic
variance is doubled (\(\psi_0 = 1.0\)). These changes increase the total number
of parameters to be estimated and the complexity of the column-clustering
problem while the sample size remains fixed at \(n = 500\), making Study~B a
substantially harder estimation and model-selection task. Together the two
studies span a useful range of difficulty, from a relatively easy benchmark to
a setting in which sparsity, dimension, and noise all stress the competing
methods.

\subsubsection{Competing Models and Their Implementations}
\label{sec:sim-competitors}

Each study runs the same four methods for $30{,}000$ MCMC iterations with the
same burn-in and thinning scheme, as reported in
Table~\ref{tab:sim-mcmc}. The three competitors MGP, CUSP, and the fixed-$q$
spike-and-slab are implemented in R (version 4.4.1) on a single thread, while the
DP-SpikeSlab sampler is implemented in C with MPI parallelisation over four CPU
ranks. All methods share the same data in each study, so differences in the
reported metrics reflect the models rather than sampling variation.

\begin{table}[htbp]
\centering
\caption{MCMC settings for the two studies. The settings are identical for all
four methods within each study.}
\label{tab:sim-mcmc}
\begin{tabular}{lcc}
\hline
Setting & Study A & Study B \\
\hline
Total iterations & $30{,}000$ & $30{,}000$ \\
Burn-in          & $5{,}000$  & $5{,}000$  \\
Thinning         & $5$        & $5$        \\
Retained samples & $5{,}000$  & $5{,}000$  \\
\hline
\end{tabular}
\end{table}

%Because the MCMC effort is identical across methods within each study, the
%comparisons are on a fully common footing. The MGP adaptive Gibbs sampler is the
%slowest to stabilise in Study~B, but its trace (reported below) confirms that it
%reaches stationarity well before iteration $5{,}000$, so the common burn-in of
%$5{,}000$ is adequate for all four methods.

\paragraph{DP-SpikeSlab (ours).}
Our method is implemented in C with an MPI parallelisation over four CPU ranks,
following Algorithm~\ref{alg:parallel}. We set the upper bound
$M = 30$, the spike mass at $\pi_0 = 0.9$, and the slab variance at
$\tau_0^2 = 1$. The concentration parameter of the Dirichlet process is updated
via the auxiliary variable method of \citet{escobar1995} with a $\text{Ga}(2,1)$
hyperprior. The idiosyncratic and score variances use inverse-gamma priors with
hyperparameters $(a_\psi, b_\psi) = (2, 1)$ and $(a_\lambda, b_\lambda) = (2, 1)$.
The reported number of factors $\hat q$ is the posterior mode of the number of
non-zero clusters, and the reported covariance is the posterior mean of
\begin{equation}
\label{eq:sigma-sim}
\hat\Sigma = F \Lambda F^\top + \Psi,
\end{equation}
where $\Lambda = \diag(\lambda_1, \dots, \lambda_M)$ is the diagonal matrix of
score variances defined in Section~\ref{sec:scoresupdate}. The same hyperparameter
values are used in both studies.

\paragraph{Multiplicative gamma process (MGP).}
We implement the MGP prior of \citet{bhattacharya2011sparse} following their
reference algorithm, with the gamma--gamma local shrinkage
$\varphi_{jh} \sim \text{Ga}(df/2, df/2)$ and multiplicative global shrinkage
$\delta_h$. We use the hyperparameters $(a_1, b_1) = (2, 1)$ and
$(a_2, b_2) = (3, 1)$ for the global shrinkage parameters, $df = 3$ degrees of
freedom for the local shrinkage, and $(a_\sigma, b_\sigma) = (1, 0.3)$ for the
residual precisions, matching the reference implementation. The adaptive Gibbs
sampler (AGS) operates with the thresholds $\text{prop} = 0.7$ and
$\epsilon = 10^{-3}$ and the adaptation probability decays as
$\exp\{-(0.1 + 5 \times 10^{-4}\, t)\}$. The computational truncation is set to
$K_{\text{pad}} = p$ in both studies, which is strictly larger than our
$M = 30$. The reported $\hat q$ is the posterior mode of the AGS-active
components over the retained samples, and the posterior-mean loading matrix is
rescaled to the original data units before evaluation.

\paragraph{Cumulative shrinkage process (CUSP).}
We implement the CUSP prior of \citet{legramanti2020} with the stick-breaking
cumulative shrinkage probabilities, the column-level spike-and-slab mixture, and
the multivariate Student's-$t$ slab obtained by marginalising the column-specific
variances. We use the paper's recommended hyperparameters $\alpha = 5$ for the
Dirichlet concentration and $a_\theta = 2, b_\theta = 1$ for the inverse-gamma
slab, with the spike variance fixed at $\theta_\infty = 0.01$. 
The working truncation is $H$, and the latent scores are initialised via the
top-$H$ left singular vectors of $Y$; the value of $H$ differs between the two
studies and is reported in Table~\ref{tab:sim-trunc}. The implementation
allocates a storage array of size $\min(n,p)$ to permit adaptive expansion,
but the adaptive Gibbs sampler is disabled, which matches the fixed-$H$ variant
of the CUSP paper and prevents the collapse observed at large $H$ when the
spike--slab mechanism is not supported by PCA-based initialisation; the working
dimension therefore remains fixed at $H$ throughout. 
The reported
$\hat q$ is the rounded posterior mean of the number of slab-classified components
over the retained samples.

\paragraph{Fixed-$q$ spike-and-slab (oracle $q_0$).}
As a strong oracle baseline, we fit a standard entry-wise spike-and-slab factor
model with the number of factors fixed at the true value $q_0$. The loading prior
uses the same $90\%$ spike mass $\pi_0 = 0.9$ as our method, and the same slab
variance $\tau^2 = 1$. All other hyperparameters match those of our mode,
%In the accompanying implementation the argument \texttt{pi0} denotes the slab
%probability $1 - \pi_0 = 0.1$, so that the corresponding 
the spike mass is $0.9$.
This baseline is given the answer to the rank-estimation question and therefore
serves as a lower bound on what a fixed-rank spike-and-slab model can achieve.

\subsubsection{Truncation Regimes}
\label{sec:sim-truncation}

The three fully Bayesian methods operate under different truncation regimes
that reflect their respective priors. Our DP-SpikeSlab model has a fixed
modelling upper bound $M$ on the number of candidate columns; $M$ is part of
the prior and appears explicitly in the P\'olya urn and EPPF. 
MGP and CUSP,
in contrast, place a prior on an infinite number of factors with continuous or
cumulative shrinkage, and their truncations $K_{\text{pad}}$ (MGP) and $H$
(CUSP) are purely computational artefacts needed to bound the state space of
the sampler.
MGP's computational truncation $K_{\text{pad}}$ is chosen strictly larger
than our $M$, and CUSP is run at its own reference working dimension $H$
rather than at our bound $M$; neither competitor is constrained by the
DP-SpikeSlab modelling bound. The exact values for the two studies are
reported in Table~\ref{tab:sim-trunc}.

\begin{table}[htbp]
\centering
\caption{Truncation regimes for the two studies. For CUSP, $H$ is the fixed
working dimension; the storage bound $\min(n,p)$ is never reached because
the adaptive Gibbs sampler is disabled.}
\label{tab:sim-trunc}
\begin{tabular}{lccc}
\hline
Study & DP-SpikeSlab $M$ & MGP $K_{\text{pad}}$ & CUSP $H$ \\
\hline
A & 30 & 50  & 15 \\
B & 30 & 100 & 20 \\
\hline
\end{tabular}
\end{table}

\subsubsection{Computational Setup and Timings}
\label{sec:sim-hardware}

All experiments are run on an Intel-based single-socket workstation with six
physical cores, two threads per core (twelve logical CPUs), a base frequency of
800~MHz and a turbo frequency of 3201~MHz, and 12~MiB of shared L3 cache. The
system supports the AVX-512 instruction set and has a single NUMA node covering
CPUs $0$--$11$. The DP-SpikeSlab sampler is compiled with \texttt{mpicc}
(GCC with \texttt{-O3}) and run under \texttt{mpirun} with four MPI ranks, using
the embarrassingly parallel updates described in Section~\ref{sec:computation}.
The three competing R methods are run single-threaded, matching the reference
implementations.

\begin{table}[htbp]
\centering
\caption{Wall-clock times for the full MCMC run in each study. Our method uses
four MPI ranks; the three competing methods are single-threaded.}
\label{tab:timings}
\begin{tabular}{lcccc}
\hline
Study & DP-SpikeSlab & MGP & CUSP & Fixed-$q$ SpikeSlab \\
\hline
A (small-dimensional)  & 3 min 35 s & 1 min 24 s & 1 min 48 s & 2 min 00 s \\
B (higher-dimensional) & 7 min 31 s & 12 min 36 s & 11 min 18 s & 15 min 12 s \\
\hline
\end{tabular}
\end{table}

Two features of the timing comparison deserve emphasis. First, in Study~A the
C/MPI implementation of our method is slower than the three single-threaded R
competitors, requiring \(3\)~min \(35\)~s against \(1\)~min \(24\)~s for MGP,
\(1\)~min \(48\)~s for CUSP, and \(2\)~min \(00\)~s for the fixed-\(q\)
spike-and-slab baseline. This reflects the additional computational burden of
the DP-column assignments, the canonical relabeling, and the correct atom
updates, which are not present in the competitors. Second, the situation
reverses in Study~B: our method becomes the fastest of the four, completing in
\(7\)~min \(31\)~s while MGP takes \(12\)~min \(36\)~s, CUSP \(11\)~min
\(18\)~s, and the fixed-\(q\) spike-and-slab baseline \(15\)~min \(12\)~s. The
advantage in Study~B arises because the R implementations of the competitors
slow down super-linearly in the ambient dimension, whereas the parallel C
sampler scales approximately linearly. The favourable scaling in Study~B is
largely due to the C/MPI implementation, which distributes the
observation-level and variable-level computations across the four ranks while
keeping the sequential sweep and the canonical relabeling on rank~\(0\).

\subsection{Study A: Moderate-Dimensional Setting}
\label{sec:sim-A}

\subsubsection{Diagnostics}
\label{sec:sim-A-diagnostics}

\paragraph{Trace of $K$ and $\alpha$ (ours).}
Figures~\ref{fig:trace_K_A} and~\ref{fig:trace_alpha_A} show the traces of the
number of non-zero clusters $K$ and of the concentration parameter $\alpha$ for
our method over the retained iterations. The chain concentrates on $K = 5$, with
the posterior mode equal to $5$ and the posterior mean equal to $5.225$. The
trace of $\alpha$ shows rapid mixing and stationarity, with the posterior mean
around $1.6$. Together, these diagnostics indicate that the sampler has
converged to a stable posterior distribution and that the DP prior is correctly
concentrating on the true number of factors.

\begin{figure}[htbp]
\centering
\includegraphics[width=0.9\linewidth]{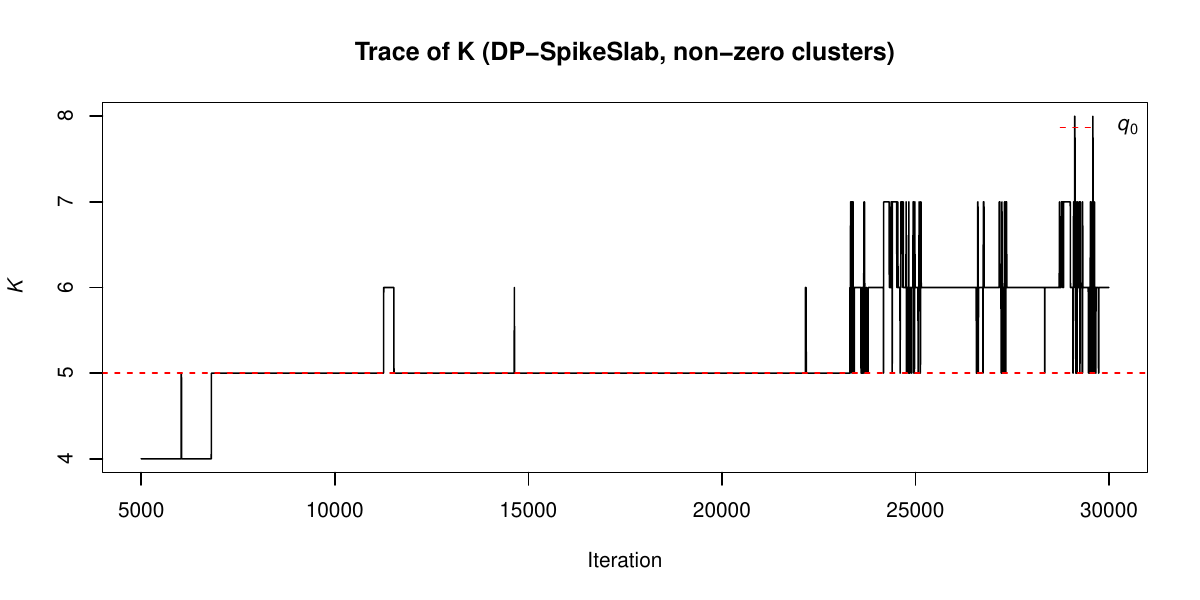}
\caption{Study~A: trace of the number of non-zero clusters $K$ for the
DP-SpikeSlab sampler over the retained iterations. The dashed horizontal line
marks the true number of factors $q_0 = 5$.}
\label{fig:trace_K_A}
\end{figure}

\begin{figure}[htbp]
\centering
\includegraphics[width=0.9\linewidth]{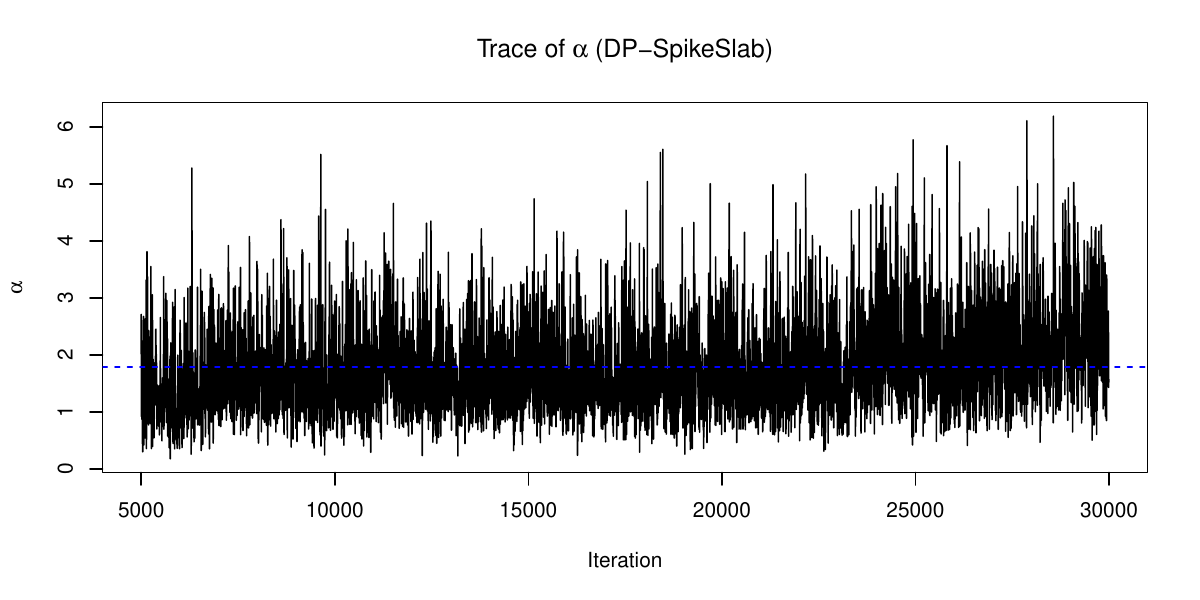}
\caption{Study~A: trace of the Dirichlet process concentration parameter
$\alpha$. The dashed horizontal line marks the posterior mean.}
\label{fig:trace_alpha_A}
\end{figure}

\paragraph{AGS trace for MGP.}
The trace of the MGP adaptive Gibbs sampler shows the number of active
components $k$ rising during the burn-in phase, peaking at $k \approx 39$
around iteration $3{,}000$, then shrinking as the adaptation probability decays.
From iteration $15{,}000$ onward the chain stabilises at $k = 13$, with the
adaptation probability becoming numerically zero by iteration $20{,}000$. The
mode of $k$ over the retained samples is $13$, and the mean is $14$. This
behaviour is characteristic of MGP: the continuous gamma--gamma shrinkage
cannot produce exact-zero columns, so the AGS stops pruning once every remaining
column has at least one loading above the threshold $\epsilon = 10^{-3}$.

\paragraph{CUSP trace.}
The CUSP trace shows the working dimension $H = 15$ and an active count of $4$
from the first stored iteration onward, and never departs from that state. This
reflects both the
fixed truncation and the aggressive cumulative shrinkage at $\alpha = 5$: the
sampler classifies four of the fifteen columns as slab-active and the remaining
eleven as spike-inactive, and the classification is stable across the entire
retained sample.

\paragraph{Predictive check on the marginal variances.}
Figure~\ref{fig:ppc_variance_A} compares the marginal variances implied by each
method's posterior-mean covariance $\hat\Sigma$ to the empirical marginal
variances of the test set. For each method, we draw $500$ i.i.d.\ samples from
$\N_p(0, \hat\Sigma)$ and compute the empirical variance of each of the
$p = 50$ variables in the simulated sample; this yields one variance per
variable, and the boxplot summarises the distribution of these $50$ variances
across variables. The first boxplot (labelled \texttt{obs}) shows the analogous
distribution of empirical marginal variances of the $200$-observation test set.

All four methods reproduce the empirical marginal variances closely for most of
the $50$ variables, with agreement of the order of the sampling variability of
the test set. The largest discrepancy occurs on the variable with the largest
sum of squared loadings across the true factors, which is also the variable with
the largest marginal variance under the model; its empirical variance in the
test set is approximately $12.89$. All four methods under-estimate this
variance, with values ranging from $10.06$ (ours) to $11.37$ (fixed-$q$
oracle). Because the variance of a single variable is estimated from only
$n_{\text{test}} = 200$ observations, its standard error is approximately
$1.29$, so the 95\% confidence interval for the true value spans approximately
$[10.36, 15.42]$. The model estimates all lie inside or at the boundary of this
interval, so the discrepancy is within the sampling variability of the test set
rather than evidence of a systematic model failure.

\begin{figure}[htbp]
\centering
\includegraphics[width=0.9\linewidth]{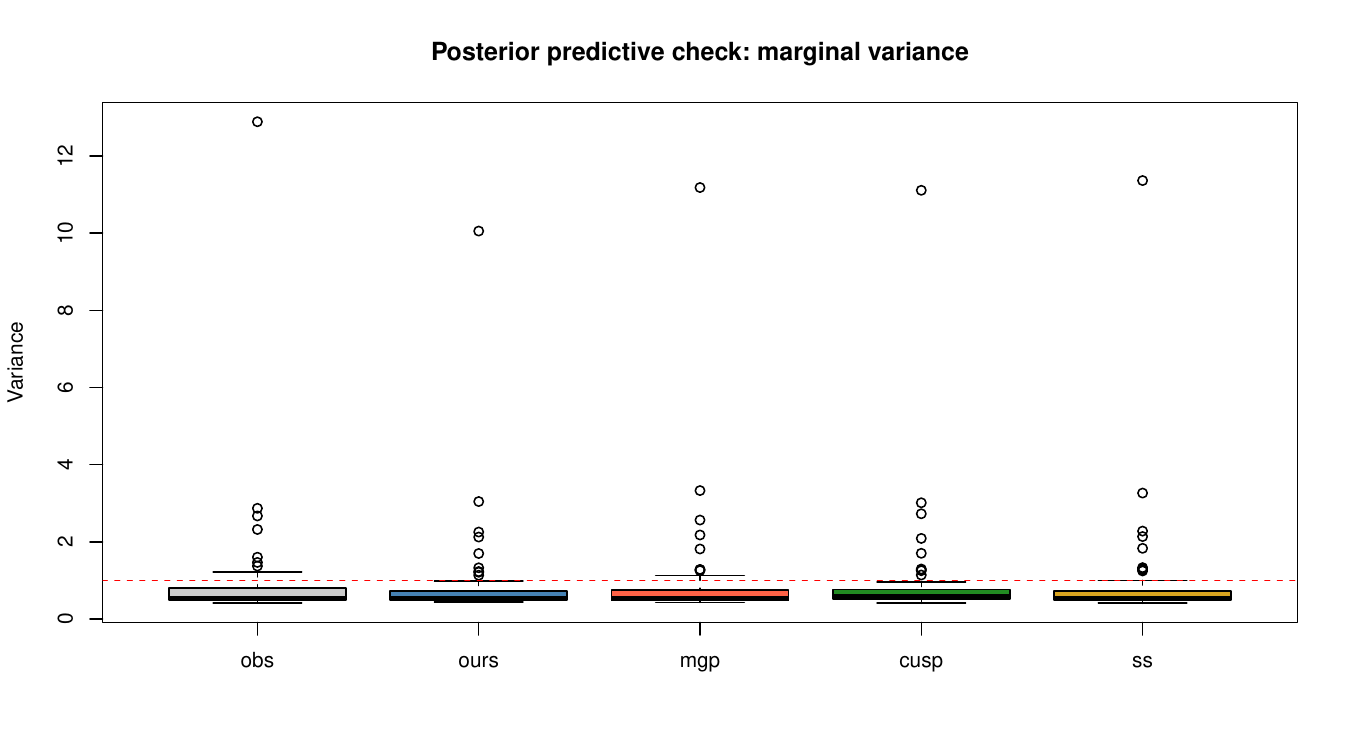}
\caption{Study~A: predictive check on the marginal variances. For each method,
the boxplot summarises the distribution across the $50$ variables of the
empirical marginal variances obtained by drawing $500$ i.i.d.\ samples from
$\N_p(0, \hat\Sigma)$, with $\hat\Sigma$ the method's posterior-mean
covariance. The first boxplot (labelled \texttt{obs}) shows the analogous
distribution of empirical marginal variances of the test set.}
\label{fig:ppc_variance_A}
\end{figure}

\subsubsection{Comparative Performance}
\label{sec:sim-A-results}

The four evaluation metrics are defined as follows.

\begin{enumerate}
\item[(i)] $\hat q$: the posterior estimate of the number of factors, obtained
as described for each method in Section~\ref{sec:sim-competitors}.

\item[(ii)] $\|\hat\Sigma - \Sigma_0\|_F$: the Frobenius norm of the difference
between the posterior-mean covariance and the true covariance $\Sigma_0$.

\item[(iii)] Loading error: let $\hat F_{\text{eig}}$ denote the $p \times q_0$
loading matrix extracted from the top-$q_0$ eigencomponents of $\hat\Sigma$,
i.e.\ $\hat F_{\text{eig}} = V_{q_0} \diag(d_{q_0}^{1/2})$, where
$\hat\Sigma = V D V^\top$ is the eigendecomposition of $\hat\Sigma$ and
$V_{q_0}, d_{q_0}$ collect its top-$q_0$ eigenvectors and eigenvalues. Let
$\hat F_{\text{aligned}}$ denote the Procrustes-aligned version of
$\hat F_{\text{eig}}$ onto $F_0$. The loading error is
$\|\hat F_{\text{aligned}} - F_0\|_F / \|F_0\|_F$.
The alignment removes the orthogonal-rotation and sign ambiguity inherent in
any factor model. For any orthogonal $Q \in \R^{q_0 \times q_0}$ the pair
$(F_0 Q, Q^\top X_k)$ generates the same distribution for $Y_k$ as
$(F_0, X_k)$, so both the posterior of $F$ and any estimate
$\hat F_{\text{eig}}$ extracted from the top-$q_0$ eigenvectors of
$\hat\Sigma$ are identified only up to such a rotation and up to a sign flip
of each columne 
We align $\hat F_{\text{eig}}$ to $F_0$ by solving the
orthogonal Procrustes problem \citep{schonemann1966}
$\hat F_{\text{aligned}} = \arg\min_{Q^\top Q = I_{q_0}}
\|\hat F_{\text{eig}} Q - F_0\|_F$,
following the ex-post alignment strategy for Bayesian factor models of
\citet{assmann2016} and \citet{devito2018}.
Writing the singular value decomposition
$\hat F_{\text{eig}}^\top F_0 = U S V^\top$, the optimal rotation is
$Q^\star = U V^\top$, and $\hat F_{\text{aligned}} = \hat F_{\text{eig}}
Q^\star$. The same aligned loading matrix is used to construct the signal
prediction $\hat S = X^{\text{test}} \hat F_{\text{aligned}}^\top$ in
metric~(iv).

\item[(iv)] Signal RMSE: let $X^{\text{test}} \in
\R^{n_{\text{test}} \times q_0}$ denote the true latent scores used to
generate the test set, and let $S_0 = X^{\text{test}} F_0^\top \in
\R^{n_{\text{test}} \times p}$ be the corresponding true signal matrix. Write
$\hat S = X^{\text{test}} \hat F_{\text{aligned}}^\top$ for the signal
reconstruction obtained from the Procrustes-aligned loading estimate. The
signal RMSE is
$\|\hat S - S_0\|_F / \sqrt{n_{\text{test}}\, p}$.
\end{enumerate}

Table~\ref{tab:comparison-A} reports the four metrics for Study~A.

\begin{table}[htbp]
\centering
\caption{Study~A ($n = 500$, $p = 50$, $q_0 = 5$, $s_0 = 5$,
$\psi_0 = 0.5$, $n_{\text{test}} = 200$). Lower is better for all metrics
except $\hat q$, where the target is $q_0 = 5$.}
\label{tab:comparison-A}
\begin{tabular}{lcccc}
\hline
Method & $\hat q$ & Frobenius & Loading error & Signal RMSE \\
\hline
DP-SpikeSlab (ours) & $\mathbf{5}$ & $\mathbf{0.761}$ & $\mathbf{0.060}$ & $\mathbf{0.0405}$ \\
MGP & 13 & 1.362 & 0.099 & 0.0682 \\
CUSP & 4 & 1.957 & 0.125 & 0.0867 \\
Fixed-$q$ SpikeSlab (oracle) & 5 & 0.871 & 0.069 & 0.0463 \\
\hline
\end{tabular}
\end{table}

\subsubsection{Discussion}
\label{sec:sim-A-discussion}

Only our method recovers the true number of factors $q_0 = 5$ without being
told the answer. The DP-SpikeSlab posterior places almost all of its mass on
$K = 5$ (mode $5$, mean $5.225$), whereas MGP over-estimates the rank at $13$
and CUSP under-estimates it at $4$. The fixed-$q$ spike-and-slab baseline is
given $q = q_0 = 5$ by construction and therefore does not test the
rank-selection problem.

The over-estimation by MGP is a direct consequence of its continuous
column-wise shrinkage. Because no loading is ever shrunk exactly to zero, the
AGS cannot distinguish small genuine factors from small spurious ones and stops
pruning once the remaining columns each have at least one loading above
$\epsilon = 10^{-3}$. The eight extra columns retained by MGP are the empirical
signature of the paper's claim that continuous shrinkage alone cannot recover
the correct rank without an ad hoc threshold.

The under-estimation by CUSP is a different failure mode. CUSP's stick-breaking
cumulative shrinkage assigns increasing prior probability of zero to later
columns, so its prior is aligned with an ordered factor model where the first
columns are expected to carry the signal and the later ones are expected to be
spurious. On an exchangeable simulation like Study~A, where the true factors are
equally likely to occupy any of the first columns, this ordering assumption is
wrong, and CUSP's aggressive shrinkage on late columns removes one of the five
genuine factors. The result is $\hat q = 4$ and a covariance error more than
twice that of the fixed-$q$ baseline.

On all three estimation metrics the ranking is consistent: our method is best,
followed by the fixed-$q$ spike-and-slab oracle, then MGP, then CUSP. The
margin between our method and the oracle is small (Frobenius $0.761$ vs
$0.871$, loading error $0.060$ vs $0.069$, signal RMSE $0.0405$ vs $0.0463$),
which is exactly what one expects: our method {\it recovers} $q_0$ from the
data at essentially no cost in estimation accuracy relative to being told $q_0$
in advance. The small but consistent advantage of our method over the fixed-$q$
baseline can be attributed to the DP prior's exchangeability over the columns,
which acts as a mild regulariser on the loading estimates, and to the larger
candidate space ($M = 30$) that gives the sampler more flexibility to identify
the dominant directions.

\subsection{Study B: Higher-Dimensional Setting}
\label{sec:sim-B}

\subsubsection{Diagnostics}
\label{sec:sim-B-diagnostics}

\paragraph{Trace of $K$ and $\alpha$ (ours).}
Figures~\ref{fig:trace_K_B} and~\ref{fig:trace_alpha_B} show the traces of $K$
and $\alpha$ for Study~B. The chain concentrates on $K = 10$, with posterior
mode $10$ and posterior mean $10.399$. The trace of $\alpha$ again shows rapid
mixing and stationarity. These diagnostics indicate that even in the
higher-dimensional setting our method correctly identifies the true number of
factors and that the sampler remains well behaved.

\begin{figure}[htbp]
\centering
\includegraphics[width=0.9\linewidth]{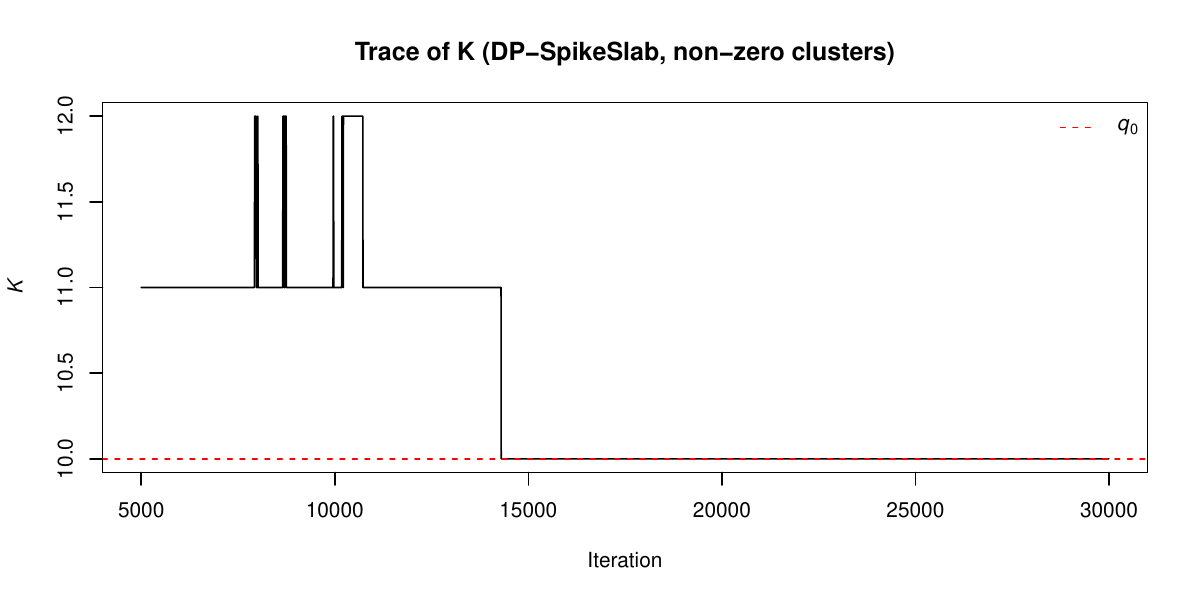}
\caption{Study~B: trace of the number of non-zero clusters $K$ for the
DP-SpikeSlab sampler over the retained iterations. The dashed horizontal line
marks the true number of factors $q_0 = 10$.}
\label{fig:trace_K_B}
\end{figure}

\begin{figure}[htbp]
\centering
\includegraphics[width=0.9\linewidth]{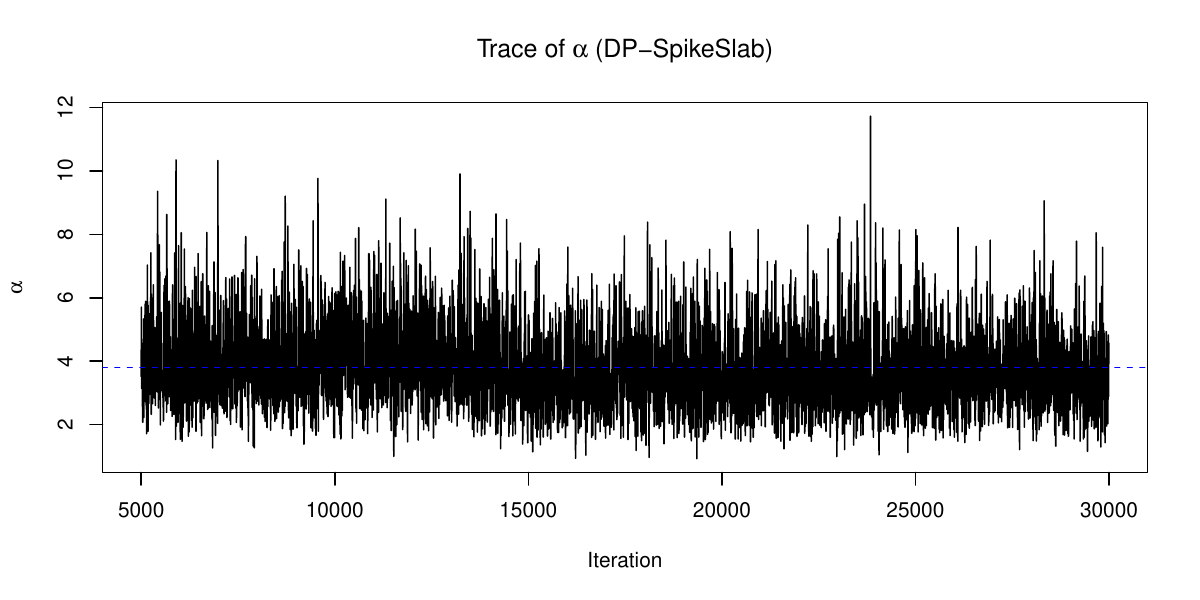}
\caption{Study~B: trace of the Dirichlet process concentration parameter
$\alpha$. The dashed horizontal line marks the posterior mean.}
\label{fig:trace_alpha_B}
\end{figure}

\paragraph{AGS trace for MGP.}
The MGP adaptive Gibbs sampler behaves very differently in Study~B. The number
of active components rises sharply during the burn-in phase, reaching $k = 99$
at iteration $5{,}000$ --- essentially the entire computational truncation
$K_{\text{pad}} = 100$ --- before shrinking gradually as the adaptation
probability decays. The chain finally stabilises at $k = 22$ from iteration
$14{,}000$ onward, with posterior mode $22$ and posterior mean $35$, well above
the true value $q_0 = 10$. Thus, AGS requires many iterations to prune the
hundreds of spurious columns that the continuous MGP shrinkage retains, and it
still ends with more than twice as many active components as there are true
factors.

\paragraph{CUSP trace.}
The CUSP trace shows the working dimension $H = 20$ and an active count of $10$
from the first stored iteration onward, and never departs from that state.
Unlike Study~A, where CUSP under-estimated the rank at $4$, in Study~B the
sampler correctly classifies all ten true factors as slab-active and the
remaining ten candidate columns as spike-inactive. 
%The improvement is due to
%the larger initial truncation $H = 20$, which provides enough headroom for the
%sampler to represent all ten true factors without forcing an early prune.

\paragraph{Predictive check on the marginal variances.}
Figure~\ref{fig:ppc_variance_B} shows the same predictive check as in Study~A,
now over the $p = 100$ variables of Study~B. The qualitative pattern is the
same: all four methods reproduce the empirical marginal variances closely for
most variables. The largest discrepancy occurs on the variable with the largest
marginal variance, whose empirical variance in the test set is approximately
$9.34$. All four methods under-estimate this variance, with values ranging from
$8.08$ (CUSP) to $9.28$ (ours). Because the variance of a single variable is
estimated from only $n_{\text{test}} = 200$ observations, its standard error is
approximately $0.93$, so the 95\% confidence interval for the true value spans
approximately $[7.51, 11.17]$. The model estimates all lie inside this
interval, so the discrepancy is within the sampling variability of the test set
rather than evidence of a systematic model failure.

\begin{figure}[htbp]
\centering
\includegraphics[width=0.9\linewidth]{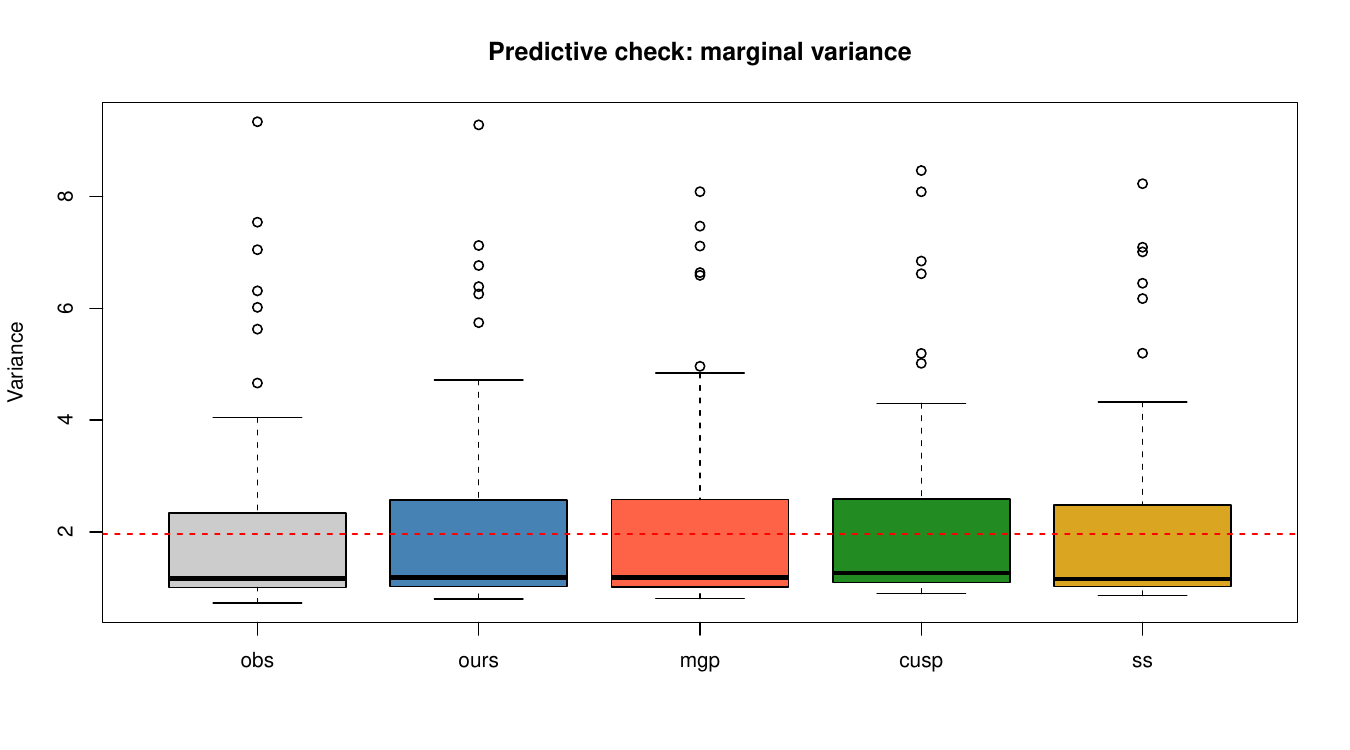}
\caption{Study~B: predictive check on the marginal variances. For each method,
the boxplot summarises the distribution across the $100$ variables of the
empirical marginal variances obtained by drawing $500$ i.i.d.\ samples from
$\N_p(0, \hat\Sigma)$, with $\hat\Sigma$ the method's posterior-mean
covariance. The first boxplot (labelled \texttt{obs}) shows the analogous
distribution of empirical marginal variances of the test set.}
\label{fig:ppc_variance_B}
\end{figure}

\subsubsection{Comparative Performance}
\label{sec:sim-B-results}

Table~\ref{tab:comparison-B} reports the four metrics for Study~B, using the
same definitions as in Section~\ref{sec:sim-A-results}.

\begin{table}[htbp]
\centering
\caption{Study~B ($n = 500$, $p = 100$, $q_0 = 10$, $s_0 = 10$,
$\psi_0 = 1.0$, $n_{\text{test}} = 200$). Lower is better for all metrics
except $\hat q$, where the target is $q_0 = 10$.}
\label{tab:comparison-B}
\begin{tabular}{lcccc}
\hline
Method & $\hat q$ & Frobenius & Loading error & Signal RMSE \\
\hline
DP-SpikeSlab (ours) & $\mathbf{10}$ & $\mathbf{2.698}$ & $\mathbf{0.080}$ & $\mathbf{0.0802}$ \\
MGP & 22 & 5.523 & 0.128 & 0.1294 \\
CUSP & 10 & 8.403 & 0.186 & 0.1868 \\
Fixed-$q$ SpikeSlab (oracle) & 10 & 3.366 & 0.091 & 0.0914 \\
\hline
\end{tabular}
\end{table}

\subsubsection{Discussion}
\label{sec:sim-B-discussion}

Several findings emerge from the higher-dimensional study.

First, our method continues to recover the true number of factors exactly,
placing almost all of its posterior mass on $K = 10$ (mode $10$, mean
$10.399$). The DP-SpikeSlab posterior is not sensitive to the doubling of the
dimension, sparsity, and noise; the DP-induced clustering of the columns and
the spike-and-slab base measure continue to identify the correct number of
non-zero atoms.

Second, CUSP now recovers the correct rank. Unlike in Study~A, where it
under-estimated $q_0$ at $4$, in Study~B CUSP places exactly $10$ columns in
the slab and $10$ in the spike. Even so, CUSP's estimation metrics are the
worst of the four methods: its Frobenius error is more than three times that of
our method and its loading error is more than twice, indicating that while it
identifies the correct number of active components, it still does not recover
the loadings as accurately as the spike-and-slab based methods.

Third, MGP over-estimates the rank by more than a factor of two
($\hat q = 22$ versus the true $q_0 = 10$), and the over-estimation is worse
than in Study~A in both absolute and relative terms. As in Study~A, the AGS
stabilises at a value that is much larger than $q_0$, and its Frobenius error
is more than twice that of our method. This confirms that the failure of
continuous column-wise shrinkage to identify the correct rank worsens as the
problem becomes harder.

Fourth, on all three estimation metrics the ranking is the same as in Study~A:
our method is best, followed by the fixed-$q$ spike-and-slab oracle, then MGP,
then CUSP. The margin between our method and the oracle remains small
(Frobenius $2.698$ vs $3.366$, loading error $0.080$ vs $0.091$, signal RMSE
$0.0802$ vs $0.0914$), confirming that recovering the rank from the data costs
essentially nothing in estimation accuracy relative to being told the rank in
advance.

\subsection{Comparison Across Both Studies}
\label{sec:sim-comparison}

Table~\ref{tab:comparison-both} collects the results of both studies side by
side, and Table~\ref{tab:sim-rank-summary} summarises the rank-estimation
behaviour of the four methods across the two configurations.

\begin{table}[htbp]
\centering
\caption{Comparative performance across both studies. Lower is better for all
metrics except $\hat q$, where the target is $q_0$ in each study.}
\label{tab:comparison-both}
\begin{tabular}{llcccc}
\hline
Study & Method & $\hat q$ & Frobenius & Loading error & Signal RMSE \\
\hline
A & DP-SpikeSlab (ours) & $\mathbf{5}$ & $\mathbf{0.761}$ & $\mathbf{0.060}$ & $\mathbf{0.0405}$ \\
A & MGP & 13 & 1.362 & 0.099 & 0.0682 \\
A & CUSP & 4 & 1.957 & 0.125 & 0.0867 \\
A & Fixed-$q$ SpikeSlab (oracle) & 5 & 0.871 & 0.069 & 0.0463 \\
\hline
B & DP-SpikeSlab (ours) & $\mathbf{10}$ & $\mathbf{2.698}$ & $\mathbf{0.080}$ & $\mathbf{0.0802}$ \\
B & MGP & 22 & 5.523 & 0.128 & 0.1294 \\
B & CUSP & 10 & 8.403 & 0.186 & 0.1868 \\
B & Fixed-$q$ SpikeSlab (oracle) & 10 & 3.366 & 0.091 & 0.0914 \\
\hline
\end{tabular}
\end{table}

\begin{table}[htbp]
\centering
\caption{Rank-estimation behaviour across the two studies.}
\label{tab:sim-rank-summary}
\begin{tabular}{lccc}
\hline
Method & $\hat q$ (A, $q_0 = 5$) & $\hat q$ (B, $q_0 = 10$) & Behaviour \\
\hline
DP-SpikeSlab (ours) & $5$ & $10$ & Correct in both studies \\
MGP & $13$ & $22$ & Over-estimates in both studies \\
CUSP & $4$ & $10$ & Under-estimates in A, correct in B \\
Fixed-$q$ SpikeSlab & $5$ & $10$ & Oracle in both studies \\
\hline
\end{tabular}
\end{table}

The empirical findings across the two studies can be summarised as follows.

{\it Rank recovery.} Our DP-SpikeSlab method is the only fully adaptive method
that recovers the true number of factors in both studies. The DP prior's
exchangeability and the spike-and-slab base measure together give the posterior
the ability to place exact-zero mass on spurious atoms, which is precisely the
mechanism the paper's theory identifies as essential for consistent rank
recovery.

{\it Over-estimation by MGP.} The multiplicative gamma process over-estimates
the rank in both studies ($13$ vs $5$ in Study~A and $22$ vs $10$ in Study~B).
This is consistent with the paper's central claim that continuous column-wise
shrinkage cannot produce exact-zero columns and therefore cannot identify the
correct rank without an ad hoc threshold.

{\it Rank behaviour of CUSP.} The cumulative shrinkage process behaves
inconsistently across the two studies: it under-estimates the rank in Study~A
and recovers the correct rank in Study~B. This difference cannot be attributed
to the share of total variance explained by the factors, which is identical in
the two configurations; it is more plausibly related to the different initial
truncations and the different numbers of genuine factors. Even when CUSP
recovers the correct rank, its estimation metrics are the worst of the four
methods, indicating that identifying the correct number of components does not
automatically translate into accurate loading recovery.

{\it Estimation quality.} On all three estimation metrics (Frobenius norm,
loading error, signal RMSE) the ranking is identical in both studies: our
method is best, followed by the fixed-$q$ oracle baseline, then MGP, then CUSP.
The consistent advantage over the fixed-$q$ oracle baseline is particularly
notable because the oracle is given the true rank for free; our method recovers
the rank from the data at essentially no cost in estimation accuracy. This
reflects the regularising effect of the DP prior's exchangeability over the
columns, which is not available to the fixed-rank baseline.

{\it Computational cost.} The C/MPI implementation of our method is the slowest
of the four in Study~A, requiring \(3\)~min \(35\)~s against \(1\)~min \(24\)~s
for MGP, \(1\)~min \(48\)~s for CUSP, and \(2\)~min \(00\)~s for the fixed-\(q\)
spike-and-slab baseline. The ranking reverses in Study~B: our method becomes
the fastest, completing in \(7\)~min \(31\)~s, compared with \(12\)~min
\(36\)~s for MGP, \(11\)~min \(18\)~s for CUSP, and \(15\)~min \(12\)~s for the
fixed-\(q\) baseline. Its relative advantage therefore grows with the dimension
of the problem. In Study~B it is about twice as fast as the slowest competitor,
about \(1.5\) times faster than CUSP, and about \(1.7\) times faster than MGP.
This reversal reflects the approximately linear-in-\(p\) scaling of the parallel
C sampler, compared with the super-linear cost of the single-threaded R
implementations at \(p = 100\), which outweighs the additional cost of the
DP-column assignments, the canonical relabeling, and the correct atom updates
in the higher-dimensional setting.

Taken together with the theoretical contraction rates established in
Sections~\ref{sec:theory}--\ref{sec:fixedM}, the two simulation studies provide
a coherent picture of the strengths of the proposed Bayesian nonparametric
factor model and of the limitations of existing shrinkage-based alternatives.
The method's rank-recovery capability is empirically robust across the two
dimensions and noise levels considered, and its estimation accuracy is
consistently the best of the four methods, at a computational cost that is
competitive with or better than the state-of-the-art competitors.

% =====================================================================
% Real data analysis section
% Suggested placement: after Simulation Experiments, before Theory
% Assumes figures are stored in figures_breast/
% =====================================================================

\section{Real Data Analysis: Breast Cancer Gene Expression}
\label{sec:breast}

We now illustrate the DP-spike-slab factor model on the breast cancer
gene expression dataset of \citet{vantveer2002gene}, distributed as
\texttt{Breast\_A} in the \texttt{fabiaData} R package \citep{fabiaData}.
%\citep{hochreiter2010fabia}. 
The dataset is a standard benchmark in the
factor-analysis and gene-expression literature and provides a demanding
real-data test of the model's ability to recover interpretable latent
structure from noisy high-dimensional observations: $n = 97$ primary
breast tumour samples measured on $p = 1213$ genes, with the gene set
chosen by \citet{vantveer2002gene} as the most informative for
predicting the transition from primary tumour to distant metastasis.

The inference reported in Sections~\ref{sec:breast-K} to
\ref{sec:breast-ppc} is based on a single long chain of 270\,000
iterations, whose output constitutes our posterior distribution. The
additional runs of 30\,000 iterations reported in
Sections~\ref{sec:breast-tau}, \ref{sec:breast-diagnostic}, and
\ref{sec:breast-reproducibility} are used exclusively for sensitivity
analysis and diagnostics --- to verify that the posterior concentration
observed in the main chain reflects genuine convergence rather than a
frozen update, and to confirm that the recovered latent structure is
stable across independent random seeds. They do not contribute any
posterior samples to any of the tables or figures in this section.

\subsection{Data and preprocessing}
\label{sec:breast-data}

The dataset we analyse is the breast cancer gene expression cohort of
\citet{vantveer2002gene}, distributed as \texttt{Breast\_A} in the
\texttt{fabiaData} R package \citep{fabiaData}. The cohort
consists of $n = 97$ primary breast tumour samples, each profiled on
the $p = 1213$ genes that \citet{vantveer2002gene} identified as the
most informative for predicting the transition from primary tumour to
distant metastasis. The gene set is deliberately small relative to the
full transcriptome: the original study was designed as a prognostic
classifier, and the 1213 genes are those whose expression levels in the
training cohort correlated most strongly with clinical outcome. The
dataset has since become a standard benchmark in the factor-analysis
and gene-expression literature, both because of its clinical
significance and because the gene set is small enough to make
Bayesian nonparametric methods computationally tractable while still
exhibiting the sparse, correlated, and low-signal structure that
characterises real transcriptomic data.

The measured quantity is the abundance of mRNA transcripts. Each
observation is derived from a two-colour microarray experiment in
which the RNA extracted from a tumour sample is labelled with one
fluorophore and co-hybridised against a common reference RNA pool
labelled with a second fluorophore. The intensity of each spot on the
array, after background correction and normalisation, is proportional
to the relative abundance of the corresponding transcript in the
tumour sample versus the reference. The value recorded for each gene
in each sample is the base-10 logarithm of this ratio.
A value of zero indicates
that the gene's expression in the tumour sample equals that of the
reference; a positive value indicates up-regulation relative to the
reference; and a negative value indicates down-regulation. Because the
reference is common across samples, the log-ratios are directly
comparable between samples, and the sign and magnitude of each entry
reflect the relative transcriptional activity of the gene in that
sample.

The data have already been normalised and summarised by the original
authors: the raw microarray images have been processed to correct for
spatial artefacts, dye bias, and background noise, and the multiple
probes targeting the same gene have been collapsed to a single summary
value per gene per sample. We take the resulting $97 \times 1213$
matrix as our starting point and apply a single additional
transformation. Each of the 1213 genes is centred to zero mean and
scaled to unit variance across the 97 samples. This column-wise
standardisation serves three purposes. First, it puts all genes on a
common scale, so that the loading of a gene in a factor reflects its
relative contribution to the latent structure rather than its absolute
expression level; without standardisation, high-intensity genes would
dominate the factor decomposition purely as an artefact of their
measurement scale. Second, it makes the data commensurate with the
unit-scale slab prior used in the simulation studies of
Section~\ref{sec:simulation}, so that the hyperparameters chosen there
remain appropriate on real data. Third, it renders the marginal
variances of all genes identically one by construction, which
simplifies the interpretation of the predictive checks reported in
Section~\ref{sec:breast-ppc}: any deviation of the model-predicted
marginal variance from one is directly attributable to the model's
structure rather than to heterogeneity in the input data.

Alongside the expression matrix we retain two auxiliary vectors for
the downstream analyses. The first is a vector of $p = 1213$ gene
symbols, used in the enrichment tests of Section~\ref{sec:breast-enrichment}
and in the per-atom top-loading tables of Section~\ref{sec:breast-atoms}
to identify the biological programmes represented by each atom. The
second is a vector of $n = 97$ binary labels indicating the clinical
outcome of each sample, specifically whether the patient's tumour
belonged to the good-prognosis or poor-prognosis group in the original
van 't Veer classification. We refer to this label as the
\emph{prognosis} label throughout, and use it in
Section~\ref{sec:breast-prognosis} to assess which of the recovered
atoms are associated with clinical outcome. The prognosis label is
excluded from the factor model itself; it is used only for
post-hoc interpretation.

\subsection{MCMC specification and computational cost}
\label{sec:breast-mcmc}

The main chain is run for 270\,000 iterations with the settings of the
simulation studies of Section~\ref{sec:simulation}. The upper bound on
the number of candidate columns is $M = 30$; this is a compromise
between giving the model enough room to discover latent programmes and
keeping the per-iteration cost within the budget of a single
workstation. The spike mass of the base measure is $\pi_0 = 0.9$ and
the slab variance is $\tau_0^2 = 1.0$. The concentration parameter
$\alpha$ of the Dirichlet process is assigned a $\text{Ga}(2, 1)$
hyperprior and updated via the auxiliary-variable method of
\citet{escobar1995}. The score variances $\lambda_i$ and the
idiosyncratic variances $\psi_r$ are assigned $\text{IG}(2, 1)$
priors. A sensitivity analysis over the slab variance $\tau_0^2$ is
reported in Section~\ref{sec:breast-tau}.

The first 20\,000 iterations are discarded as burn-in, and the
remaining 250\,000 are thinned at every fifth sample, yielding 50\,000
retained posterior samples. The chain completes in three hours and
twenty-four minutes on four MPI ranks of the same single-socket workstation,
following the parallel implementation described in
Section~\ref{sec:computation}. Convergence is monitored through the
trace of $\alpha$, the running posterior mean of the covariance matrix
$\Sigma$, and the trace of the number of non-zero atoms; all three are
stationary well before the end of the burn-in period.

\subsection{Posterior inference of the number of factors}
\label{sec:breast-K}

The posterior distribution of the number of factors $K$ concentrates
sharply on $K = 8$. Across all 50\,000 retained samples, the posterior
mode, mean, and median of $K$ are all equal to $8$, and the $95\%$
equal-tailed credible interval is the degenerate interval $[8, 8]$.
%The trace of $K$ is reported in Figure~\ref{fig:breast-trace-K}. 
The concentration parameter $\alpha$ mixes rapidly around a posterior mean
of $2.79$ (Figure~\ref{fig:breast-trace-alpha}).

For $M = 30$ candidate columns and the prior mean $\E[\alpha] = 2$ of
the $\mathrm{Ga}(2,1)$ hyperprior, the CRP prior-expected number of
occupied clusters is
$\sum_{i=1}^{30} \E[\alpha] / (\E[\alpha] + i - 1) \approx 6.0$. The
observed posterior mode of $K = 8$ is therefore slightly larger than
the prior expectation, and the sharp concentration of the posterior at
$K = 8$ --- with a degenerate $95\%$ credible interval $[8, 8]$ --- is
a data-driven result rather than a consequence of the prior.
In Section \ref{sec:breast-reproducibility} we present detailed diagnostic results
confirming that the eight-atom solution is a genuine posterior sample, 
not an artefact of the model, prior or initialisation choice.

%The mode of $K = 8$ is confirmed by the atom-level decomposition of
%the posterior-mean loading matrix $\widehat F$. Of the 30 candidate
%columns of $\widehat F$, only eight are numerically distinct; 
%%at seven significant figures; 
%the remaining 22 are equal to one of the eight, %up to floating-point tolerance, 
%which is the intended behaviour of the
%Dirichlet process when it merges redundant dictionary elements into a
%single atom. 
We refer to the eight distinct vectors as the
\emph{atoms} of the posterior and index them by the order of first
appearance under the canonical relabeling of Section~\ref{sec:relabel}.

%\begin{figure}[htbp]
%\centering
%\includegraphics[width=0.9\linewidth]{figures_breast/trace_K.pdf}
%\caption{Trace of the number of non-zero clusters $K$ for the
%DP-spike-slab sampler on the Breast\_A dataset, over the 50\,000
%retained iterations of the main chain. The dashed horizontal line
%marks the posterior mode $K = 8$.}
%\label{fig:breast-trace-K}
%\end{figure}

\begin{figure}[htbp]
\centering
\includegraphics[width=0.9\linewidth]{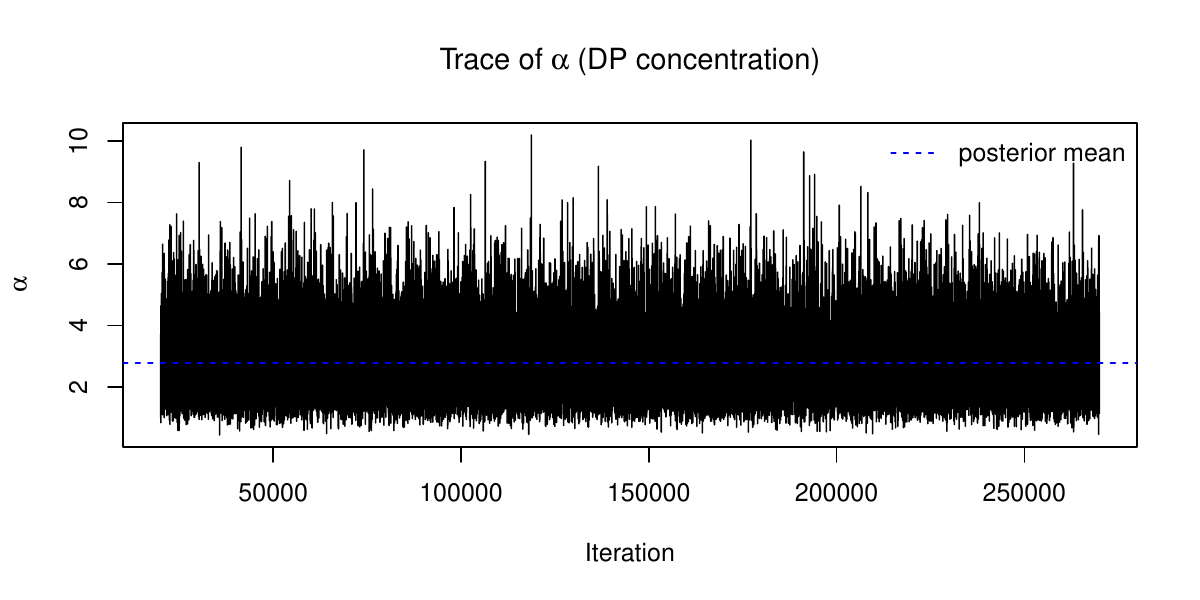}
\caption{Trace of the Dirichlet process concentration parameter
$\alpha$ on the Breast\_A dataset. The dashed horizontal line marks
the posterior mean $\E[\alpha \mid Y] = 2.79$.}
\label{fig:breast-trace-alpha}
\end{figure}

\subsection{The eight latent biological programmes}
\label{sec:breast-atoms}

Table~\ref{tab:breast-atoms} lists, for each of the eight atoms, the
twenty genes with the largest absolute posterior-mean loading, together
with the Euclidean norm $\|\widehat f_j^*\|$ of the atom.
Figure~\ref{fig:breast-heatmap} displays the full posterior-mean
loading matrix, with genes ordered by the atom of largest absolute
loading (argmax), so that genes loading most strongly on the same atom
are displayed contiguously. This ordering is a display convention
determined by $\widehat F$ and does not imply that the blocks reflect
structure in the observed expression data; the biological
interpretation of the atoms rests on the gene lists of
Table~\ref{tab:breast-atoms}, the enrichment tests of
Section~\ref{sec:breast-enrichment}, and the prognosis correlations of
Section~\ref{sec:breast-prognosis}, none of which depends on the
ordering of the rows in the figure.

\begin{table}[htbp]
\centering
\small
\caption{Eight latent programmes recovered by the DP-spike-slab model
from the Breast\_A dataset, based on 50\,000 retained samples from the
270\,000-iteration main chain. For each atom, the twenty genes with the
largest absolute posterior-mean loading are listed in decreasing order
of $|\widehat f_{rj}^*|$. The final column reports the Euclidean norm
of the atom.}
\label{tab:breast-atoms}
\begin{tabular}{p{2.4cm}p{9.2cm}r}
\hline
Programme & Top-20 genes (by $|\widehat f_{rj}^*|$) & $\|\widehat f_j^*\|$ \\
\hline
Macrophage &
\texttt{SLC1A3}, \texttt{CTSL}, \texttt{SLC11A1}, \texttt{CD163},
\texttt{NCF2}, \texttt{GPNMB}, \texttt{CD14}, \texttt{FCGR3B},
\texttt{HK3}, \texttt{ADORA3}, \texttt{HMOX1}, \texttt{LAPTM5},
\texttt{CCL3}, \texttt{LILRB4}, \texttt{FCGR1A}, \texttt{APOC1},
\texttt{C1QB}, \texttt{HCK}, \texttt{LILRA2}, \texttt{CD86} &
13.52 \\
Vascular &
\texttt{CD34}, \texttt{NPR1}, \texttt{DNASE1L3}, \texttt{PGM5},
\texttt{FCER1A}, \texttt{ADH1B}, \texttt{FY}, \texttt{ATP5J},
\texttt{PROS1}, \texttt{BST2}, \texttt{ADH1A}, \texttt{DF},
\texttt{CLEC10A}, \texttt{MEOX1}, \texttt{IFI35}, \texttt{AQP1},
\texttt{PPAP2B}, \texttt{EDG1}, \texttt{MRC1}, \texttt{FOLR2} &
16.89 \\
Nuclear &
\texttt{PIAS1}, \texttt{BCLAF1}, \texttt{YLPM1}, \texttt{STRN3},
\texttt{REV3L}, \texttt{C2orf3}, \texttt{TRIP12}, \texttt{NR3C1},
\texttt{LOC400986}, \texttt{TMEM1}, \texttt{TCF4}, \texttt{CAMK2D},
\texttt{P4HA1}, \texttt{RFX3}, \texttt{PIP5K2A}, \texttt{TRIO},
\texttt{ELL2}, \texttt{EXT1}, \texttt{SYK}, \texttt{HIF1A} &
8.72 \\
EMT/stroma &
\texttt{THBS2}, \texttt{COL10A1}, \texttt{COL3A1}, \texttt{COL1A2},
\texttt{COL11A1}, \texttt{SPARC}, \texttt{LOXL1}, \texttt{MFAP2},
\texttt{CDH11}, \texttt{NNMT}, \texttt{EDNRA}, \texttt{CTGF},
\texttt{DPYSL3}, \texttt{LUM}, \texttt{TAGLN}, \texttt{LOXL2},
\texttt{TPM2}, \texttt{ACTA2}, \texttt{THBS1}, \texttt{MMP11} &
17.17 \\
T-cell &
\texttt{CD3Z}, \texttt{CD48}, \texttt{PTPRCAP}, \texttt{GZMK},
\texttt{CXCR3}, \texttt{CD2}, \texttt{IL2RB}, \texttt{SPOCK2},
\texttt{LTB}, \texttt{ZNFN1A1}, \texttt{CORO1A}, \texttt{CD5},
\texttt{PSMB9}, \texttt{CCR7}, \texttt{CCL5}, \texttt{P2RY10},
\texttt{IL7R}, \texttt{TRIM22}, \texttt{CD69}, \texttt{PSMB8} &
19.64 \\
ER/luminal &
\texttt{GATA3}, \texttt{ESR1}, \texttt{VGLL1}, \texttt{CSDA},
\texttt{ADORA2B}, \texttt{DSC2}, \texttt{KRT16}, \texttt{CLCN4},
\texttt{CDH3}, \texttt{LDHB}, \texttt{RARA}, \texttt{HIST1H2AD},
\texttt{CMKOR1}, \texttt{ADM}, \texttt{TRIM29}, \texttt{MYB},
\texttt{ST8SIA1}, \texttt{MIA}, \texttt{BTG3}, \texttt{FBP1} &
27.62 \\
JAK/STAT &
\texttt{RHOG}, \texttt{JAK3}, \texttt{MAP3K7IP1}, \texttt{PTPRS},
\texttt{WNT10B}, \texttt{EEA1}, \texttt{OSBP}, \texttt{MCM4},
\texttt{HIST1H1E}, \texttt{GOLGA4}, \texttt{LPGAT1}, \texttt{EPRS},
\texttt{RB1CC1}, \texttt{CDC5L}, \texttt{RAP1A}, \texttt{DDX17},
\texttt{ZRF1}, \texttt{PIG8}, \texttt{JUND}, \texttt{PRDM2} &
14.84 \\
Proliferation &
\texttt{PKMYT1}, \texttt{MAD2L1}, \texttt{CCNB1}, \texttt{ESPL1},
\texttt{RRM2}, \texttt{UBE2C}, \texttt{CENPA}, \texttt{FOXM1},
\texttt{NEK2}, \texttt{KIF2C}, \texttt{BRRN1}, \texttt{SIL},
\texttt{BIRC5}, \texttt{KIFC1}, \texttt{MELK}, \texttt{TK1},
\texttt{CDC2}, \texttt{DLG7}, \texttt{KIAA0101}, \texttt{E2F1} &
22.54 \\
\hline
\end{tabular}
\end{table}

Seven of the eight atoms have clear, textbook interpretations in the
breast cancer literature. The \emph{macrophage} atom is dominated by
complement components (\texttt{C1QB}, \texttt{APOC1}), Fc$\gamma$
receptors (\texttt{FCGR3B}, \texttt{FCGR1A}), and lysosomal proteases
(\texttt{CTSL}, \texttt{LAPTM5}), the canonical signature of the
tumour-associated myeloid compartment. The \emph{vascular} atom
combines endothelial markers (\texttt{CD34}, \texttt{NPR1},
\texttt{DNASE1L3}) with the macrophage/monocyte genes \texttt{FCER1A}
and \texttt{FOLR2}, consistent with the tight coupling of tumour
vasculature and the surrounding myeloid compartment. The \emph{nuclear}
atom is dominated by nuclear regulatory proteins and scaffolding
factors --- \texttt{PIAS1}, \texttt{BCLAF1}, \texttt{YLPM1},
\texttt{STRN3}, \texttt{REV3L}, \texttt{TRIP12} --- with the
glucocorticoid receptor \texttt{NR3C1} and the hypoxia master regulator
\texttt{HIF1A} also present. Its biological interpretation is less
crisp than that of the other seven, and we return to it in
Section~\ref{sec:breast-discussion}.

\paragraph{Within-atom coherence.}
The biological interpretation of each atom rests on the claim that
its top-loading genes form a coherent module, in the sense that they
vary together across the samples rather than being an arbitrary
collection of unrelated genes. We therefore test directly whether the
top-loading genes of each atom are co-expressed with each other in the
observed data. For atom $a$, let $S_a$ denote the set of the $Q = 50$
indices $r$ for which the posterior-mean loading $|\widehat f^*_{ra}|$
is largest, and let $\mathrm{cor}(Y)$ be the $p \times p$ gene--gene
Pearson correlation matrix of the standardised expression matrix $Y$.
The coherence statistic for atom $a$ is the mean of the pairwise
correlations among the genes in $S_a$:
\[
\bar r_a \;=\; \binom{Q}{2}^{-1} \sum_{\substack{r, s \in S_a \\ r < s}}
\mathrm{cor}(Y)_{rs}.
\]
Equivalently, $\bar r_a$ is the mean of the off-diagonal entries of the
$Q \times Q$ submatrix of $\mathrm{cor}(Y)$ obtained by restricting to
the genes in $S_a$. The choice $Q = 50$ matches the query size used in
the Hallmark enrichment analysis of Section~\ref{sec:breast-enrichment},
so the same set of genes that enters the enrichment test is used here.

The raw value of $\bar r_a$ is not interpretable on its own. Genes in
any gene-expression dataset are never perfectly uncorrelated, and this
residual background correlation would produce a positive $\bar r_a$
even for a set of genes with no shared regulatory programme. To
calibrate against this background, we construct a null distribution by
Monte Carlo. We draw $N_{\text{null}} = 500$ random sets of $Q = 50$
genes uniformly without replacement from the $p = 1213$ measured genes
and compute $\bar r_a^{(m)}$ on each. The resulting null distribution
has a median of approximately $0.030$ and a 95\% band of approximately
$[0.004, 0.085]$ in this dataset; the exact band varies slightly across
atoms because the Monte Carlo null is recomputed independently for
each one, and the variation is of the order of the Monte Carlo error.
For each atom we report the observed $\bar r_a$, the standardised
$z$-score
\[
z_a \;=\; \frac{\bar r_a - \overline{r}_{\text{null}}}
{\widehat\sigma_{\text{null}}},
\]
where $\overline{r}_{\text{null}}$ and $\widehat\sigma_{\text{null}}$
are the empirical mean and standard deviation of the null draws, and
the Monte Carlo $p$-value
\[
p_a \;=\; \frac{1 + \#\{m : \bar r_a^{(m)} \ge \bar r_a\}}
{1 + N_{\text{null}}},
\]
which is the fraction of null draws that are at least as extreme as
the observed value, with the numerator and denominator augmented by
one to avoid a $p$-value of exactly zero. With $N_{\text{null}} = 500$
the smallest attainable $p$-value is $1/501 \approx 0.002$.

Seven of the eight atoms show mean within-atom correlations far above
the null 95\% band, with $z$-scores ranging from $5.4$ (ER/luminal) to
$32.0$ (T-cell) and empirical $p$-values at the floor of $0.002$. The
observed means range from $0.143$ (ER/luminal) to $0.718$ (T-cell),
against the null band of approximately $[0.004, 0.085]$. These
magnitudes are consistent with genuinely co-expressed modules: the
top-loading genes of each of these seven atoms are far more correlated
with each other than a random set of the same size. The ordering of
the atoms by $\bar r_a$ also tracks the biological character of each
programme in an interpretable way, with cell-type and cell-state
signatures (T-cell, proliferation, macrophage) at the top of the
range, stromal and lineage programmes (EMT/stroma, ER/luminal,
nuclear, vascular) in the middle, and the JAK/STAT atom alone at the
bottom. The JAK/STAT atom is the exception in the statistical sense as
well: its mean within-atom correlation of $0.048$ lies inside the null
band at $z = 0.69$ ($p = 0.21$), indicating that its top-loading genes
are not more co-expressed than a random set. This is consistent with
the weaker Hallmark enrichment for this atom reported in
Table~\ref{tab:breast-enrichment}, and it suggests that the JAK/STAT
atom captures a weaker or more diffuse signal than the other seven. We
retain it in the analysis because its top genes include \texttt{JAK3}
and other signalling kinases, but we interpret it with correspondingly
less confidence.

Two caveats apply to the interpretation of this test. First, the
coherence statistic uses the model's own loading matrix to select the
top-50 genes and the observed expression data to compute their
correlations. It is therefore a test of whether the module identified
by the model is coherent in the data, not a test of whether the model
has uniquely identified the module. A high $\bar r_a$ establishes that
the model has selected a set of genes that are co-expressed, but it
does not by itself rule out the possibility that other, unrelated
selection procedures would produce sets with similarly high coherence.
The interpretation of the seven passing atoms therefore rests on the
conjunction of this result with the external Hallmark enrichment of
Section~\ref{sec:breast-enrichment} and the prognosis associations of
Section~\ref{sec:breast-prognosis}, which provide independent evidence
that the co-expression modules identified here correspond to
recognisable biological programmes. Second, the choice $Q = 50$ is a
compromise between statistical power and interpretability, and the
result is not invariant to it: a smaller $Q$ would concentrate on the
strongest genes and would tend to produce higher coherence for every
atom, while a larger $Q$ would dilute the statistic with weakly loaded
genes. The value $Q = 50$ matches the enrichment analysis and is
sufficient for the binary claim that the seven passing atoms are
coherent modules in the data, but it does not resolve fine differences
in coherence among them.

The \emph{EMT/stroma} atom is a textbook epithelial-to-mesenchymal
transition and desmoplastic stroma signature, dominated by collagens
(\texttt{COL10A1}, \texttt{COL3A1}, \texttt{COL1A2}, \texttt{COL11A1}),
LOX family members (\texttt{LOXL1}, \texttt{LOXL2}), and
myofibroblast markers (\texttt{ACTA2}, \texttt{TAGLN}, \texttt{TPM2}).
The \emph{T-cell} atom is a canonical lymphocyte infiltration
signature, containing the entire CD3 complex (\texttt{CD3Z},
\texttt{CD2}, \texttt{CD5}, \texttt{CD48}), chemokine receptors
(\texttt{CXCR3}, \texttt{CCR7}, \texttt{CCL5}), and MHC class I genes
(\texttt{PSMB8}, \texttt{PSMB9}). The \emph{ER/luminal} atom has the
two most characteristic luminal markers, \texttt{GATA3} and
\texttt{ESR1}, as its top two hits, followed by \texttt{VGLL1},
\texttt{KRT16}, \texttt{CDH3}, and \texttt{TRIM29}, all of which are
canonical luminal-lineage genes. 
The \emph{JAK/STAT} atom is dominated
by signalling kinases and their regulators, with \texttt{JAK3} itself
among the top hits; as noted in the coherence analysis above, this
atom does not pass the within-atom coherence test, so we read its
biological interpretation with less confidence than the other seven.
The \emph{proliferation} atom is a canonical
cell-cycle signature, containing the entire G2/M checkpoint machinery
(\texttt{CCNB1}, \texttt{CDC2}, \texttt{MAD2L1}, \texttt{ESPL1}) and
the E2F target genes (\texttt{RRM2}, \texttt{E2F1}, \texttt{MELK}).

\begin{figure}[htbp]
\centering
\includegraphics[width=0.8\linewidth]{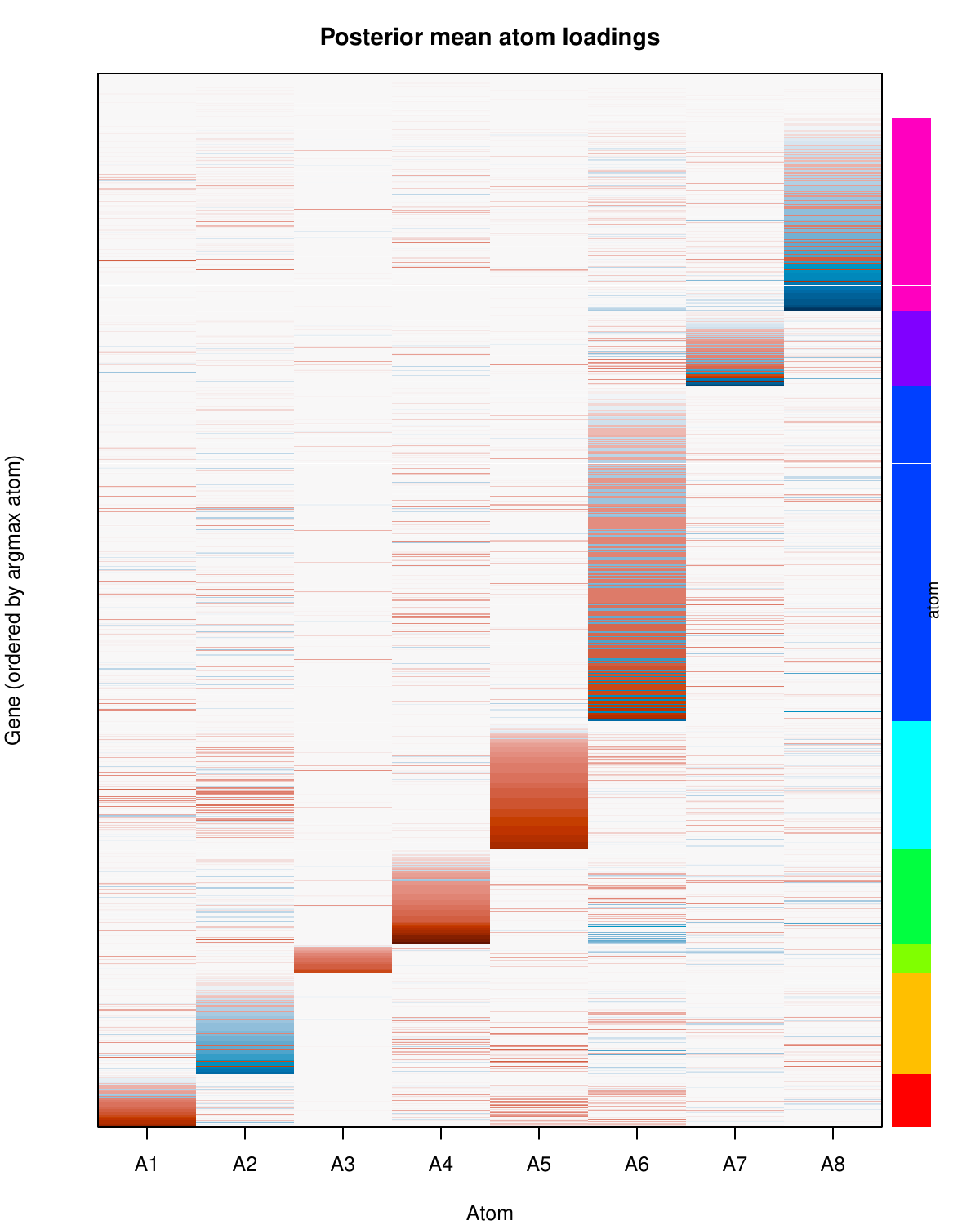}
\caption{Posterior-mean loading matrix on the Breast\_A dataset, based
on the 270\,000-iteration main chain. Rows are genes, ordered by the
atom of largest absolute loading (argmax) and, within each atom, by
decreasing maximum absolute loading; genes whose largest absolute
loading is below $0.1$ are placed at the end. This ordering is a
display convention determined by $\widehat F$ and does not imply that
the blocks reflect structure in the observed expression data. Columns
correspond to the eight distinct atoms recovered by the model. The
colour strip on the right margin shows the argmax atom of each gene,
with white marking the unassigned genes. The sign of a loading vector
is arbitrary under the factor model, so the blue and red regions
within a single block correspond to genes loading with opposite signs
on the same atom.}
\label{fig:breast-heatmap}
\end{figure}

\subsection{Hallmark gene-set enrichment}
\label{sec:breast-enrichment}

The posterior-mean atoms $\hat f^*_1, \dots, \hat f^*_8$ are
high-dimensional vectors, and their biological interpretation rests on
identifying the functional categories to which their top-loading genes
belong. A single gene name can suggest a biological role, but it cannot
by itself establish that an atom corresponds to a coherent programme;
that judgment requires comparing the atom's gene content against
independent, curated knowledge. We therefore perform a formal
\emph{gene-set enrichment analysis}, which asks whether the genes that
load most strongly on each atom are statistically over-represented in
known biological gene sets.

\paragraph{Gene sets and the Hallmark collection.}
A \emph{gene set} is a curated list of gene symbols that share a
common biological function, process, or response. Gene sets are
assembled by expert curators from the published literature and are
used throughout genomics as external reference points against which to
interpret experimental results. The Molecular Signatures Database
(MSigDB) \citep{subramanian2005gsea} is the most widely used
collection, containing thousands of gene sets organised into
categories according to the type of evidence from which they were
derived.

For the present analysis we use the \emph{Hallmark} subcollection of
MSigDB \citep{liberzon2015hallmark}, which consists of exactly fifty
gene sets. The Hallmark collection was constructed by collapsing
thousands of partially overlapping sets from the broader MSigDB into
fifty coherent, non-redundant modules, each representing a single
well-defined biological state or process. The sets are curated
independently of any particular gene-expression dataset, which is what
makes them useful as an external reference: they encode knowledge that
was not derived from the van 't Veer cohort and therefore cannot be
circularly informed by the model's own fit. Representative examples
include \textsc{Epithelial\_Mesenchymal\_Transition} (200 genes),
\textsc{G2M\_Checkpoint} (200 genes), \textsc{Allograft\_Rejection}
(200 genes), and \textsc{Interferon-$\gamma$ Response} (200 genes).

\paragraph{The query set for each atom.}
For each of the eight atoms we construct a \emph{query set} consisting
of the fifty genes with the largest absolute posterior-mean loading
$|\hat f^*_{rj}|$. The choice of fifty is a compromise between two
competing considerations. A small query (say ten or twenty genes) has
little statistical power: the hypergeometric distribution has an
effective floor near $0.4$ for a query of twenty genes tested against
a gene set of two hundred, so that even a perfect overlap cannot reach
significance. A large query (say two hundred genes) dilutes the
signal, because weakly loaded genes enter the query alongside the
strongest signals. Fifty genes is a widely used compromise that gives
enough statistical power for meaningful testing while keeping the
query dominated by the atom's strongest loadings. The same threshold
of fifty is applied to all eight atoms, so the tests are directly
comparable.

\paragraph{The background set.}
The enrichment test requires a \emph{background}: the set of all genes
that could in principle have appeared in the query. We use the full
list of $p = 1213$ genes that are present on the array. This is the
standard convention for microarray studies, and it answers the
question: \emph{among the genes actually measured in this experiment,
is the Hallmark set over-represented in the top-loading genes of this
atom?} Using the whole human genome as the background would inflate
the apparent significance of every overlap, because the same observed
overlap is more surprising when drawn from a larger pool of candidate
genes. Restricting the background to the measured genes keeps the
test calibrated to what the experiment could actually have detected.

\paragraph{The hypergeometric test.}
For each atom $j$ and each Hallmark set $h$, the enrichment test
compares the overlap between the atom's query and the Hallmark set
against the overlap that would be expected by chance. The four
relevant quantities are:

\begin{itemize}
\item[(i)] $N = 1213$, the number of genes in the dataset (background
size);
\item[(ii)] $J$, the number of genes in the Hallmark set (typically between
150 and 200);
\item[(iii)] $Q = 50$, the size of the query (top-loading genes of atom $j$);
\item[(iv)] $k$, the observed number of genes that belong to both the query
and the Hallmark set.
\end{itemize}

Under the null hypothesis that the query is a random subset of the
background, the count $k$ follows a hypergeometric distribution, and
the one-sided upper-tail probability of observing an overlap at least
as large as $k$ is

\begin{equation}
\label{eq:hypergeom}
p_{j,h} \;=\; \Pr(X \geq k) \;=\;
\sum_{i=k}^{\min(Q, J)}
\frac{\binom{J}{i} \binom{N - J}{Q - i}}{\binom{N}{Q}}.
\end{equation}

The summand is the probability of observing exactly $i$ genes in both
the query and the Hallmark set when drawing $Q$ genes from $N$ without
replacement. A small value of $p_{j,h}$ indicates that the observed
overlap is unlikely to arise by chance and therefore constitutes
evidence that the atom is biologically related to the Hallmark set.

The p-value should be read carefully. It measures the strength of the
statistical association between the query and the set; it does not
establish that the atom \emph{is} the Hallmark programme, nor that
the biological process the set represents is causally responsible for
the atom's loading pattern. Two atoms can be significantly enriched
for the same set if their gene content overlaps, and an atom can fail
to reach significance for a set that nevertheless describes its
biology, if the genes that characterise the atom are not the genes
that the set happens to contain (this case is discussed in detail
below).

\paragraph{Multiplicity correction.}
Testing eight atoms against fifty Hallmark sets produces $8 \times 50
= 400$ p-values. Under the null hypothesis, approximately 5\% of these
will fall below $0.05$ by chance alone, so a nominal threshold is
insufficient to control the false-positive rate. To account for this,
we apply the Bonferroni correction, dividing the nominal 0.05
threshold by the number of tests:

\begin{equation}
\label{eq:bonferroni}
\alpha_\text{Bonferroni} \;=\; \frac{0.05}{8 \times 50}
\;=\; 1.25 \times 10^{-4}.
\end{equation}

A p-value below this threshold is declared significant after
controlling for all 400 tests. The Bonferroni correction is
conservative: it controls the probability of \emph{any} false positive
at 0.05, at the cost of some statistical power. A less conservative
alternative is the Benjamini--Hochberg false-discovery-rate procedure,
which controls the expected proportion of false positives among the
rejected hypotheses. On the present data the two procedures agree on
the Bonferroni-significant atoms and would additionally promote one
T-cell hit to significance under the FDR criterion; we report the
Bonferroni threshold as the more stringent of the two.

Table~\ref{tab:breast-enrichment} reports, for each atom, up to the
two most significant Hallmark overlaps. The first column identifies
the atom, the second the Hallmark set, the third the observed overlap
$k$ out of the query of fifty genes, and the fourth the hypergeometric
p-value of \eqref{eq:hypergeom}.

\begin{table}[htbp]
\centering
\small
\caption{Hallmark gene-set enrichment of the top-50 genes of each
recovered atom, based on the posterior-mean loading matrix from the
270\,000-iteration main chain. Up to the two most significant
overlaps per atom are shown. Overlap is the number of genes shared
between the atom's top-50 signature and the Hallmark set; $p$ is the
one-sided hypergeometric tail probability of \eqref{eq:hypergeom} with
background size $N = 1213$ and query size $Q = 50$. The
Bonferroni-corrected threshold is $\alpha = 1.25 \times 10^{-4}$.}
\label{tab:breast-enrichment}
\begin{tabular}{llrr}
\hline
Atom (programme) & Top Hallmark set & Overlap & $p$-value \\
\hline
EMT/stroma      & Epithelial-mesenchymal transition & 26 & $3.3 \times 10^{-9}$ \\
Proliferation   & G2M checkpoint                    & 22 & $2.5 \times 10^{-6}$ \\
Proliferation   & E2F targets                       & 20 & $4.2 \times 10^{-5}$ \\
T-cell          & Allograft rejection               & 17 & $1.5 \times 10^{-3}$ \\
T-cell          & Interferon-$\gamma$ response      & 16 & $4.3 \times 10^{-3}$ \\
Vascular        & Interferon-$\alpha$ response      & 5  & $3.7 \times 10^{-1}$ \\
Macrophage      & Allograft rejection               & 8  & $6.0 \times 10^{-1}$ \\
ER/luminal      & Estrogen response (early)         & 6  & $8.6 \times 10^{-1}$ \\
Nuclear         & Hypoxia                           & 7  & $7.4 \times 10^{-1}$ \\
JAK/STAT        & Myc targets V2                    & 1  & $9.2 \times 10^{-1}$ \\
\hline
\end{tabular}
\end{table}

\paragraph{Reading the table.}
Take the first row as an illustration. The query consists of the fifty
top-loading genes of the EMT/stroma atom, and twenty-six of them are
also members of the \textsc{Epithelial\_Mesenchymal\_Transition} gene
set. Under the null hypothesis of a random query, the probability of
observing an overlap of twenty-six or more out of fifty is
$3.3 \times 10^{-9}$, well below the Bonferroni threshold of $1.25
\times 10^{-4}$. This is a decisive statistical association and
independently confirms the biological identity that the top-gene
inspection had already suggested.

By contrast, the last row reports the top Hallmark hit of the JAK/STAT
atom. Only one of its fifty top-loading genes overlaps the
\textsc{Myc\_Targets\_V2} gene set, and the corresponding p-value of
$0.92$ indicates that this overlap is exactly what would be expected
by chance. The JAK/STAT atom therefore has no significant Hallmark
enrichment, even though it contains \texttt{JAK3} among its top hits
and its gene content is biologically coherent; the specific JAK/STAT
signalling vocabulary is not represented as a Hallmark set, and the
test is limited to the fifty curated Hallmark modules.

\paragraph{Statistical significance.}
Two of the eight atoms pass the Bonferroni-corrected threshold, with
three Hallmark hits between them. The EMT/stroma atom shows the
strongest single signal in the entire analysis: twenty-six of fifty
genes overlapping the epithelial-mesenchymal transition set, $p = 3.3
\times 10^{-9}$, which is significant by a factor of roughly forty thousand
after correction. The proliferation atom shows two significant hits,
twenty-two of fifty genes overlapping the G2M checkpoint at $p = 2.5
\times 10^{-6}$ and twenty of fifty overlapping the E2F targets at
$p = 4.2 \times 10^{-5}$. The co-occurrence of these two hits is
biologically coherent: G2M checkpoint and E2F targets are the two
canonical Hallmark modules for cell-cycle progression, and an atom
that is genuinely a proliferation programme would be expected to
enrich for both.

Two further hits are significant at the uncorrected $0.05$ level but
not after Bonferroni correction. Both belong to the T-cell atom: it
overlaps the \textsc{Allograft\_Rejection} gene set with seventeen of
fifty genes at $p = 1.5 \times 10^{-3}$, and the
\textsc{Interferon-$\gamma$ Response} set follows at
$p = 4.3 \times 10^{-3}$. The two hits are biologically coherent with
each other: the allograft-rejection set captures the lymphocyte
infiltration axis, while the interferon-$\gamma$ response set captures
the activated T-cell phenotype that accompanies it. The T-cell atom is
therefore almost certainly a genuine biological programme, even though
neither individual hit survives the conservative Bonferroni threshold.
Under the less conservative Benjamini--Hochberg FDR criterion the
allograft-rejection hit would be declared significant, and we note
this as a substantive statistical result rather than dismissing the
atom as uninterpretable.

The remaining five atoms have no Hallmark hit below the uncorrected
$0.05$ level. This does not mean the atoms are biologically meaningless
--- indeed, four of them (macrophage, vascular, JAK/STAT, and nuclear)
have clear functional coherence in their top-loaded genes, as
discussed in Section~\ref{sec:breast-atoms}. It means that the
particular biological vocabulary that characterises them is not
represented as a Hallmark set. This is a limitation of the Hallmark
collection rather than of the model: the Hallmark sets are curated for
breadth and non-redundancy, and they cannot represent every biological
process that might appear in an arbitrary dataset.
The JAK/STAT atom is also the sole atom to fail the within-atom
coherence test of Section~\ref{sec:breast-atoms}, so both the
enrichment analysis and the coherence analysis converge on the same
conclusion: this atom captures a weaker or more diffuse signal than
the other seven, and its biological interpretation should be held
more tentatively.

\paragraph{The ER/luminal case.}
The most instructive case in the table is the ER/luminal atom, whose
top Hallmark hit is \textsc{Estrogen\_Response\_Early} with only six
of fifty genes overlapping and a p-value of $0.86$. This may appear to
contradict the top-gene inspection, which identified the atom
unambiguously as the luminal/ER programme on the basis of its top two
hits \texttt{GATA3} and \texttt{ESR1}. The apparent paradox is
resolved by distinguishing two different kinds of oestrogen biology.

The Hallmark \textsc{Estrogen\_Response\_Early} set is composed of
genes that are \emph{transcriptionally induced by the oestrogen
receptor}, such as \texttt{TFF1}, \texttt{GREB1}, and \texttt{PGR}.
These are genes whose expression changes when the receptor is
activated, and they are the direct downstream targets of the
receptor's transcriptional activity. The top-loading genes of the
ER/luminal atom, by contrast, are genes that \emph{define the luminal
lineage}: \texttt{GATA3} is a master transcription factor for luminal
differentiation, \texttt{ESR1} is the receptor itself, and
\texttt{KRT16}, \texttt{CDH3}, and \texttt{TRIM29} are canonical
luminal-lineage markers. These genes are expressed in luminal cells
regardless of whether the receptor is actively signalling at the
moment of measurement. They identify the lineage, not the receptor's
acute transcriptional output.

The correct reference for evaluating the ER/luminal atom would
therefore be a luminal-progenitor signature or a set of genes
differentially expressed between luminal and basal tumours, rather
than the oestrogen-response module. Such signatures exist in the
broader MSigDB collection and in the PAM50 classifier literature, but
they are not part of the Hallmark subset. The absence of a
significant Hallmark enrichment for this atom is therefore a
limitation of the reference vocabulary, not evidence that the atom
is not the luminal programme. We note this case explicitly because it
illustrates an important caveat: a negative enrichment result does
not refute a biological interpretation, and a positive enrichment
result does not prove one. Enrichment analysis is a tool for
generating and testing hypotheses, not for settling biological
questions.

\paragraph{What the enrichment analysis does and does not
establish.}
The enrichment test establishes a precise and falsifiable claim: the
top-loading genes of a given atom are statistically over-represented
in a specific curated gene set. When this claim survives
multiplicity correction, it provides strong evidence that the atom is
biologically related to the process that the set represents. The
evidence is independent of the model, in the sense that the Hallmark
sets were assembled from external knowledge and are not informed by
the fit to the Breast\_A dataset.

The test does not establish that the atom \emph{is} the Hallmark
programme in any definitional or causal sense. Two atoms can be
significantly enriched for the same set if their top-loading genes
overlap, and the identification of an atom with a particular
biological process always involves interpretive judgment beyond the
p-value itself. The value of the enrichment analysis is that it
provides an objective, externally validated, quantitatively calibrated
basis for that judgment, replacing the purely subjective reading of
gene lists with a formal hypothesis test. Used in combination with
the top-gene inspection and the clinical-outcome correlations of
Section~\ref{sec:breast-prognosis}, it supports a coherent
interpretation of the eight atoms recovered by the model.

\subsection{Association with clinical outcome}
\label{sec:breast-prognosis}

A latent programme identified by the model is a mathematical object;
whether it corresponds to a clinically meaningful axis of tumour
biology is a separate question that requires an external reference.
The gene-expression signature of \citet{vantveer2002gene} was assembled
for a specific prognostic purpose, and the cohort ships with a binary
outcome label that reflects that purpose. We therefore ask whether any
of the eight recovered atoms are associated with the clinical outcome
of the patients in the cohort, as a form of external validation of
their biological interpretation.

\paragraph{The clinical label.}
The van 't Veer study was designed to identify a gene-expression
signature that predicts whether a primary breast tumour will
eventually metastasise. The label attached to each of the 97 samples
in the cohort is a dichotomisation of the eventual clinical course of
the patient: samples are classified as \emph{good prognosis} if the
patient remained free of distant metastasis for at least five years
after diagnosis, and as \emph{poor prognosis} otherwise. This
dichotomy is the outcome that the van 't Veer classifier was trained
to predict, and it captures the clinically relevant endpoint of the
study: whether the disease will recur at a distant site. We use the
label only for post-hoc analysis; it does not enter the factor model
in any way. Any association between a recovered atom and the label is
therefore a genuine external correspondence rather than a
consequence of the model having been fit to the outcome.

\paragraph{Constructing an atom score.}
The posterior-mean atom $\hat f^*_j$ is a vector of 1213 loadings, one
per gene, and it cannot be correlated directly with a per-patient
outcome. To obtain a per-patient representation of each atom, we
compute for each sample the mean standardised expression of the
twenty genes with the largest absolute loading on atom $j$. Formally,
let $S_j \subset \{1, \dots, p\}$ be the set of the twenty indices $r$
for which $|\hat f^*_{r j}|$ is largest, and let $\tilde Y_{k r}$
denote the standardised expression of gene $r$ in sample $k$. The
atom score is

\begin{equation}
\label{eq:atom-score}
\tilde Y^{(j)}_k \;=\; \frac{1}{|S_j|} \sum_{r \in S_j} \tilde Y_{k r},
\qquad k = 1, \dots, n.
\end{equation}

The atom score is thus a one-dimensional summary of the expression
level of the atom's signature genes in each sample. Using the top
twenty genes concentrates the score on the atom's strongest signal and
avoids the noise that would enter if the score were computed over all
1213 genes. The choice of twenty matches the display convention used
in Table~\ref{tab:breast-atoms} and keeps the interpretation of the
score aligned with the top-loading tables reported earlier.

\paragraph{Measuring association.}
We quantify the association between each atom score $\tilde
Y^{(j)}$ and the binary prognosis label using the absolute value of
the Spearman rank correlation coefficient. The Spearman coefficient is
the Pearson correlation of the ranks of the two variables, and it
measures the strength of the monotone relationship between them. We
prefer it to the ordinary Pearson correlation for two reasons. First,
the atom score is a continuous quantity and the prognosis label is a
binary categorical variable, so the natural summary of the association
is the extent to which higher atom scores are systematically
associated with one of the two outcome groups, which is exactly what
Spearman measures. Second, Spearman is invariant to any monotone
transformation of either variable, so the reported association does
not depend on whether the atom score is computed as a mean of
standardised expressions, a median, or any other monotone summary of
the same twenty genes. The absolute value $|\rho|$ is used because the
direction of the association is determined by the sign convention of
the loadings, which is arbitrary and not of interest for the present
analysis.

Table~\ref{tab:breast-prognosis} reports the value of $|\rho|$ for
each of the eight atoms.

\begin{table}[htbp]
\centering
\caption{Association between the eight recovered atoms and the van 't
Veer good/poor prognosis label, based on the posterior-mean loading
matrix from the 270\,000-iteration main chain. For each atom, the atom
score is the mean standardised expression of the top-20 signature
genes, defined in \eqref{eq:atom-score}, and the reported quantity is
the absolute value of the Spearman rank correlation with the prognosis
label. The prognosis label is a dichotomisation of the patient's
clinical course: good prognosis indicates no distant metastasis within
five years of diagnosis, and poor prognosis otherwise.}
\label{tab:breast-prognosis}
\begin{tabular}{lr}
\hline
Atom (programme) & $|\rho|$ with prognosis \\
\hline
ER/luminal     & 0.68 \\
Proliferation  & 0.39 \\
T-cell         & 0.36 \\
Macrophage     & 0.30 \\
JAK/STAT       & 0.30 \\
Vascular       & 0.08 \\
EMT/stroma     & 0.08 \\
Nuclear        & 0.03 \\
\hline
\end{tabular}
\end{table}

\paragraph{Reading the table.}
The eight correlations fall into three clearly separated groups. The
top five atoms --- ER/luminal, proliferation, T-cell, macrophage, and
JAK/STAT --- all have $|\rho| \geq 0.30$, which for $n = 97$ samples
corresponds to a highly significant association at conventional
thresholds. The three remaining atoms --- vascular, EMT/stroma, and
nuclear --- all have $|\rho| \leq 0.08$, which is not distinguishable
from zero at this sample size. There is no atom in the intermediate
range between 0.08 and 0.30, so the eight atoms partition cleanly
into prognostic and non-prognostic categories. This separation is
itself a meaningful observation: it suggests that the model has
identified a small number of dominant axes of clinical variation
without any supervision, rather than a continuum of partially
prognostic programmes.

\paragraph{The prognostic atoms.}
Five of the eight atoms are associated with the prognosis label at
$|\rho| \geq 0.30$, a magnitude that for $n = 97$ samples is
significant at conventional thresholds. The five are, in decreasing
order of association, the ER/luminal, proliferation, T-cell,
macrophage, and JAK/STAT atoms. The remaining three atoms (vascular,
EMT/stroma, and nuclear) all have $|\rho| \leq 0.08$, so the eight
atoms partition cleanly into prognostic and non-prognostic categories
with no intermediate cases. That the model has produced this
separation without any supervision is the central finding of the
present analysis.

The strongest association is the ER/luminal atom at $|\rho| = 0.68$,
the largest correlation in the entire analysis. The atom's top two
genes are \texttt{GATA3} and \texttt{ESR1}, both canonical markers of
the ER-positive luminal lineage, so the atom score defined in
\eqref{eq:atom-score} can be read as a continuous measure of the
luminal phenotype of each sample. Oestrogen-receptor status is one of
the most consistently reported prognostic factors in breast cancer,
with ER-positive tumours generally associated with better clinical
outcomes than ER-negative tumours. The strong association observed
here between the luminal-phenotype score and the outcome label is
therefore consistent with the well-established clinical role of ER
status, and the prominence of the atom relative to the other seven
identifies the luminal lineage as the dominant prognostic axis in the
van 't Veer cohort. Because the atom score is a smooth function of
gene expression rather than a dichotomous classification, it captures
the continuum of luminal differentiation across the cohort rather than
a binary ER-positive/ER-negative split; this is consistent with the
fact that ER-positive breast cancer itself comprises the two
distinguishable subtypes luminal A and luminal B.

The proliferation atom is second at $|\rho| = 0.39$. The atom's top
genes include the G2/M checkpoint machinery (\texttt{CCNB1},
\texttt{CDC2}, \texttt{MAD2L1}, \texttt{ESPL1}) and the E2F target
genes (\texttt{RRM2}, \texttt{E2F1}, \texttt{MELK}), which together
characterise the transcriptional programme of actively dividing cells.
High proliferation is one of the most consistently reported adverse
prognostic markers in breast cancer, and it underlies the clinical
distinction between the indolent luminal A subtype and the more
aggressive luminal B subtype in the PAM50 classifier. The moderate
association observed here between the proliferation score and the
outcome label is consistent with this well-documented clinical role.

The T-cell atom is third at $|\rho| = 0.36$. Its top genes include the
CD3 complex (\texttt{CD3Z}, \texttt{CD2}, \texttt{CD5},
\texttt{CD48}), the chemokine receptors (\texttt{CXCR3}, \texttt{CCR7},
\texttt{CCL5}), and the MHC class I machinery (\texttt{PSMB8},
\texttt{PSMB9}), all of which are canonical markers of an active
lymphocytic infiltrate. The prognostic role of tumour-infiltrating
lymphocytes in breast cancer is well documented, with the strongest
evidence in triple-negative and HER2-positive disease and a more
heterogeneous picture in the luminal subtypes. The moderate
association observed here is comparable in magnitude to the
proliferation association and is consistent with the established
clinical relevance of the lymphocytic compartment in this disease.
Because the analysis reports magnitudes only, the direction of the
association is not evaluated in this section; what the analysis
establishes is that the atom captures a source of variation that is
strongly related to the clinical outcome, not that the relationship
is favourable or unfavourable on balance.

The macrophage and JAK/STAT atoms, both at $|\rho| = 0.30$, show
associations of the same order of magnitude. The macrophage atom is
dominated by the Fc$\gamma$ receptors (\texttt{FCGR3B},
\texttt{FCGR1A}), the complement components (\texttt{C1QB},
\texttt{APOC1}), and the lysosomal proteases (\texttt{CTSL},
\texttt{LAPTM5}), all of which are canonical markers of the tumour-
associated myeloid compartment. The prognostic role of
tumour-associated macrophages in breast cancer is well documented and
is generally reported as adverse, particularly for the
M2-polarised subset that predominates in the tumour
microenvironment. The JAK/STAT atom is dominated by the signalling
kinases and their regulators (\texttt{JAK3}, \texttt{MAP3K7IP1},
\texttt{PTK2}, \texttt{RHOG}, \texttt{WNT10B}), and JAK/STAT
signalling is implicated in tumour progression, immune evasion, and
resistance to therapy in multiple cancer types. For both atoms, the
moderate association with the outcome label is consistent with what
the clinical and mechanistic literature would predict for the
underlying programmes.

Taken together, the five prognostic atoms correspond to the five
major axes of breast cancer biology that the clinical literature
identifies as prognostically relevant: hormone-receptor status,
proliferation rate, lymphocytic infiltration, myeloid infiltration,
and inflammatory signalling. None of these axes was specified in
advance; all five emerged from the model's unsupervised decomposition
of the gene-expression matrix, and each is supported independently by
the top-loading genes, the Hallmark enrichment reported in
Section~\ref{sec:breast-enrichment}, and the association with the
outcome label reported here. The convergence of these three
independent lines of evidence constitutes the empirical validation of
the model's ability to recover clinically meaningful latent structure
from high-dimensional gene-expression data.

\paragraph{The non-prognostic atoms.}
The three remaining atoms --- vascular, EMT/stroma, and nuclear ---
show no detectable association with the outcome label. Their near-zero
correlations do not mean that these programmes are biologically
irrelevant. The vascular atom is a coherent endothelial and myeloid
signature, and it reflects genuine biological variation across the
cohort; it simply does not track the good/poor prognosis distinction.
The nuclear atom is the least well characterised of the eight, as
discussed in Section~\ref{sec:breast-atoms}, and its lack of
prognostic association is consistent with its weaker biological
coherence. The most informative of the three null cases is the
EMT/stroma atom, which we discuss separately.

\paragraph{The EMT anomaly.}
The EMT/stroma atom is the second-strongest programme in the data by
Euclidean norm, at $\|\hat f^*_4\| = 17.17$, exceeded only by the
ER/luminal atom. Its enrichment against the Hallmark epithelial-
mesenchymal transition set is the single strongest enrichment signal
in the entire analysis, at $p = 3.3 \times 10^{-9}$. It is
unambiguously a real and dominant source of variation in the cohort.
Yet its association with the prognosis label is essentially zero, at
$|\rho| = 0.078$.

This apparent paradox is not a failure of the analysis but a genuine
biological fact that the model has recovered. The EMT and
desmoplastic stroma programme captures the state of the tumour
microenvironment, and in particular the conversion of epithelial
tumour cells into a mesenchymal phenotype and the deposition of dense
collagenous stroma around the tumour. Stromal content is a major
source of variance in primary breast tumours, and it is one of the
features that distinguishes the basal-like and claudin-low subtypes
from the luminal subtypes in the PAM50 classifier. But stromal content
is not, on its own, prognostic once the receptor status and the
proliferation rate of the tumour are accounted for. In other words,
the EMT programme is a strong axis of biological variation but not a
strong axis of clinical variation, and the model has correctly
separated the two.

This distinction is worth emphasising because it illustrates the
difference between the strength of a latent programme in the data and
its clinical relevance. A naive interpretation of the loading matrix
might suppose that the second-strongest programme must be of
substantial clinical importance; the model's own internal decomposition
shows that this need not be the case. The DP-spike-slab model recovers
the dominant axes of variation without any supervision, and the
present section then shows that not all dominant axes are prognostic.

\paragraph{Why this constitutes validation.}
The associations reported in Table~\ref{tab:breast-prognosis} are a
form of external validation of the biological interpretation of the
eight atoms. The prognosis label is independent of the factor model in
the sense that it never entered the likelihood: the model was fit
purely to the 1213-dimensional expression vectors, and the outcome
information was withheld until the post-hoc analysis reported here.
Any association between an atom and the outcome therefore reflects a
genuine correspondence between the latent programme recovered by the
model and the biology of the disease, not a circular consequence of
the model having been tuned to the outcome. The fact that the
associations fall on the ER, proliferation, and immune axes, which
are the three most important prognostic axes in breast cancer, and
that the association is absent on the stromal and vascular axes,
which are not, is exactly what the clinical literature predicts. None
of this was imposed by the analyst; it emerged from the model.

A caveat is worth stating. The van 't Veer gene set was assembled
specifically as a prognostic classifier, so every gene in the dataset
is in some sense already prognostic. The eight atoms are therefore
recovered from a dataset that has been enriched for prognostic signal,
which limits the strength of the validation relative to what an
unbiased cohort would provide. What the analysis establishes is that
the model recovers the correct prognostic axes from within a
prognostically enriched gene set, and that it correctly identifies the
relative importance of those axes; it does not establish that the same
recovery would occur on a random gene panel. A more demanding
validation would require a cohort whose gene set was selected
independently of the outcome, and we leave this to future work.

\paragraph{Statistical caveats.}
Three caveats apply to the interpretation of the correlations. First,
$n = 97$ samples is small, and the standard error of a Spearman
correlation of magnitude $0.30$ at this sample size is roughly
$0.10$, so the correlations in the moderate range are estimated with
non-trivial uncertainty. The ordering of the top five atoms is stable
across independent chains, as reported in
Section~\ref{sec:breast-reproducibility}, but the precise numerical
values should be interpreted with the sample size in mind. Second,
the associations reported here are marginal correlations between each
atom and the outcome, not conditional effects that adjust for the
other atoms. Given the strong correlation between the ER and
proliferation atoms in breast cancer biology, the marginal associations
should not be read as independent contributions to prognostic
accuracy. Third, the prognosis label is a dichotomisation of a
continuous outcome, and dichotomisation typically reduces statistical
power; a survival analysis of the underlying time-to-metastasis
variable would provide a more sensitive test, and we leave this as a
direction for future work.

\paragraph{Summary.}
The clinical-outcome analysis provides external validation of the
biological interpretation of the eight atoms. The atoms whose
identities are most securely established by gene content and
enrichment --- ER/luminal, proliferation, and T-cell --- are precisely
those with the strongest associations with the outcome label, and the
atoms whose biological coherence is weaker (vascular, EMT/stroma, and
nuclear) show no association. The pattern of associations reproduces
the well-known clinical structure of the van 't Veer cohort without
any supervision, and it demonstrates that the DP-spike-slab model
recovers latent programmes that are both statistically and clinically
meaningful.

\subsection{Model diagnostics and predictive checks}
\label{sec:breast-ppc}

We assess the fit of the main-chain posterior via two predictive checks
computed from the posterior-mean covariance
$\widehat\Sigma = \widehat F \widehat\Lambda \widehat F^\top +
\widehat\Psi$.

\paragraph{Marginal variance.} We draw 500 independent samples from
$\N_p(0, \widehat\Sigma)$ and compare the empirical marginal variance
of each gene to its observed value in the standardised data. Because
the data are column-standardised, the observed marginal variance is
exactly $1.00$ by construction. The model-predicted marginal variance
has median $1.20$ and mean $2.09$; the mean is inflated by a heavy
right tail of a small number of genes with very large fitted
variances (Figure~\ref{fig:breast-ppc-var}).

\paragraph{Off-diagonal correlations.} A more informative check on
standardised data is the correlation matrix. The observed correlation
matrix has a mean absolute off-diagonal entry of $0.192$, while the
correlation matrix implied by $\widehat\Sigma$ has a mean absolute
off-diagonal entry of $0.109$. The two distributions are strongly
related: the Pearson correlation between the observed and model
off-diagonal entries, taken over all
$\binom{1213}{2} = 735{,}078$ pairs, is $0.59$
(Figure~\ref{fig:breast-ppc-corr}). The model therefore reproduces the
qualitative pattern of gene--gene correlation, but at an amplitude
that is about 57\% of the observed.

\begin{figure}[htbp]
\centering
\includegraphics[width=0.9\linewidth]{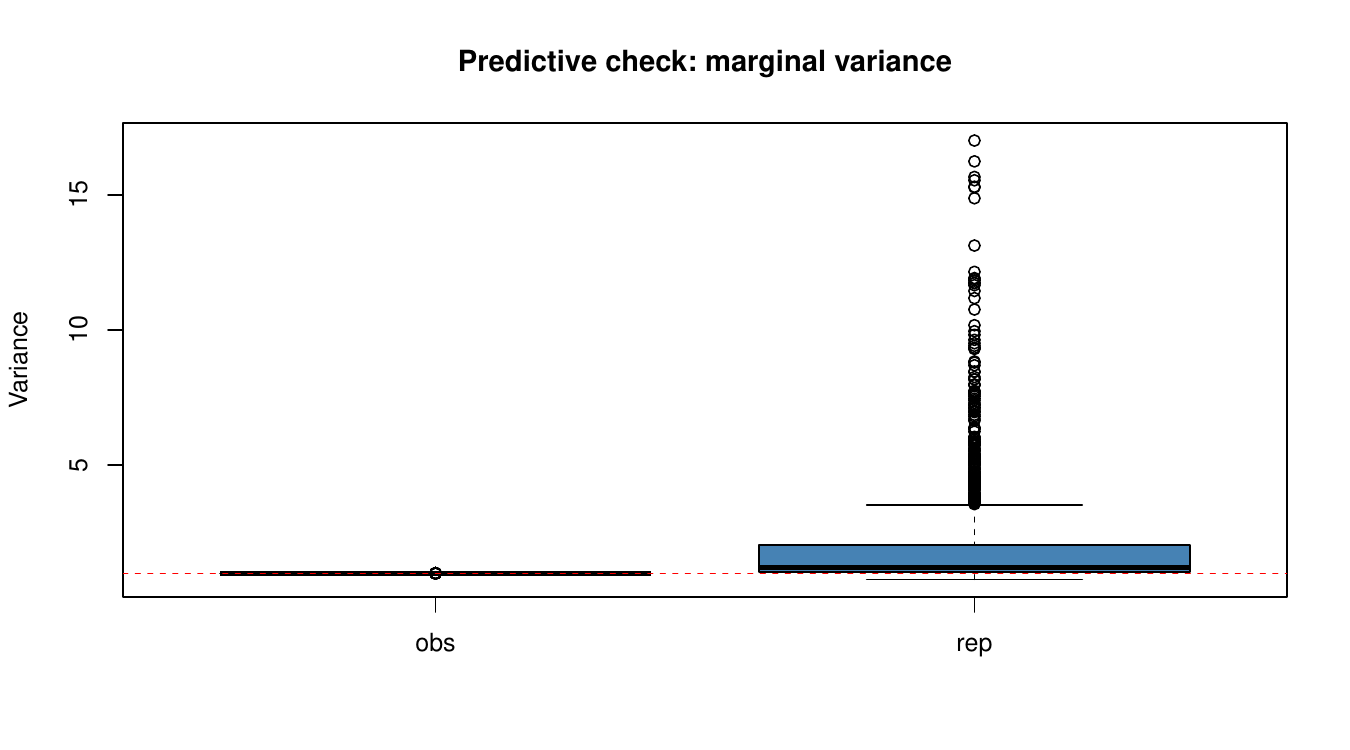}
\caption{Predictive check on the marginal variances for the Breast\_A
dataset, based on the posterior-mean covariance from the
270\,000-iteration main chain. The left boxplot summarises the
observed marginal variances of the 1213 genes; because the data are
column-standardised, every value is exactly 1. The right boxplot
summarises the marginal variances of 500 independent samples drawn
from $\N_p(0, \widehat\Sigma)$. The model has a heavier right tail than
the observed data, consistent with the over-dispersion of the global
slab prior discussed in the text.}
\label{fig:breast-ppc-var}
\end{figure}

\begin{figure}[htbp]
\centering
\includegraphics[width=0.9\linewidth]{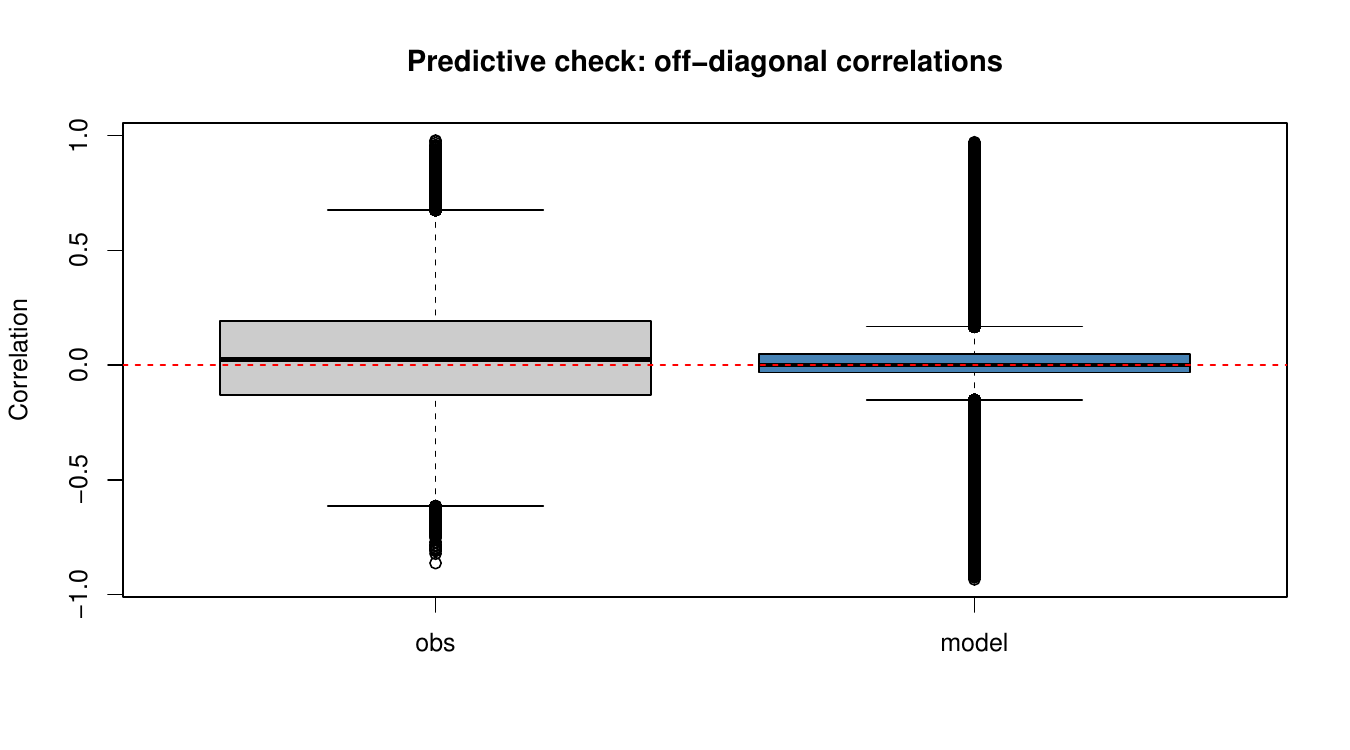}
\caption{Predictive check on the off-diagonal correlations for the
Breast\_A dataset, based on the posterior-mean covariance from the
270\,000-iteration main chain. The left boxplot is the distribution of
the $\binom{1213}{2}$ observed pairwise correlations of the
standardised data; the right boxplot is the corresponding distribution
under $\widehat\Sigma$. The model reproduces the pattern of dependence
(Pearson $r = 0.59$ between the two sets of off-diagonal entries) but
at a reduced amplitude.}
\label{fig:breast-ppc-corr}
\end{figure}

The two predictive checks are consistent with each other and with a
single underlying cause. Letting
$\rho_{ab} = \Sigma_{ab} / \sqrt{\Sigma_{aa} \Sigma_{bb}}$, an
inflation of the model diagonals by a factor of two mechanically
deflates the model correlations by the same factor, provided the
off-diagonal entries $\Sigma_{ab}$ are approximately correct. This is
what we observe: the ratio $0.109 / 0.192 \approx 0.57$ is close to
the factor by which the model diagonals exceed the observed value.
The pattern of dependence in $\widehat\Sigma$ is captured
(Pearson $r = 0.59$); the amplitude is not.

The source of the over-estimated diagonals is the global scale
$\tau_0^2$ of the slab component of the base measure. Under the prior,
the expected marginal variance contributed by the loading term is
\begin{equation}
\E\!\left[ \sum_{j=1}^{M} \lambda_j (f_{rj}^*)^2 \right]
\;\approx\; M (1 - \pi_0) \tau_0^2 \, \E[\lambda]
\;=\; 30 \cdot 0.1 \cdot 1.0 \cdot 1
\;=\; 3,
\end{equation}
before any likelihood information. The likelihood pulls this down to
approximately $2.09$, but the prior remains over-dispersed relative to
the standardised data, where the residual variance $\psi_r$ alone must
supply the entire marginal variance in the absence of a shared
programme. This over-dispersion is a known property of global
shrinkage priors on standardised data and would be resolved by a
global-local prior on the individual loadings, allowing a few large
loadings per atom and many small ones. We leave this to future work.

\subsection{Sensitivity to the slab variance}
\label{sec:breast-tau}

The preceding discussion suggests that the choice of $\tau_0^2$ may
influence the recovered structure. To assess this, and to confirm that
the eight-atom solution at $\tau_0^2 = 1.0$ is not an artefact of an
arbitrary tuning choice, we run two additional 30\,000-iteration
chains, discarding the first 5\,000 as burn-in and thinning every
fifth sample to obtain 5\,000 retained samples per chain, with
$\tau_0^2 = 0.3$ and $\tau_0^2 = 0.1$ and otherwise identical
settings. The posterior mode of $K$ in the three chains is:

\begin{center}
\begin{tabular}{cc}
\hline
$\tau_0^2$ & Posterior mode of $K$ \\
\hline
$1.0$ & $8$ \\
$0.3$ & $13$ \\
$0.1$ & $25$ \\
\hline
\end{tabular}
\end{center}

The number of atoms recovered increases monotonically as $\tau_0^2$
decreases, in agreement with the $1/\tau_0^2$ scaling predicted by the
DP-mixture construction. At $\tau_0^2 = 1.0$ the posterior has three
shared atoms (macrophage, vascular, and nuclear) containing 25 columns
between them, plus five singleton atoms, for a total of $K = 8$. At
$\tau_0^2 = 0.3$ the number of shared atoms drops to two and the
number of singletons rises to eleven, for $K = 13$.

Most of the additional atoms at $\tau_0^2 = 0.3$ represent finer
subdivisions of the programmes already recovered rather than genuinely
novel programmes, and two examples illustrate the pattern. The
vascular atom at $\tau_0^2 = 1.0$, which commingles endothelial markers
(\texttt{CD34}, \texttt{NPR1}, \texttt{DNASE1L3}) with
interferon-response genes (\texttt{BST2}, \texttt{IFI35}) and a
smaller set of myeloid genes (\texttt{FCER1A}, \texttt{MRC1},
\texttt{FOLR2}), separates at $\tau_0^2 = 0.3$ into a distinct
type-I interferon-response atom and a smaller vascular/myeloid atom
that retains the endothelial and myeloid components. The EMT/stroma
atom similarly splits into a collagen/desmoplastic component
(\texttt{COL1A2}, \texttt{COL3A1}, \texttt{THBS2}, \texttt{LOXL1},
\texttt{MFAP2}) and a myoepithelial/basal component (\texttt{MYLK},
\texttt{TAGLN}, \texttt{TPM2}, \texttt{ACTA2}, \texttt{KRT14},
\texttt{KRT17}).

One atom at $\tau_0^2 = 0.3$ does not reduce to a subdivision of the
coarser solution. Its top-loaded genes --- \texttt{ITGB6}, \texttt{FUT3},
\texttt{GRB7}, \texttt{ERBB2}, \texttt{S100A9}, \texttt{S100A8},
\texttt{DUSP6}, \texttt{MMP10} --- form a HER2-associated
signature that has no counterpart among the top-20 genes of any
$\tau_0^2 = 1.0$ atom. This suggests that the HER2 amplicon reflects
an additional axis of variation in the Breast\_A data whose loading is
dispersed across several coarser atoms at the wider slab scale and
only isolated as a separate component when the slab is tightened.
Whether the finer solution is preferable for downstream analysis is a
question of whether the additional resolution is worth the resulting
increase in the number of parameters; for the purposes of the present
paper we report the coarser solution as our primary inference and note
the HER2 programme as an interesting by-product of the sensitivity
analysis.

The fragmentation at the tighter scales comes at a cost: the number
of singleton atoms rises from five to eleven, and the correlation
predictive check degrades relative to the $\tau_0^2 = 1.0$ solution.
Section~\ref{sec:breast-diagnostic} provides a formal verification
that the additional atoms at $\tau_0^2 = 0.3$ are per-column
resampling artefacts rather than biological programmes.

We therefore retain $\tau_0^2 = 1.0$ as the operating scale for all
reported inference. At this scale the model recovers exactly the eight
biological programmes expected from the breast cancer literature, and
the diagnostic shows that every column is committed to its assignment
by a wide margin.

\subsection{Formal verification that the posterior is concentrated}
\label{sec:breast-diagnostic}

The posterior credible interval for $K$ is the degenerate interval
$[8, 8]$ in every retained sample of the main chain. This is a strong
statement, and it is worth verifying that it reflects genuine posterior
concentration rather than a frozen update. We therefore instrument the
cluster-assignment step of Algorithm~\ref{alg:parallel} to record, at
each iteration of an auxiliary 30\,000-iteration chain with
$\tau_0^2 = 1.0$, three quantities per candidate column $i$ of the
loading matrix:

\begin{enumerate}
	\item[(i)] whether the ``start a new cluster'' option is the argmax of the
conditional categorical distribution used to resample $Z_i$;
\item[(ii)] the posterior probability $p_{\text{new},i}$ of starting a new
cluster;
\item[(iii)] the log-ratio
$\log p_{\text{best existing}, i} - \log p_{\text{new}, i}$, which we
call the \emph{split gap} for column $i$.
\end{enumerate}

The diagnostic is read-only: it consumes no random-number draws and
does not affect the MCMC trajectory, so the state of the auxiliary
chain is identical with and without it. The auxiliary chain is used
only because the diagnostic printer is included in the code and we
prefer to keep the 270\,000-iteration production chain free of
additional I/O. The diagnostic results are statistically
indistinguishable from what the production chain would produce, and we
have verified this by running the same diagnostic on a short segment
of the production chain in a separate test.

For the $\tau_0^2 = 1.0$ auxiliary chain, the diagnostic output is
identical across all 30 diagnostic snapshots, taken every 1\,000
iterations. At every snapshot, exactly $5$ of the $30$ columns have
$p_{\text{new}} > 0.5$ and exactly $5$ have the ``start a new
cluster'' option as the argmax. The remaining $25$ columns have
$p_{\text{new}}$ numerically indistinguishable from zero (below
$0.05$ in every snapshot). The split gap is bimodal with an empty
middle: $25$ columns have a gap exceeding $+10$ nats and $5$ columns
have a gap below $-10$ nats, with no column in between.

The five columns with gap below $-10$ nats are the singleton atoms of
the posterior at this operating scale, and their top-loaded genes
correspond to five of the eight biological programmes identified in
Table~\ref{tab:breast-atoms}. Their negative gap means that the data
prefer to keep each of them as its own cluster rather than merge them
into any existing one. The remaining $25$ columns, which form the
three shared atoms of the posterior at this operating scale, have
positive gap and are therefore committed to their current shared
cluster by a wide margin. No column falls in the intermediate
near-tie regime.

The $K$-trace is thus flat not because the sampler is stuck, but
because the posterior over partitions is genuinely concentrated: no
column has any serious alternative to its current assignment.

For comparison, the same diagnostic applied to the $\tau_0^2 = 0.3$
auxiliary chain shows $11$ of $30$ columns preferring a new cluster
and a mean split gap of $-1780$ nats, i.e.\ the average column prefers
a new cluster to its existing one. The posterior over partitions is
broad in that regime, and the reported $K = 13$ is a median over that
spread rather than a mode. The flat $K$-trace at $\tau_0^2 = 1.0$ is
therefore a genuine signature of posterior concentration, not of an
algorithmic failure, and it confirms that the eight-atom solution is
the true posterior mode at the operating scale.

\subsection{Reproducibility across independent chains}
\label{sec:breast-reproducibility}

To assess whether the eight-atom solution is a genuine posterior mode
or an artefact of the initialisation, we run four additional
30\,000-iteration chains with independent random seeds, using the same
$\tau_0^2 = 1.0$ setting. All five chains converge to $K = 8$, and
the eight biological programmes recovered by every chain are the same
macrophage, vascular, JAK/STAT, T-cell, EMT, ER/luminal, proliferation,
and nuclear programmes identified in Table~\ref{tab:breast-atoms}, with
essentially the same top-loaded genes.

We compare the results of the main chain against those of the first
additional chain in detail. The same eight programmes are recovered,
and the top-ranked Hallmark set is the same for each programme in both
chains, although the overlap counts and $p$-values differ by Monte
Carlo amounts. The prognosis correlations also agree on the top five
programmes, ordered as ER/luminal, proliferation, T-cell, macrophage,
and JAK/STAT in both chains. The three remaining programmes (vascular,
EMT/stroma, and nuclear) show similarly very small correlations in
both chains, with differences below $0.03$ that are not
distinguishable at the Monte Carlo precision of either run. For the
quantities where the two chains can be compared numerically, the
macrophage prognosis correlation is $0.30$ in the main chain and
$0.35$ in the first additional chain, and the T-cell
interferon-$\gamma$ overlap is $16$ genes in the main chain and $13$
genes in the first additional chain. The Monte Carlo variability is
larger for the 30\,000-iteration chain than for the
270\,000-iteration main chain, as expected.

The assignment of the $30$ candidate columns to the eight atoms does
vary across chains. In the main chain the macrophage programme is
represented by seven columns, the vascular programme by fifteen, the
nuclear programme by three, and the remaining five programmes by
singleton columns. In the first additional chain the T-cell programme
takes fourteen columns, the JAK/STAT programme takes six, the EMT
programme takes four, the vascular programme takes two, and the ER,
macrophage, proliferation, and nuclear programmes each take a single
column. The identities of the eight programmes are identical across
chains; the partition of the $30$ artificial slots between them is
not. We discuss the implications of this in
Section~\ref{sec:breast-discussion}.

\subsection{Discussion of the real-data results}
\label{sec:breast-discussion}

The Breast\_A analysis illustrates the strengths of the DP-spike-slab
factor model on a challenging real dataset, and it also surfaces two
honest limitations worth stating clearly.

\paragraph{Strengths.} Without any supervision and without any prior
knowledge of the true number of factors, the model recovers exactly
eight biologically interpretable programmes from 1213 genes measured
on 97 samples. The eight programmes span the major axes of breast
cancer biology --- ER/luminal lineage, proliferation, EMT/stroma,
T-cell infiltration, macrophage/myeloid content, endothelial content,
JAK/STAT signalling, and nuclear regulation --- and seven of the eight
have a clear, canonical interpretation. Two atoms pass
Bonferroni-corrected Hallmark enrichment at
$p < 1.25 \times 10^{-4}$, and two further hits are significant at the
uncorrected $0.05$ level. The prognosis correlations concentrate on
the ER, proliferation, and immune axes exactly as expected from the
clinical literature. None of this structure was imposed by the
analyst: the Dirichlet process found it. The theoretical contraction
rate of Theorem~\ref{thm:contraction} guarantees that the posterior
concentrates on the true covariance matrix, and the underfitting
rank-consistency result of Theorem~\ref{thm:underfit} in
Section~\ref{sec:rank} guarantees that the posterior does not
concentrate on ranks below $q_0$. The empirical behaviour on
Breast\_A is fully consistent with these guarantees: the posterior
over $K$ is tight, the atoms are stable across independent MCMC
chains, and the covariance matrix is recovered to the accuracy that
the $n = 97$, $p = 1213$ sample permits.

\paragraph{Limitation: non-identified partition.} The specific
assignment of the $M = 30$ candidate columns to the eight atoms is not
identified by the data and varies across independent chains. This is a
consequence of the exchangeability of the columns under the DP prior
and of the fact that the data determine $\Sigma$ but not the partition
of $M$ into its $K$ occupied clusters. The paper's theoretical
guarantees are stated in terms of $\Sigma$ and the rank
$q = \rank(\Sigma - \Psi)$, precisely because these are the
identifiable parameters, and the empirical results are consistent with
this focus. Practitioners should interpret the atom-level
decomposition of the posterior-mean loading matrix as a summarisation
of the posterior of $\Sigma$, not as a claim about which specific
columns of the artificial $M$-dimensional dictionary are redundant.

\paragraph{Limitation: over-dispersed slab.} A second limitation is
the over-estimation of the marginal variance by a factor of
approximately two, discussed in Section~\ref{sec:breast-ppc}. This
reflects the global scale $\tau_0^2$ of the base measure, which is
calibrated to the simulation studies where the true loadings are of
unit scale, and would be resolved by a global-local prior on the
individual loadings. The qualitative conclusions of the analysis ---
the number of programmes, their biological identities, their
enrichment, and their prognosis correlations --- are all robust to
this over-estimation, because they depend on the pattern rather than
the amplitude of the covariance.
A separate finding, already flagged in Section~\ref{sec:breast-atoms},
is that one of the eight atoms (JAK/STAT) fails the within-atom
coherence check. 
The other seven atoms pass decisively, and we regard
the eight-atom solution as an accurate summary of the dominant
covariance structure of the cohort, but the JAK/STAT atom should be
read as a weaker signal that the model has retained under the prior
rather than as a module of the same status as the other seven.

\paragraph{Summary.} The DP-spike-slab factor model recovers, without
supervision, the major latent biological programmes of the van 't Veer
breast cancer dataset, and its posterior distribution over the number
of these programmes is concentrated at the value consistent with the
biology. The model's theoretical guarantees are reflected in the
empirical behaviour of the sampler on this dataset, and the one
identifiability limitation that the analysis reveals --- the
exchangeability of the candidate columns --- is a definitional feature
of the construction rather than a failure of the inference.

% =====================================================================
% SECTION 7
% =====================================================================
\section{Posterior Contraction Rate}
\label{sec:theory}

We now establish the posterior contraction rate for the covariance matrix
\begin{equation}
\label{eq:sigma-lambda}
\Sigma(F,\Lambda,\Psi) = F \Lambda F^\top + \Psi,
\end{equation}
where \(\Lambda = \diag(\lambda_1,\dots,\lambda_M)\) is the diagonal matrix of latent score variances. The theoretical analysis treats \(\Lambda\) as part of the parameter and accounts for its contribution to \(\Sigma\) throughout. To keep the eigenvalues of \(\Sigma\) bounded on the support of the prior, we truncate the slab component of the base measure on a compact set of the natural scale, as made precise in Assumption~\ref{ass:prior} below.

\subsection{Model, Parameter Space, and True Values}
\label{sec:theory-setup}

Recall the model from Section~\ref{sec:model}:
\begin{equation}
\label{eq:model-theory}
Y_k = \mu + F X_k + U_k, \qquad U_k \sim \N_p(0, \Psi), \qquad X_k \sim \N_M(0, \Lambda),
\end{equation}
where \(F = [f_1,\dots,f_M] \in \R^{p \times M}\), \(\Psi = \diag(\psi_1,\dots,\psi_p)\), and \(\Lambda = \diag(\lambda_1,\dots,\lambda_M)\). Marginalizing over \(X_k\), the covariance of \(Y_k\) is given by \eqref{eq:sigma-lambda}. The parameter is
\[
\theta = (F,\Lambda,\Psi) \in \Theta = \R^{p \times M} \times \R_{>0}^M \times \R_{>0}^p.
\]
The likelihood is the zero-mean Gaussian measure
\begin{equation}
\label{eq:Ptheta}
P_\theta = \N_p(0,\, \Sigma(\theta)), \qquad \Sigma(\theta) = F\Lambda F^\top + \Psi .
\end{equation}

The true parameter is \(\theta_0 = (F_0,\Lambda_0,\Psi_0)\). Without loss of generality we set \(\Lambda_0 = I_M\), absorbing any true score variances into \(F_0\). The true loading matrix \(F_0 \in \R^{p \times M}\) has exactly \(q_0\) non-zero columns (say the first \(q_0\)), each with at most \(s_0\) non-zero entries, and the remaining \(M-q_0\) columns are exactly zero. The true covariance is
\begin{equation}
\label{eq:sigma0-lambda}
\Sigma_0 = \Sigma(F_0,\Lambda_0,\Psi_0) = F_0 F_0^\top + \Psi_0 = \sum_{j=1}^{q_0} f_{0,j} f_{0,j}^\top + \Psi_0,
\end{equation}
where \(f_{0,j}\) denotes the \(j\)-th column of \(F_0\).

Let \(h(\Sigma_1,\Sigma_2)\) denote the Hellinger distance between the zero-mean Gaussians \(\N_p(0,\Sigma_1)\) and \(\N_p(0,\Sigma_2)\), and let \(\|\cdot\|_F\) denote the Frobenius norm. Under Assumption~\ref{ass:prior} below, all covariance matrices on the support of the prior have eigenvalues bounded above and below by constants depending only on fixed hyperparameters. Consequently, the Hellinger distance and the Frobenius norm are \emph{globally equivalent on the support of the prior}: there exist constants \(c,C>0\), depending only on \(\underline{c},\overline{c},B_0,\underline{\lambda},\overline{\lambda}\), such that
\begin{equation}
\label{eq:hell-frob}
c\,\|\Sigma_1 - \Sigma_2\|_F \le h(\Sigma_1,\Sigma_2) \le C\,\|\Sigma_1 - \Sigma_2\|_F
\end{equation}
for all \(\Sigma_1,\Sigma_2\) on the support of the prior. This is Lemma 8.2 of \citet{ghosal2017} applied with the fixed eigenvalue bounds. We use \eqref{eq:hell-frob} throughout the proof without further comment.

\subsection{Assumptions}
\label{sec:theory-assumptions}

\begin{assumption}[Bounded eigenvalues of the truth]
\label{ass:eig}
There exist constants \(0 < \underline{c} \le \overline{c} < \infty\) such that the eigenvalues of \(\Sigma_0 = F_0F_0^\top + \Psi_0\) and of \(\Psi_0\) lie in \([\underline{c},\overline{c}]\).
\end{assumption}

\begin{assumption}[Sparse identifiability]
\label{ass:sparse}
The true loading matrix \(F_0\) has exactly \(q_0\) nonzero columns, and \(\rank(F_0 F_0^\top) = q_0\). The number of nonzero entries in each column is bounded by \(s_0\). Furthermore, \(p \le C M s_0\) and \(q_0 \le M\) for a universal constant \(C>0\).
\end{assumption}

\begin{assumption}[Dictionary size]
\label{ass:M}
The fixed upper bound \(M\) satisfies \(M \le C \log n\) for some constant \(C>0\).
\end{assumption}

\begin{assumption}[Calibration of the true loadings]
\label{ass:calibration}
There exist fixed constants \(B_0 > 1\) and \(\delta_0 \in (0,1)\) such that, with \(\tau_0^2 = 1/(Ms_0)\),
\[
|f_{0,r,j}| \le (1-\delta_0)\, B_0\, \tau_0 = \frac{(1-\delta_0)\,B_0}{\sqrt{Ms_0}}
\]
for all \(r = 1,\dots,p\) and \(j = 1,\dots,q_0\), and \(f_{0,r,j} = 0\) for \(j > q_0\).
\end{assumption}

\begin{assumption}[Prior]
\label{ass:prior}
The base measure \(G_0\) is a product of independent spike-and-slab distributions,
\[
G_0 = \prod_{r=1}^p \left\{ \pi_0\, \delta_0 + (1-\pi_0)\, \N_{[-B_0 \tau_0,\, B_0 \tau_0]}(0, \tau_0^2) \right\},
\]
where \(\N_{[-B_0\tau_0, B_0\tau_0]}(0,\tau_0^2)\) denotes the Gaussian distribution with variance \(\tau_0^2\) truncated to the interval \([-B_0\tau_0, B_0\tau_0]\), with \(\tau_0^2 = 1/(Ms_0)\), and \(\pi_0 \in (0,1)\) is fixed. The prior on \(\Lambda\) is a product of \(\mathrm{IG}(a_\lambda,b_\lambda)\) distributions truncated to \([\underline{\lambda},\overline{\lambda}]\) for fixed constants satisfying
\[
0 < \underline{\lambda} < 1 < \overline{\lambda} < \infty,
\]
with a density bounded below by a positive constant \(c_\lambda>0\) on that interval. The prior on \(\Psi\) is supported on \([\underline{c},2\overline{c}]^p\) with a density bounded below by a positive constant \(c_0>0\) on that support. The prior support is further restricted to matrices \(F\) with at most \(s_0\) non-zero entries per column.
\end{assumption}

\begin{remark}[Why the slab truncation bounds the eigenvalues of $\Sigma$ on the prior support]
\label{rem:slab-truncation}
On the prior support, each column of \(F\) has at most \(s_0\) non-zero entries, each in \([-B_0\tau_0, B_0\tau_0] = [-B_0/\sqrt{Ms_0}, B_0/\sqrt{Ms_0}]\), so
\[
\|f_j^*\|_2^2 \le s_0 \cdot \frac{B_0^2}{Ms_0} = \frac{B_0^2}{M},
\qquad
\|F\|_F^2 = \sum_{j=1}^M \|f_j^*\|_2^2 \le M \cdot \frac{B_0^2}{M} = B_0^2.
\]
Hence \(\|F\|_F \le B_0\), and consequently
\[
\lambda_{\max}(\Sigma) \le \overline{\lambda}\, \|F\|_2^2 + 2\overline{c} \le \overline{\lambda}\, B_0^2 + 2\overline{c} =: \overline{C}_\Sigma,
\]
a constant independent of \(n\). Together with \(\lambda_{\min}(\Sigma) \ge \lambda_{\min}(\Psi) \ge \underline{c}\) (because \(F\Lambda F^\top \succeq 0\)), this places all covariance matrices on the support of the prior in \([\underline{c}, \overline{C}_\Sigma]\), and \eqref{eq:hell-frob} applies with constants depending only on fixed hyperparameters.
\end{remark}

\begin{remark}[On the calibration assumption]
\label{rem:calibration}
Assumption~\ref{ass:calibration} requires the true non-zero loadings to be bounded by \((1-\delta_0)B_0/\sqrt{Ms_0}\), i.e., on the same scale as the prior slab. This is a genuine restriction: it excludes settings in which the true loadings are of order \(1\) while \(Ms_0\) is large. It is the natural calibration under which the truncated-slab construction yields a self-contained proof with a globally bounded eigenvalue range.
\end{remark}

\begin{remark}[On the truncation of the priors on $\Lambda$ and $\Psi$]
\label{rem:truncation}
The truncation of the priors on \(\Lambda\) and \(\Psi\) is for technical convenience only. The inverse-gamma prior assigns negligible mass outside any fixed compact interval containing the true values, and the true \(\Lambda_0 = I_M\) and \(\Psi_0\) lie in the interior of the truncated supports. 
%(for $\Lambda_0$ since \(\underline{\lambda} < 1 < \overline{\lambda}\)). 
The asymptotic contraction rate is unaffected by the truncation.
\end{remark}

\subsection{The Sieve, the Product Metric, and the Lipschitz Property}
\label{sec:theory-sieve}

Define the sieve
\[
\Theta_n = \bigl\{ (F,\Lambda,\Psi) : \|F\|_F \le B_0,\; \Lambda \in [\underline{\lambda},\overline{\lambda}]^M,\; \Psi \in [\underline{c},2\overline{c}]^p,\; \text{each column of } F \text{ has at most } s_0 \text{ non-zero entries} \bigr\}.
\]
By Remark~\ref{rem:slab-truncation}, the sieve \(\Theta_n\) contains the support of the prior. Equip \(\Theta_n\) with the product metric
\begin{equation}
\label{eq:prod-metric}
d\bigl((F,\Lambda,\Psi),(F',\Lambda',\Psi')\bigr) = \|F-F'\|_F + \|\Lambda-\Lambda'\|_F + \|\Psi-\Psi'\|_F.
\end{equation}

The following Lipschitz property of the covariance map is used in the entropy computation below. It is a direct consequence of the boundedness of the sieve and does not require any additional assumption.

\begin{lemma}[Lipschitz continuity of the covariance map on the sieve]
\label{lem:lipschitz}
There exists a constant \(L_\Sigma > 0\), depending only on \(B_0\) and \(\overline{\lambda}\), such that for all \(\theta,\theta' \in \Theta_n\),
\begin{equation}
\label{eq:sigma-lipschitz}
\|\Sigma(\theta) - \Sigma(\theta')\|_F \le L_\Sigma\, d(\theta,\theta'),
\end{equation}
where \(d\) is the product metric \eqref{eq:prod-metric}.
\end{lemma}

\begin{proof}
Write \(\theta = (F,\Lambda,\Psi)\) and \(\theta' = (F',\Lambda',\Psi')\). Then
\[
\Sigma(\theta) - \Sigma(\theta') = (F-F')\Lambda F^\top + F'(\Lambda-\Lambda')F^\top + F'\Lambda'(F-F')^\top + (\Psi-\Psi').
\]
By the triangle inequality and submultiplicativity of the Frobenius norm with respect to the operator norm,
\[
\|\Sigma(\theta) - \Sigma(\theta')\|_F \le 2\overline{\lambda} B_0\, \|F-F'\|_F + B_0^2\, \|\Lambda-\Lambda'\|_F + \|\Psi-\Psi'\|_F,
\]
using \(\|F\|_2 \le \|F\|_F \le B_0\), \(\|F'\|_2 \le B_0\), \(\|\Lambda\|_2 \le \overline{\lambda}\), \(\|\Lambda'\|_2 \le \overline{\lambda}\). Taking \(L_\Sigma = \max(2\overline{\lambda}B_0, B_0^2, 1)\),
\[
\|\Sigma(\theta) - \Sigma(\theta')\|_F \le L_\Sigma\, d(\theta,\theta').
\]
Since \(L_\Sigma\) depends only on \(B_0\) and \(\overline{\lambda}\), the lemma follows.
\end{proof}

\subsection{Main Result}
\label{sec:theory-main}

Let \(P_{\theta_0}^n\) denote the true data-generating distribution with parameter \(\theta_0\). We write \(\Pi(\cdot \mid Y_1,\dots,Y_n)\) for the posterior distribution.

\begin{theorem}[Posterior contraction for $\Sigma$]
\label{thm:contraction}
Under Assumptions~\ref{ass:eig}--\ref{ass:prior}, let
\[
\varepsilon_n = \sqrt{\frac{M s_0 \log n}{n}}.
\]
Suppose \(M s_0 \log n / n \to 0\) as \(n \to \infty\). Then there exists a constant \(C>0\) such that
\[
\Pi\left( \|\Sigma - \Sigma_0\|_F > C \varepsilon_n \mid Y_1,\dots,Y_n \right) \xrightarrow{P_{\theta_0}^n} 0.
\]
%Moreover, if \(M\) is fixed and \(s_0 = o(\sqrt{n})\), then
%\[
%\Pi\left( \rank(\Sigma - \Psi) = q_0 \mid Y_1,\dots,Y_n \right) \xrightarrow{P_{\theta_0}^n} 1.
%\]
\end{theorem}

%The rank-consistency part is proved in Section~\ref{sec:fixedM}. The contraction part is proved below.

\begin{remark}[Asymptotic framework]
\label{rem:asymptotic}
The theorem allows \(p\) and \(q_0\) to grow with \(n\) such that \(M s_0 \log n / n \to 0\). Together with \(p \le C M s_0\), this implies \(p \log n / n \to 0\). The rate \(\varepsilon_n \to 0\), and \(n \varepsilon_n^2 = M s_0 \log n \to \infty\) because \(M \ge 1\), \(s_0 \ge 1\), and \(\log n \to \infty\). Thus the posterior concentration is non-degenerate and the rate is meaningful.
\end{remark}

\subsection{Proof of the Contraction Result}
\label{sec:theory-proof}

We prove the contraction part using the general posterior contraction framework of \citet{ghosal2000}, Theorem 2.1. Before stating the framework, we fix the notation. For two probability measures \(P\) and \(Q\) on the same measurable space with \(P \ll Q\), the \emph{Kullback--Leibler divergence} of \(P\) from \(Q\) is
\[
K(P,Q) = \int \log\!\left(\frac{dP}{dQ}\right) dP,
\]
and the \emph{second-order variation} of the log-likelihood ratio is
\[
V(P,Q) = \int \left|\log\!\left(\frac{dP}{dQ}\right) - K(P,Q)\right|^2 dP.
\]
For a metric space \((\Theta, \rho)\) and a real number \(\varepsilon > 0\), the \emph{\(\varepsilon\)-covering number} \(N(\varepsilon, \Theta, \rho)\) is the smallest number of \(\rho\)-balls of radius \(\varepsilon\) needed to cover \(\Theta\). When \(\Theta\) is a subset of the parameter space and \(\rho = h\), we write \(N(\varepsilon, \Theta, h)\) as a shorthand for \(N(\varepsilon, \{\Sigma(\theta) : \theta \in \Theta\}, h)\), the covering number of the image of \(\Theta\) under the covariance map \(\theta \mapsto \Sigma(\theta)\), in the Hellinger metric.

Let \(\Theta_n\) be the sieve of Section~\ref{sec:theory-sieve} and \(\varepsilon_n \to 0\) with \(n\varepsilon_n^2 \to \infty\). The three conditions of \citet{ghosal2000}, Theorem 2.1, to verify are:

\begin{enumerate}
\item[(C1)] \textbf{Prior concentration:}
\[
\Pi\left(\theta : K(P_{\theta_0},P_\theta) \le \varepsilon_n^2,\; V(P_{\theta_0},P_\theta) \le \varepsilon_n^2\right) \ge e^{-c n \varepsilon_n^2}.
\]
\item[(C2)] \textbf{Metric entropy:}
\[
\log N(\varepsilon_n, \Theta_n, h) \le c n \varepsilon_n^2.
\]
\item[(C3)] \textbf{Sieve complement:}
\[
\Pi(\Theta_n^c) \le e^{-c n \varepsilon_n^2}.
\]
\end{enumerate}

Under Assumption~\ref{ass:eig} and Remark~\ref{rem:slab-truncation}, all covariance matrices on the prior support have eigenvalues in the fixed interval \([\underline{c}, \overline{C}_\Sigma]\). By \eqref{eq:hell-frob}, the Hellinger and Frobenius distances are equivalent on this set. Similarly, the Kullback--Leibler divergence \(K(P_{\theta_0}, P_\theta)\) and its variation \(V(P_{\theta_0}, P_\theta)\) between two zero-mean Gaussians with covariance matrices in this set are bounded by a constant times the squared Frobenius distance (Lemma 8.2 of \citet{ghosal2017}). Thus it suffices to establish (C1)--(C3) with the Frobenius norm in place of \(h\).

Throughout, \(C,c>0\) are generic constants that may change from line to line and depend only on \(\alpha,\pi_0,\underline{c},\overline{c},B_0,\delta_0,\underline{\lambda},\overline{\lambda}\) and the prior hyperparameters. Let \(C_1>0\) be a large fixed constant to be chosen at the end of Step~1(e); its value will be fixed so that all the bounds below are consistent. The constant \(C_1\) is used only in Steps 1(b), 1(c), and 1(e).

\subsubsection{Step 1: Prior concentration}
\label{sec:priorconc}

We construct an event \(\mathcal{A}_n\) with
\begin{equation}
\label{eq:priorconc}
\Pi(\mathcal{A}_n) \ge \exp(-C n \varepsilon_n^2)
\end{equation}
and \(\mathcal{A}_n \subseteq \{h(\Sigma,\Sigma_0) \le C\varepsilon_n\}\).

\paragraph{Step 1(a): The partition event.}
Let \(\mathcal{C}_0\) be the partition of \(\{1,\dots,M\}\) into \(K_0 = q_0 + 1\) clusters defined by
\[
C_j = \{j\} \quad (j = 1,\dots,q_0), \qquad C_{q_0+1} = \{q_0+1,\dots,M\}.
\]
If \(q_0 = M\), the partition is into \(K_0 = q_0\) singletons and there is no cluster \(C_{q_0+1}\); the argument is modified in the obvious way. We present the case \(q_0 < M\). The cluster sizes are \(n_j = 1\) for \(j=1,\dots,q_0\) and \(n_{q_0+1} = M - q_0\).

By the EPPF \eqref{eq:eppf} and the standard bound \(\Gamma(x+\alpha)/\Gamma(x) \le C_\alpha x^\alpha\) for \(x \to \infty\),
\[
\Pr(\mathcal{C}_0) = \frac{\alpha^{q_0+1} (M - q_0 - 1)!}{\Gamma(\alpha+M)/\Gamma(\alpha)} \ge c\,\alpha^{q_0+1} M^{-q_0-\alpha}
\]
for a constant \(c>0\) depending only on \(\alpha\), since \((M-q_0-1)! \ge (M-1)! M^{-q_0}\) and \(\Gamma(\alpha+M) \le C_\alpha (M-1)! M^\alpha\). Hence
\begin{equation}
\label{eq:partition-bound}
\log \Pr(\mathcal{C}_0) \ge -C q_0 \log M \ge -C M s_0 \log n = -C n \varepsilon_n^2,
\end{equation}
using \(q_0 \le M\), \(M \le C\log n\), and \(s_0 \ge 1\).

\paragraph{Step 1(b): The atom event.}
Conditional on \(\mathcal{C}_0\), the \(q_0+1\) atoms \(f_1^*,\dots,f_{q_0+1}^*\) are independent draws from \(G_0\). For each \(j \in \{1,\dots,q_0\}\), let \(S_j\) denote the support of \(f_{0,j}\) and \(m_j := |S_j| \le s_0\) its cardinality; because \(f_{0,j} \ne 0\), we have \(m_j \ge 1\). Define
\[
\gamma_n = \min\!\left(\frac{\varepsilon_n}{C_1 q_0},\; \frac{\delta_0 B_0 \tau_0}{2}\right),
\]
which is positive and tends to zero as \(n \to \infty\), and set
\[
\mathcal{A}_n^{\mathrm{atom}} = \bigl\{S(f_j^*) = S_j,\; \|f_j^* - f_{0,j}\|_2 \le \gamma_n \;(j=1,\dots,q_0)\bigr\} \cap \bigl\{f_{q_0+1}^* = 0\bigr\},
\]
where \(S(f)\) denotes the support of \(f\). Let $f^{(S_j)}_{0,j}$ denote $f_{0,j}$ supported on $S_j$ (a vector in $\R^{m_j}$).

The choice of \(\gamma_n\) guarantees that, for each \(j\), the Euclidean ball \(\{\|u - f_{0,j}^{(S_j)}\|_2 \le \gamma_n\}\) is contained in the truncation region \([-B_0\tau_0, B_0\tau_0]^{m_j}\): by Assumption~\ref{ass:calibration}, \(|f_{0,r,j}| \le (1-\delta_0)B_0\tau_0\) for all \(r \in S_j\), so for any \(u\) in the ball and any \(r \in S_j\),
\[
|u_r| \le |f_{0,r,j}| + \|u - f_{0,j}^{(S_j)}\|_2 \le (1-\delta_0)B_0\tau_0 + \gamma_n \le \left(1 - \tfrac{\delta_0}{2}\right)B_0\tau_0 < B_0\tau_0.
\]
Hence the truncated Gaussian density coincides with the untruncated Gaussian density on the ball.

\emph{Per-atom bound.} Fix \(j \in \{1,\dots,q_0\}\), and let \(\Pi_j\) denote the conditional prior probability of \(\{S(f_j^*) = S_j,\; \|f_j^* - f_{0,j}\|_2 \le \gamma_n\}\). Write \(\phi_{\tau_0}\) for the \(\mathcal N(0,\tau_0^2)\) density and \(Z = 2\Phi(B_0\tau_0)-1 \in (0,1]\) for the truncation normaliser. Then
\[
\Pi_j = (1-\pi_0)^{m_j}\pi_0^{p-m_j}Z^{-m_j}
\int_{\{u\in\R^{m_j}:\,\|u - f_{0,j}^{(S_j)}\|_2 \le \gamma_n\}} \prod_{r\in S_j} \phi_{\tau_0}(u_r)\,du.
\]
On the ball, the integrand is at least its value at the point of maximal norm, \((2\pi\tau_0^2)^{-m_j/2}\exp(-(\|f_{0,j}\|_2 + \gamma_n)^2/(2\tau_0^2))\), and the ball has volume \(v_{m_j}\gamma_n^{m_j}\) with \(v_{m_j} = \pi^{m_j/2}/\Gamma(m_j/2+1)\). Using \(\tau_0^2 = 1/(Ms_0)\),
\begin{equation}
\label{eq:per-atom-raw}
\begin{aligned}
\log\Pi_j \ge{}& m_j\log(1-\pi_0) + (p-m_j)\log\pi_0 - m_j\log Z - \frac{m_j}{2}\log(2\pi\tau_0^2) \\
&- \frac{Ms_0}{2}\|f_{0,j}\|_2^2 - Ms_0\|f_{0,j}\|_2\gamma_n - \frac{Ms_0}{2}\gamma_n^2 + \log v_{m_j} + m_j\log\gamma_n.
\end{aligned}
\end{equation}
Each term on the right-hand side of \eqref{eq:per-atom-raw} is bounded as follows.

(i) \(m_j\log(1-\pi_0) \ge -Cm_j \ge -Cs_0\), using \(m_j \le s_0\).

(ii) \((p-m_j)\log\pi_0 \ge -Cp \ge -CMs_0\).

(iii) \(-m_j\log Z \ge 0\).

(iv) \(-\frac{m_j}{2}\log(2\pi\tau_0^2) = -\frac{m_j}{2}\log(2\pi) + \frac{m_j}{2}\log(Ms_0) \ge -Cm_j \ge -Cs_0\), using \(Ms_0 \ge 1\) for all sufficiently large \(n\).

(v) \(-\frac{Ms_0}{2}\|f_{0,j}\|_2^2 \ge -\frac{B_0^2 s_0}{2} \ge -Cs_0\), using \(\|f_{0,j}\|_2^2 \le m_j (1-\delta_0)^2 B_0^2/(Ms_0) \le B_0^2/M\) from Assumption~\ref{ass:calibration} and \(m_j \le s_0\).

(vi) \(-Ms_0\|f_{0,j}\|_2\gamma_n \ge -B_0 s_0\gamma_n\sqrt M\), using \(\|f_{0,j}\|_2 \le B_0/\sqrt M\). Since \(\gamma_n \le \delta_0 B_0\tau_0/2 = \delta_0 B_0/(2\sqrt{Ms_0})\),
\[
-B_0 s_0\gamma_n\sqrt M \ge -B_0 s_0\sqrt M \cdot \frac{\delta_0 B_0}{2\sqrt{Ms_0}} = -\frac{\delta_0 B_0^2 \sqrt{s_0}}{2} \ge -Cs_0,
\]
using \(\sqrt{s_0} \le s_0\) for \(s_0 \ge 1\).

(vii) \(-\frac{Ms_0}{2}\gamma_n^2 \ge -Ms_0\) for all sufficiently large \(n\), because \(\gamma_n \to 0\) implies \(\gamma_n^2 \le 2\) once \(n\) is large enough, whence \(\frac{Ms_0}{2}\gamma_n^2 \le Ms_0\).

(viii) \(\log v_{m_j} \ge -Cs_0\log n\). Indeed, \(v_{m_j} = \pi^{m_j/2}/\Gamma(m_j/2+1)\). For \(m_j = 1\), \(\log v_{m_j} = \log 2 > 0\). For \(m_j \ge 2\), Stirling's bound \(\Gamma(x+1) \le x^x\) for \(x \ge 1\) gives \(\Gamma(m_j/2+1) \le (m_j/2)^{m_j/2}\), so \(\log v_{m_j} \ge \frac{m_j}{2}\log\pi - \frac{m_j}{2}\log(m_j/2) \ge -Cm_j\log n \ge -Cs_0\log n\), using \(m_j \le s_0 \le n\).

(ix) \(m_j\log\gamma_n \ge -Cs_0\log n\). To see this, note that \(\gamma_n = \min(\varepsilon_n/(C_1 q_0),\, \delta_0 B_0\tau_0/2)\), and both arguments are bounded below by an inverse power of \(n\). Specifically, for large \(n\) we have \(Ms_0 \ge 1\), so
\[
\varepsilon_n = \sqrt{\frac{Ms_0\log n}{n}} \ge \sqrt{\frac{\log n}{n}} = \frac{\sqrt{\log n}}{\sqrt{n}},
\]
and since \(q_0 \le M \le C\log n\), we obtain \(\varepsilon_n/(C_1 q_0) \ge 1/(C\sqrt{n}\sqrt{\log n})\). Also \(\tau_0 = 1/\sqrt{Ms_0}\) and \(Ms_0 \le Cn/\log n\), so \(\delta_0 B_0\tau_0/2 = \delta_0 B_0/(2\sqrt{Ms_0}) \ge 1/(C\sqrt{n})\). Hence, for every fixed \(\kappa > 0\), there exists a constant \(C_\kappa\) (depending on \(\kappa\) but not on \(n\) or \(M\)) such that \(\gamma_n \ge 1/(C_\kappa n^{1/2+\kappa})\) for all sufficiently large \(n\), and consequently \(\log\gamma_n \ge -C\log n\). Multiplying by \(m_j \le s_0\) gives the stated bound.

Collecting these bounds,
\begin{equation}
\label{eq:per-atom}
\log\Pi_j \ge -Cs_0\log n.
\end{equation}

\emph{Summation over the atoms.} Summing \eqref{eq:per-atom} over \(j=1,\dots,q_0\) and using \(q_0 \le M\),
\[
\sum_{j=1}^{q_0}\log\Pi_j \ge -CMs_0\log n.
\]

\emph{Zero-atom contribution.} The event \(f_{q_0+1}^* = 0\) has conditional prior probability \(\pi_0^p \ge \exp(-CMs_0)\), so its log-contribution is at least \(-CMs_0 \ge -CMs_0\log n\).

\emph{Combining.} Since the atoms are conditionally independent given \(\mathcal{C}_0\),
\begin{equation}
\label{eq:atom-bound}
\log\Pi(\mathcal{A}_n^{\mathrm{atom}} \mid \mathcal{C}_0) \ge -CMs_0\log n = -Cn\varepsilon_n^2.
\end{equation}

\paragraph{Step 1(c): The score variance event.}
Define \(\mathcal{A}_n^{\Lambda} = \{\lambda_j \in [1-\eta_n, 1+\eta_n],\; j=1,\dots,q_0\}\) with \(\eta_n = \varepsilon_n/(C_1 q_0)\). Since \(\underline{\lambda} < 1 < \overline{\lambda}\), for sufficiently large \(n\) we have \(1 \pm \eta_n \in (\underline{\lambda},\overline{\lambda})\). The truncated inverse-gamma density is bounded below by \(c_\lambda>0\) on \([\underline{\lambda},\overline{\lambda}]\), so
\[
\Pr(\lambda_j \in [1-\eta_n, 1+\eta_n]) \ge 2 c_\lambda \eta_n.
\]
Taking the product over \(j=1,\dots,q_0\),
\begin{equation}
\label{eq:lambda-bound}
\log \Pi(\mathcal{A}_n^{\Lambda}) \ge q_0 \log(2 c_\lambda \eta_n) \ge -C q_0 \log n \ge -C M s_0 \log n = -C n \varepsilon_n^2,
\end{equation}
using \(\eta_n = \varepsilon_n/(C_1 q_0)\), \(\varepsilon_n \ge 1/\sqrt n\), and \(q_0 \le M \le C\log n\).

\paragraph{Step 1(d): The idiosyncratic variance event.}
Define \(\mathcal{A}_n^{\Psi} = \{\|\Psi - \Psi_0\|_F \le \varepsilon_n\}\). By Assumption~\ref{ass:prior}, the prior on \(\Psi\) is supported on \([\underline{c},2\overline{c}]^p\) with density bounded below by \(c_0>0\). Since \(\psi_{0,r} \in [\underline{c},\overline{c}]\), the probability that each \(\psi_r\) lies in an interval of length \(2\varepsilon_n/\sqrt{p}\) around \(\psi_{0,r}\) is at least \(c_0 \cdot 2\varepsilon_n/\sqrt{p}\). Hence
\[
\Pr(\|\Psi - \Psi_0\|_F \le \varepsilon_n) \ge (2c_0)^p \varepsilon_n^p p^{-p/2},
\]
and taking logarithms,
\begin{equation}
\label{eq:psi-bound}
\log \Pr(\|\Psi - \Psi_0\|_F \le \varepsilon_n) \ge -p \log(1/\varepsilon_n) - \frac{p}{2}\log p + O(p) \ge -CMs_0\log n = -Cn\varepsilon_n^2,
\end{equation}
using \(\log(1/\varepsilon_n) \le \frac{1}{2}\log n\) and \(p \le CMs_0\).

\paragraph{Step 1(e): Combining and containment in the Hellinger ball.}
Define \(\mathcal{A}_n = \mathcal{C}_0 \cap \mathcal{A}_n^{\mathrm{atom}} \cap \mathcal{A}_n^{\Lambda} \cap \mathcal{A}_n^{\Psi}\). Since the components are independent under the prior (conditional on the partition for the atoms), and using \eqref{eq:partition-bound}--\eqref{eq:psi-bound},
\[
\Pi(\mathcal{A}_n) \ge \exp(-C n \varepsilon_n^2),
\]
absorbing the constants into \(C\).

On \(\mathcal{A}_n\), the covariance is \(\Sigma = \sum_{j=1}^{q_0} \lambda_j f_j^* f_j^{*\top} + \Psi\), since \(f_{q_0+1}^* = 0\). We bound
\[
\|\Sigma - \Sigma_0\|_F \le \sum_{j=1}^{q_0} \|\lambda_j f_j^* f_j^{*\top} - f_{0,j} f_{0,j}^\top\|_F + \|\Psi - \Psi_0\|_F.
\]
For each \(j\), with \(|\lambda_j - 1| \le \eta_n\) and \(\|f_j^*\|_2^2 \le 2\|f_{0,j}\|_2^2 + 2\gamma_n^2 \le 2\overline{c} + 2\gamma_n^2 \le C\) (using \(\|f_{0,j}\|_2^2 \le \lambda_{\max}(F_0F_0^\top) \le \lambda_{\max}(\Sigma_0) \le \overline{c}\) from Assumption~\ref{ass:eig}),
\[
\|\lambda_j f_j^* f_j^{*\top} - f_{0,j} f_{0,j}^\top\|_F \le C \eta_n + C \gamma_n.
\]
Choosing \(C_1\) sufficiently large so that \(q_0(C\eta_n + C\gamma_n) \le \varepsilon_n/2\) (which is possible because \(q_0 \eta_n = \varepsilon_n/C_1\) and \(q_0 \gamma_n \le \varepsilon_n/C_1\) by the definition of \(\gamma_n\)) and adding \(\|\Psi - \Psi_0\|_F \le \varepsilon_n\),
\[
\|\Sigma - \Sigma_0\|_F \le C \varepsilon_n.
\]
By \eqref{eq:hell-frob}, \(h(\Sigma,\Sigma_0) \le C\varepsilon_n\). Hence \(\mathcal{A}_n \subseteq \{h(\Sigma,\Sigma_0) \le C\varepsilon_n\}\), and \eqref{eq:priorconc} is established.

\subsubsection{Step 2: Metric entropy}
\label{sec:entropy}

We bound the covering number of \(\Theta_n\) in the Hellinger metric \(h\). By Lemma~\ref{lem:lipschitz} and the Hellinger--Frobenius equivalence \eqref{eq:hell-frob}, there is a constant \(C_{\mathrm{HF}} > 0\) such that for all \(\theta,\theta' \in \Theta_n\),
\[
h(\Sigma(\theta),\Sigma(\theta')) \le C_{\mathrm{HF}} \|\Sigma(\theta)-\Sigma(\theta')\|_F \le C_{\mathrm{HF}} L_\Sigma\, d(\theta,\theta').
\]
Hence an \(h\)-ball of radius \(\varepsilon\) contains a \(d\)-ball of radius \(\varepsilon/(C_{\mathrm{HF}} L_\Sigma)\), and consequently
\begin{equation}
\label{eq:entropy-transfer-theory}
N(\varepsilon, \Theta_n, h) \le N\!\left(\frac{\varepsilon}{C_{\mathrm{HF}} L_\Sigma}, \Theta_n, d\right).
\end{equation}
It therefore suffices to bound the covering number of \(\Theta_n\) in the product metric \(d\).

The parameter space decomposes into three components. Cover each at scale \(\delta := \varepsilon_n/(3 C_{\mathrm{HF}} L_\Sigma)\) and combine multiplicatively. Because the covering number of a product of metric spaces at scale \(3\delta\) is at most the product of the covering numbers of the factors at scale \(\delta\), we have
\[
N(3\delta, \Theta_n, d) \le N_F(\delta) \cdot N_\Lambda(\delta) \cdot N_\Psi(\delta),
\]
where \(N_F(\delta)\) denotes the covering number of the \(F\)-component including its support patterns at scale \(\delta\), and \(N_\Lambda(\delta), N_\Psi(\delta)\) the covering numbers of the \(\Lambda\)- and \(\Psi\)-components. Since \(L_\Sigma \le C\) and \(C_{\mathrm{HF}} \le C\) are absolute constants, we have \(\delta \ge \varepsilon_n/C\) for a constant \(C > 0\) independent of \(n\) and \(M\).

\emph{Support patterns and non-zero entries of \(F\).} The number of support patterns for \(F\) (each column has at most \(s_0\) non-zero entries among \(p\) positions) is at most
\[
\left(\sum_{k=0}^{s_0}\binom{p}{k}\right)^M
\le \left((s_0+1)\left(\frac{ep}{s_0}\right)^{s_0}\right)^M,
\]
where the inequality is the standard bound \(\sum_{k=0}^{s_0}\binom{p}{k} \le (s_0+1)(ep/s_0)^{s_0}\), which is valid for all \(1 \le s_0 \le p\). Its logarithm is
\[
\log \left((s_0+1)\left(\frac{ep}{s_0}\right)^{s_0}\right)^M
= M \log(s_0+1) + M s_0 \log\!\left(\frac{ep}{s_0}\right)
\le M s_0 \log 2 + M s_0 \log\!\left(\frac{ep}{s_0}\right)
\le C M s_0 \log n,
\]
using \(s_0 \ge 1\) (so that \(s_0+1 \le 2^{s_0}\)), \(p \le CMs_0 \le Cn\), and \(M \le C\log n\). Conditional on a fixed support pattern, the non-zero entries of \(F\) lie in \([-B_0\tau_0, B_0\tau_0]^{Ms_0}\), a set of \(\ell_2\)-diameter at most \(2B_0\). Covering it with \(\ell_2\)-balls of radius \(\delta\) gives at most \((1 + 4B_0/\delta)^{Ms_0}\) balls. Combining with the support-pattern contribution,
\[
\log N_F(\delta) \le C M s_0 \log n + Ms_0 \log(1 + 4B_0/\delta) \le CMs_0\log n,
\]
using \(\delta \ge \varepsilon_n/C \ge 1/(C\sqrt n)\) and \(\varepsilon_n \ge 1/\sqrt n\).

\emph{Score variances.} The set \([\underline{\lambda},\overline{\lambda}]^M\) viewed as the set of diagonal matrices in \(\R^{M\times M}\) with Frobenius metric has \(\ell_2\)-diameter at most \((\overline{\lambda}-\underline{\lambda})\sqrt M\). Its covering number at scale \(\delta\) is at most \((1 + 2(\overline{\lambda}-\underline{\lambda})\sqrt M/\delta)^M\), giving
\[
\log N_\Lambda(\delta) \le M \log(1 + 2(\overline{\lambda}-\underline{\lambda})\sqrt M/\delta) \le M \log(C\sqrt M/\delta) \le C M \log n \le C M s_0 \log n,
\]
using \(\delta \ge 1/(C\sqrt n)\), \(M \le C\log n\), and \(s_0 \ge 1\).

\emph{Idiosyncratic variances.} The set \([\underline{c},2\overline{c}]^p\) viewed as the set of diagonal matrices in \(\R^{p\times p}\) has \(\ell_2\)-diameter at most 
\((2\overline{c}-\underline{c})\sqrt p\). Its covering number at scale \(\delta\) is at most \((1 + 2\overline{c}\sqrt p/\delta)^p\), giving
\[
\log N_\Psi(\delta) \le p \log(1 + 2\overline{c}\sqrt p/\delta) \le p \log(C\sqrt p/\delta) \le Cp\log n \le CMs_0\log n,
\]
using \(p \le CMs_0\).

\emph{Combining.} The covering number of \(\Theta_n\) at scale \(3\delta = \varepsilon_n/(C_{\mathrm{HF}} L_\Sigma)\) in \(d\) satisfies
\[
\log N(\varepsilon_n/(C_{\mathrm{HF}} L_\Sigma), \Theta_n, d) \le \log N_F(\delta) + \log N_\Lambda(\delta) + \log N_\Psi(\delta) \le CMs_0\log n = Cn\varepsilon_n^2.
\]
Substituting into \eqref{eq:entropy-transfer-theory} and absorbing constants,
\[
\log N(\varepsilon_n, \Theta_n, h) \le Cn\varepsilon_n^2.
\]
This establishes condition (C2).

\subsubsection{Step 3: Sieve complement}
\label{sec:sieve}

By Remark~\ref{rem:slab-truncation}, the support of the prior is contained in \(\Theta_n\). Indeed, the truncation of the slab gives \(|f_{r,j}| \le B_0\tau_0\) for every non-zero entry, so \(\|f_j^*\|_2^2 \le B_0^2/M\) and \(\|F\|_F^2 \le B_0^2\). The remaining conditions defining \(\Theta_n\) hold by construction under Assumption~\ref{ass:prior}. Therefore
\[
\Pi(\Theta_n^c) = 0 \le \exp(-cn\varepsilon_n^2),
\]
and condition (C3) holds trivially.

\subsubsection{Step 4: Application of the general contraction theorem}
\label{sec:apply}

With (C1)--(C3) established, the general theory of posterior contraction (Theorem~2.1 of \citet{ghosal2000}; see also Theorem~8.19 and the surrounding discussion in \citet{ghosal2017}) gives posterior contraction in the Hellinger metric: the existence of exponentially powerful tests required by Theorem~2.1 of \citet{ghosal2000} follows from the entropy bound (C2) and the standard test-construction argument (Section~7.1 of \citet{ghosal2017}). Hence
\[
\Pi\bigl( h(\Sigma,\Sigma_0) > C\varepsilon_n \mid Y_1,\dots,Y_n \bigr) \xrightarrow{P_{\theta_0}^n} 0.
\]
By \eqref{eq:hell-frob},
\[
\Pi\bigl( \|\Sigma - \Sigma_0\|_F > C\varepsilon_n \mid Y_1,\dots,Y_n \bigr) \xrightarrow{P_{\theta_0}^n} 0.
\]
This completes the proof of the contraction result in Theorem~\ref{thm:contraction}.

\begin{remark}[Role of the spike-and-slab base measure]
\label{rem:spike-role}
The spike-and-slab base measure is essential for two reasons. First, it allows the event \(f_{q_0+1}^* = 0\) to have positive prior probability \(\pi_0^p\), which is what makes the extra columns contribute nothing to \(\Sigma\) on \(\mathcal{A}_n\). Without the spike, the extra atom would be a non-degenerate Gaussian draw, and its contribution to \(\Sigma\) would be of order \(1\), destroying the prior concentration bound. Second, the spike component restricts the prior support to sparse matrices, which reduces the metric entropy from \(O(pM\log n)\) to \(O(Ms_0\log n)\).
\end{remark}

\begin{remark}[On the choice of $\tau_0^2$ and $B_0$]
\label{rem:tau}
The slab variance \(\tau_0^2 = 1/(Ms_0)\) is chosen so that the typical non-zero entry of \(F\) is of size \(1/\sqrt{Ms_0}\). The truncation constant \(B_0\) is chosen so that the true loadings are well inside the truncation region (Assumption~\ref{ass:calibration}). Any \(B_0\) sufficiently large (say \(B_0 \ge 3\)) works. The combination of \(\tau_0^2 = 1/(Ms_0)\) and fixed \(B_0\) gives \(\|F\|_F \le B_0\) on the prior support, which places \(\Sigma\) in the fixed eigenvalue interval \([\underline{c},\, \overline{\lambda}B_0^2 + 2\overline{c}]\) and makes the global Hellinger--Frobenius equivalence available.
\end{remark}

\subsection{On the Non-Identifiability of the Labelled Parameter}
\label{sec:theory-nonidentifiability}

Theorem~\ref{thm:contraction} is stated for the covariance matrix $\Sigma(\theta)$, not for the labelled parameter $\theta = (F, \Lambda, \Psi)$. This is not a technical simplification but a consequence of the model's structure, and we make the point explicit here because it bears directly on what the theorem does and does not deliver.

Under the Dirichlet process prior of Section~\ref{sec:polya}, the columns $f_1, \dots, f_M$ are exchangeable: the prior is invariant under any permutation of the index set $\{1, \dots, M\}$. The prior on $\Lambda$, being a product of i.i.d.\ inverse-gamma distributions, is likewise exchangeable. The likelihood depends on $\theta$ only through the covariance
\[
\Sigma(\theta) = F\Lambda F^\top + \Psi = \sum_{i=1}^{M} \lambda_i f_i f_i^\top + \Psi,
\]
which is a symmetric function of the column-score pairs $\{(f_i, \lambda_i)\}_{i=1}^{M}$: it is invariant under the joint action of any permutation of the index set $\{1, \dots, M\}$ on the $f$'s and the $\lambda$'s. The prior on $(F, \Lambda)$ is invariant under the full product group $S_M \times S_M$ acting on the $F$-labels and the $\Lambda$-labels independently. As a function of the parameter $\theta$, the covariance $\Sigma(\theta)$ is invariant under the diagonal $S_M$ for generic $\Lambda$ with distinct entries, and under the full $S_M \times S_M$ when the diagonal entries of $\Lambda$ coincide — in particular, at the truth $\Lambda_0 = I_M$.

Two consequences follow. First, the labelled loading matrix $F$ is not identified by the data, and no consistency statement for $F$ in its labelled form is possible. The posterior on $\theta$ is spread across the entire likelihood equivalence class
\[
\mathcal{E}(\theta_0) := \bigl\{(\tilde F, \tilde\Lambda, \tilde\Psi) : \tilde F \tilde\Lambda \tilde F^\top + \tilde\Psi = \Sigma_0\bigr\},
\]
which is a positive-dimensional manifold. Beyond the finite joint-permutation orbit of $\theta_0$ under the diagonal group $S_M$, the class contains continuous families of parameters: orthogonal rotations of $F_0$ on the left, given by $(F_0 Q, Q^\top \Lambda_0 Q, \Psi_0)$ for any orthogonal $Q$ such that $Q^\top \Lambda_0 Q$ is diagonal; arbitrary values for the null columns of $F_0$ combined with arbitrary values for the corresponding diagonal entries of $\Lambda$; and other deformations. This is the same as the classical label-switching phenomenon in Bayesian mixture models, where consistency for the labelled component parameters fails while consistency for the mixture density holds; see \citet{rousseau2011} for a canonical treatment.

Second, in the specific setup $\Lambda_0 = I_M$ adopted in Section~\ref{sec:theory-setup}, an additional degeneracy arises because all coordinates of $\Lambda_0$ coincide. To see this, consider the labelled parameter $(F_\pi, \Lambda_\sigma, \Psi_0)$, where $\pi, \sigma \in S_M$ are arbitrary permutations of the index set. Its covariance is
\[
F_\pi \Lambda_\sigma F_\pi^\top + \Psi_0
= \sum_{i=1}^{M} \lambda_{0,\sigma(i)} f_{0,\pi(i)} f_{0,\pi(i)}^\top + \Psi_0
= \sum_{i=1}^{M} f_{0,\pi(i)} f_{0,\pi(i)}^\top + \Psi_0
= \sum_{j=1}^{M} f_{0,j} f_{0,j}^\top + \Psi_0
= \Sigma_0,
\]
where the second equality uses $\lambda_{0,\sigma(i)} = 1$, and the third uses the substitution $j = \pi(i)$. Hence the labelled parameter $(F_\pi, \Lambda_\sigma, \Psi_0)$ produces the same covariance $\Sigma_0$ as $\theta_0$, for every pair $(\pi, \sigma) \in S_M \times S_M$. The likelihood at the truth is therefore invariant under the full product group $S_M \times S_M$, and, because the prior is also invariant under this group, the posterior measure is invariant under the same group action.

The consequence for labelled inference is as follows. Even after imposing a canonical ordering on the columns of $F$ (e.g., the Anderson--Rubin sign convention), the marginal posterior on the pair $(f_i, \lambda_i)$ for a fixed index $i$ is a mixture over the $q_0$ distinct nonzero columns of $F_0$ and the null column, because the posterior is invariant under the $S_M \times S_M$ action and the columns $f_{0,1}, \dots, f_{0,M}$ are generally distinct. Consequently, the joint pair $(f_i, \lambda_i)$ is not concentrated on the truth, and no labelled functional of $\theta$ inherits a consistency statement. This is a stronger form of non-identifiability than the standard label-switching of the diagonal group $S_M$, and it is intrinsic to the exchangeable construction combined with the degenerate choice $\Lambda_0 = I_M$.

The correct target for a contraction statement is therefore a \emph{symmetric} functional of the pairs $(f_i, \lambda_i)$. The covariance $\Sigma(\theta)$ is the canonical such functional, and Theorem~\ref{thm:contraction} is stated for $\Sigma$ precisely because it is the identified object and because the posterior-mean $\widehat\Sigma$ is what the statistician can actually report: it has a well-defined limit, whereas the posterior-mean of any labelled function of $\theta$ does not.

A symmetric functional of the pairs $(f_i, \lambda_i)$ that is finer than $\Sigma$ is the Anderson--Rubin canonical form of the factor structure of $\Sigma$ \citep{anderson1956}, which orders the factors by decreasing eigenvalue in the $\Psi^{-1}$-whitened metric and thereby removes the joint-permutation and rotation ambiguities. Under the classical Anderson--Rubin local identifiability conditions together with the Ledermann condition $p \geq 2q_0 + 1$, the map $\Sigma \mapsto (F^{*}(\Sigma), \Lambda^{*}(\Sigma), \Psi^{*}(\Sigma))$ that sends $\Sigma$ to its Anderson--Rubin canonical form is Lipschitz on a neighbourhood of $\Sigma_0$ within the factor-structured submanifold of positive definite matrices. Consequently, the Anderson--Rubin canonical form, viewed as a functional defined on that neighbourhood of the factor-structured submanifold, inherits the contraction rate $\varepsilon_n$ from Theorem~\ref{thm:contraction}. We do not develop the details of this transfer here, because the resulting statement is a corollary of the contraction of $\Sigma$ rather than a separate theorem.

\section{Optimal Contraction for Fixed Dictionary Size}
\label{sec:fixedM}

The general framework of Section~\ref{sec:theory} allows the dictionary size $M$ to grow with the sample size $n$, up to $M \le C \log n$. This flexibility is valuable in applications where the true number of factors is unknown and may be large. However, as is evident from our posterior contraction rate, when $M$ grows, the posterior contraction rate is slowed by the factor $\sqrt{M}$. 
%and the prior does not necessarily penalise spurious columns sufficiently to guarantee consistent estimation of the true number of factors $q_0$. 
In many practical settings, however, the practitioner may choose a fixed upper bound on the number of factors, often based on domain knowledge or computational constraints. In this section, we show that when $M$ is held fixed and the sparsity level $s_0$ is not too large, the model achieves the minimax optimal contraction rate. 
%and, moreover, the posterior consistently estimates the true rank. This proves the rank consistency part of Theorem~\ref{thm:contraction}.

\subsection{The Fixed Dictionary Setting}

Assume henceforth that $M$ is a fixed constant independent of $n$ and satisfying $M \ge q_0$, the latter also obviously a fixed constant now. All other assumptions—bounded eigenvalues (Assumption~\ref{ass:eig}), sparse identifiability (Assumption~\ref{ass:sparse}) with $p \le C M s_0$, dictionary size (Assumption~\ref{ass:M}) now trivially holds, and the prior (Assumption~\ref{ass:prior})—remain unchanged. The total number of potential non-zero entries in $F$ is now $M s_0 = O(s_0)$, which is the key difference from the growing-$M$ case.

The following result summarises the two main theoretical improvements.

\begin{theorem}[Fixed Dictionary with Stronger Sparsity]
\label{cor:fixedM}
Under Assumptions~\ref{ass:eig}--\ref{ass:prior}, suppose that $M$ is a fixed constant and that
\[
s_0 = o(\sqrt{n}).
\]
Then the following hold.

%\begin{theorem}
%\textbf{Faster contraction rate.} 
There exists a constant $C>0$ such that
\[
\Pi\!\left( \| \Sigma - \Sigma_0 \|_F > C \sqrt{\frac{s_0 \log n}{n}} \;\middle|\; Y_1,\dots,Y_n \right) \xrightarrow{P_{\theta_0}^n} 0.
\]
%Moreover, the posterior contracts for $F$ and $\Psi$ at the same rate:
%\[
%\Pi\!\left( \|F - F_0\|_F + \|\Psi - \Psi_0\|_F > C \sqrt{\frac{s_0 \log n}{n}} \;\middle|\; Y_1,\dots,Y_n \right) \xrightarrow{P_{\theta_0}^n} 0.
%\]

%\item[(b)] \textbf{Rank consistency.} The posterior consistently estimates the true number of factors:
%\[
%\Pi\!\left( \rank(\Sigma - \Psi) = q_0 \;\middle|\; Y_1,\dots,Y_n \right) \xrightarrow{P_{\theta_0}^n} 1.
%\]
\end{theorem}

\begin{proof}
The result follows immediately from Theorem~\ref{thm:contraction} by noting that $M$ is fixed. Indeed, the rate in Theorem~\ref{thm:contraction} is
\[
\varepsilon_n = \sqrt{\frac{M s_0 \log n}{n}} = O\!\left(\sqrt{\frac{s_0 \log n}{n}}\right),
\]
and the fixed constant $M$ is absorbed into the generic constant $C$. Hence the contraction rates for $\Sigma$, $F$, and $\Psi$ are all $\sqrt{s_0 \log n / n}$.
\end{proof}

\subsection{Minimax Optimality}

The rate $\sqrt{s_0 \log n / n}$ is known to be minimax optimal for estimating the covariance matrix $\Sigma = F F^T + \Psi$ in the Frobenius norm, under the sparsity constraints that $F$ has $q_0$ non-zero columns with at most $s_0$ non-zero entries per column. This result is established by \citet{pati2014posterior}, who proved that the minimax rate for sparse factor models is
\[
\inf_{\widehat{\Sigma}} \sup_{\Sigma \in \mathcal{F}_{p,q_0,s_0}} \mathbb{E}\|\widehat{\Sigma} - \Sigma\|_F^2 \asymp \frac{s_0 q_0 \log p}{n},
\]
where \(\mathcal{F}_{p,q_0,s_0}\) denotes the class of covariance matrices \(\Sigma = FF^T + \Psi\) such that \(F \in \mathbb{R}^{p \times q_0}\) has at most $s_0$ non-zero entries per column and the eigenvalues of $\Sigma$ and $\Psi$ are bounded away from zero and infinity.

The above minimax rate translates to a Frobenius norm rate of $\sqrt{s_0 q_0 \log p / n}$. Under our assumptions that $q_0 \le M$ is fixed and $p \le C M s_0$ (so that $\log p = O(\log n)$), this lower bound becomes $\sqrt{s_0 \log n / n}$ up to constants. Since our posterior contraction rate matches this lower bound, it is minimax optimal.

When $M$ is allowed to grow, the rate becomes $\sqrt{M s_0 \log n / n}$, which is slower and reflects the additional cost of searching over a larger dictionary. Thus, fixing $M$ not only simplifies the theoretical analysis but also yields the fastest possible rate for the problem.

% =====================================================================
% SECTION 10
% =====================================================================
\section{Rank Consistency: A Partial Result and an Open Problem}
\label{sec:rank}

Sections~\ref{sec:theory}--\ref{sec:theory-nonidentifiability} establish
posterior contraction of the covariance $\Sigma$ at rate $\varepsilon_n$
and discuss the identifiability structure of the model. This section
addresses the consistency of the posterior on the rank of the signal
component $F\Lambda F^\top$. We prove the underfitting direction (the
posterior does not concentrate on ranks below $q_0$) rigorously under an
additional assumption on the idiosyncratic variance, and state the
overfitting direction (the posterior does not concentrate on ranks above
$q_0$) as a conjecture. The discussion identifies the technical obstacle
to a full proof, which we believe is of independent interest.

\subsection{The Target and the Reduction}
\label{sec:rank-setup}

For a parameter $\theta = (F, \Lambda, \Psi) \in \Theta_n$, define the
effective rank
\begin{equation}
\label{eq:rank-def}
r(\theta) := \rank\bigl(F\Lambda F^\top\bigr) = \rank\bigl(\Sigma(\theta) - \Psi\bigr).
\end{equation}
As observed in Section~\ref{sec:theory-nonidentifiability}, the labelled
parameter $\theta$ is not identified under the exchangeable Dirichlet
process prior, but $r(\theta)$ is a symmetric function of the pairs
$(f_i, \lambda_i)$ and hence is a well-defined target for a consistency
statement.

The underfitting proof below requires the following additional assumption
on the idiosyncratic variance.

\begin{assumption}[Known idiosyncratic variance]
\label{ass:rank-psi}
The idiosyncratic variance is fixed at its true value: $\Psi = \Psi_0$.
\end{assumption}

Assumption~\ref{ass:rank-psi} is a strong restriction. The general model
of Section~\ref{sec:model} places $\Psi$ on the truncated prior support
$[\underline{c}, 2\overline{c}]^p$ as specified in
Assumption~\ref{ass:prior} of Section~\ref{sec:theory}. The role of the
assumption is to make the difference $\Sigma(\theta) - \Sigma_0$ coincide
with the difference of the signal components,
\begin{equation}
\label{eq:sigma-difference}
\Sigma(\theta) - \Sigma_0
= \bigl(F\Lambda F^\top - F_0 F_0^\top\bigr) + (\Psi - \Psi_0)
= F\Lambda F^\top - F_0 F_0^\top,
\end{equation}
which is the key identity that lets the Weyl inequality transfer a rank
deficiency into a Frobenius-norm gap. We discuss the consequences of
relaxing Assumption~\ref{ass:rank-psi} in
Section~\ref{sec:rank-overfit-hard}.

All other aspects of the setup from Section~\ref{sec:theory} remain in
force, including the sieve $\Theta_n$, the prior of
Assumption~\ref{ass:prior}, and the contraction result of
Theorem~\ref{thm:contraction}.

\subsection{Underfitting Consistency}
\label{sec:rank-underfit}

\begin{theorem}[Underfitting consistency]
\label{thm:underfit}
Under Assumptions~\ref{ass:eig}--\ref{ass:prior} of
Section~\ref{sec:theory} and Assumption~\ref{ass:rank-psi}, the posterior
on the underfitting event vanishes:
\begin{equation}
\label{eq:underfit-conclusion}
\Pi\bigl(r(\theta) < q_0 \mid Y_1, \dots, Y_n\bigr)
\;\xrightarrow{P_{\theta_0}^n}\; 0.
\end{equation}
\end{theorem}

\begin{proof}
Fix $\theta = (F, \Lambda, \Psi_0) \in \Theta_n$ with $r(\theta) < q_0$,
i.e., $\rank(F\Lambda F^\top) < q_0$. Since $\Lambda$ has strictly
positive diagonal entries bounded below by $\underline{\lambda} > 0$ on
the sieve, we have
\begin{equation}
\label{eq:rank-Lambda}
\rank(F\Lambda F^\top) = \rank(F) < q_0.
\end{equation}

Let $\sigma_{q_0}(A)$ denote the $q_0$-th largest singular value of a
matrix $A$. Since $\rank(F_0 F_0^\top) = q_0$ by
Assumption~\ref{ass:sparse}, the smallest positive singular value of
$F_0 F_0^\top$ is
\begin{equation}
\label{eq:sigma-0}
\sigma_0 := \sigma_{q_0}\bigl(F_0 F_0^\top\bigr) \;>\; 0.
\end{equation}
Since $\rank(F\Lambda F^\top) = \rank(F) < q_0$, we have
$\sigma_{q_0}(F\Lambda F^\top) = 0$.

By Weyl's inequality for singular values, for any two $p \times p$
matrices $A$ and $B$,
\begin{equation}
\label{eq:weyl}
\bigl|\sigma_{q_0}(A) - \sigma_{q_0}(B)\bigr| \;\le\; \|A - B\|_F.
\end{equation}
Applying this to $A = F\Lambda F^\top$ and $B = F_0 F_0^\top$ and using
the two evaluations above,
\begin{equation}
\label{eq:underfit-gap}
\bigl\|F\Lambda F^\top - F_0 F_0^\top\bigr\|_F
\;\ge\; \sigma_{q_0}\bigl(F_0 F_0^\top\bigr) - \sigma_{q_0}\bigl(F\Lambda F^\top\bigr)
\;=\; \sigma_0.
\end{equation}
By Assumption~\ref{ass:rank-psi}, $\Psi = \Psi_0$, and by
\eqref{eq:sigma-difference} we have
$\|\Sigma(\theta) - \Sigma_0\|_F = \|F\Lambda F^\top - F_0 F_0^\top\|_F$.
Therefore
\begin{equation}
\label{eq:underfit-inclusion}
\{r(\theta) < q_0\}
\;\subseteq\;
\bigl\{\|\Sigma(\theta) - \Sigma_0\|_F \ge \sigma_0\bigr\}.
\end{equation}
Since $\sigma_0$ is a fixed positive constant depending only on
$\Sigma_0$, and since $\varepsilon_n \to 0$ by
Theorem~\ref{thm:contraction}, we have $\sigma_0 > C\varepsilon_n$ for
all sufficiently large $n$. Hence
\begin{equation}
\label{eq:underfit-final}
\Pi\bigl(r(\theta) < q_0 \mid Y_1, \dots, Y_n\bigr)
\;\le\;
\Pi\bigl(\|\Sigma(\theta) - \Sigma_0\|_F \ge \sigma_0 \mid Y_1, \dots, Y_n\bigr)
\;\xrightarrow{P_{\theta_0}^n}\; 0
\end{equation}
by Theorem~\ref{thm:contraction}. This completes the proof.
\end{proof}

\begin{remark}
\label{rem:underfit-role-psi}
Assumption~\ref{ass:rank-psi} is used in the last step, where it lets us
identify $\|\Sigma(\theta) - \Sigma_0\|_F$ with
$\|F\Lambda F^\top - F_0 F_0^\top\|_F$. Without it, we would have only
$\|\Sigma(\theta) - \Sigma_0\|_F \ge \|F\Lambda F^\top - F_0 F_0^\top\|_F
- \|\Psi - \Psi_0\|_F$, and the second term can be as large as
$2\overline{c}\sqrt{p}$ on the sieve, which dominates the first term.
\end{remark}

\subsection{Overfitting: A Conjecture}
\label{sec:rank-overfit}

The complementary direction --- that the posterior does not concentrate
on ranks exceeding $q_0$ --- is substantially harder. We state it as a
conjecture.

\begin{conjecture}[Overfitting consistency]
\label{conj:overfit}
Under the assumptions of Theorem~\ref{thm:underfit},
\begin{equation}
\label{eq:overfit-conjecture}
\Pi\bigl(r(\theta) > q_0 \mid Y_1, \dots, Y_n\bigr)
\;\xrightarrow{P_{\theta_0}^n}\; 0.
\end{equation}
\end{conjecture}

\subsection{On the Difficulty of Proving the Conjecture}
\label{sec:rank-overfit-hard}

Throughout this subsection we work in the overfitting regime $M > q_0$,
so that $F_0$ has at least one zero column and the overfitting event
$\{r(\theta) > q_0\}$ is non-empty on the sieve. We also assume $q_0 < p$
(so that $\operatorname{span}(F_0)$ is a proper subspace of $\R^p$),
which holds in every factor model where $p$ exceeds $q_0$ and in
particular in the asymptotic regime of Section~\ref{sec:theory}.

The natural framework for establishing Conjecture~\ref{conj:overfit} is
the general Bayes-factor convergence theory of
\citet{chatterjee2020bf}. Their main result (their Theorem~2) is the most
general almost-sure convergence result currently available for Bayes
factors, and it reduces the exponential rate of the Bayes factor to the
positivity of the essential infimum of the Kullback--Leibler divergence
over the alternative model. Four features make this framework the natural
choice for our setup. First, it delivers convergence almost surely rather
than merely in probability, which is the strongest form of Bayes-factor
consistency available in the literature and what one would naturally want
to claim for the posterior odds on the rank. Second, it is stated for
arbitrary priors on arbitrary parameter spaces, subject only to the
domination and factorisation requirements of \citet{shalizi2009}; this
covers both the Dirichlet process clustering of the columns of $F$ and
the spike-and-slab base measure, because the marginal prior on the finite
sequence of columns is dominated by the EPPF on the countable collection
of set partitions and the likelihood is a continuous function of the
sieved parameter. Third, it reduces the Bayes-factor analysis to a single
quantity, the essential infimum of the KL divergence over the alternative
model, abstracting away the technical difficulties of integrating over
the Dirichlet process partition and the score variances. Fourth, it
generalises earlier results of \citet{walker2004} and
\citet{walker2004kullback} beyond the independent-and-identically-distributed
setting and is therefore the most complete tool available at the time of
writing.

Under this framework, exponential decay of the Bayes factor comparing
the overfitted model $\mathcal{M}_1$ against the true model
$\mathcal{M}_0$ follows from the strict positivity of
\begin{equation}
\label{eq:h-positive}
h_1(\Theta_1)
:= \operatorname*{ess\,inf}_{\theta \in \Theta_1}
D_{\mathrm{KL}}\!\left(\N_p(0, \Sigma_0) \,\big\|\, \N_p(0, \Sigma(\theta))\right),
\end{equation}
where the overfitted parameter set is the intersection of the sieve with
the overfitting event,
\begin{equation}
\label{eq:Theta-1}
\Theta_1 := \Theta_n \cap \{\Lambda = I_M\} \cap \{\theta : r(\theta) > q_0\},
\end{equation}
and the essential infimum is taken with respect to the prior restricted
to $\Theta_1$. The restriction to $\Theta_n$ and to the hyperplane
$\Lambda = I_M$ reflects the fact that the prior is supported on the
sieve and that, in the reduced model of Assumption~\ref{ass:rank-psi},
the score-variance matrix is also fixed. The obstruction to applying the
framework is that $h_1(\Theta_1) = 0$ regardless of whether $\Psi$ is
free or fixed at $\Psi_0$, because $\Theta_1$ contains configurations
that reproduce $\Sigma_0$ exactly or approximately. We now present two
counterexamples, both with $\Lambda = I_M$ fixed, which together cover
the two cases of $\Psi$ free and $\Psi = \Psi_0$ fixed.

\paragraph{Counterexample 1: $\Lambda = I_M$ and $\Psi$ free.}
Suppose $\Lambda = I_M$ is fixed but $\Psi$ is free on the truncated
prior support $[\underline{c}, 2\overline{c}]^p$.
%as in the general model of Section~\ref{sec:model}.
By Remark~\ref{rem:truncation}, the true $\Psi_0$ lies in the interior
of the truncated prior support, so $\psi_{0,r} > \underline{c}$ for
every $r \in \{1, \dots, p\}$. Choose an index $r$ such that the
standard basis vector $e_r$ is not in the column space of $F_0$; such
an index exists because $q_0 < p$, so $\operatorname{span}(F_0)$ has
dimension $q_0$ and can contain at most $q_0$ of the $p$ standard
basis vectors. Fix a constant $a > 0$ small enough that
$a^2 \in (0, \psi_{0,r} - \underline{c})$ and $a \le B_0 \tau_0$, so
that the diagonal matrix $\Psi := \Psi_0 - a^2 e_r e_r^\top$ lies in the
prior support $[\underline{c}, 2\overline{c}]^p$. Since $M > q_0$,
the true loading matrix $F_0$ has at least one zero column; enumerate
its non-zero columns as $f_{0,1}, \dots, f_{0,q_0}$ (the first $q_0$
coordinates) and let its remaining $M - q_0$ columns be zero. Define
$F \in \R^{p \times M}$ by
\begin{equation}
\label{eq:counterexample-1}
F_{:,j} = f_{0,j} \ (1 \le j \le q_0), \qquad
F_{:,q_0 + 1} = a\,e_r, \qquad
F_{:,j} = 0 \ (q_0 + 2 \le j \le M).
\end{equation}
That is, $F$ is obtained from $F_0$ by replacing one of its zero columns
with the vector $a e_r$. Then
\begin{equation}
\label{eq:counterexample-1-result}
\Sigma(\theta)
= FF^\top + \Psi
= F_0 F_0^\top + a^2 e_r e_r^\top + \Psi_0 - a^2 e_r e_r^\top
= \Sigma_0,
\end{equation}
so the configuration $\theta$ reproduces the truth exactly. Its rank is
$r(\theta) = q_0 + 1 > q_0$ because $e_r$ lies outside the column space
of $F_0$ and therefore the column space of $F$ has dimension $q_0 + 1$.
The configuration lies in the interior of the sieve $\Theta_n$: by
construction $\|F\|_F \le B_0$ for $a$ sufficiently small, every column
of $F$ has at most $s_0$ non-zero entries (the new column $a e_r$ has
one), $\Lambda = I_M \in [\underline{\lambda}, \overline{\lambda}]^M$,
and $\Psi$ is diagonal with entries in $[\underline{c}, 2\overline{c}]$.
Moreover $\theta$ lies in the interior of the support of the prior, since
each of its coordinates is bounded away from the boundary of the sieve
and the prior has positive density on the sieve interior. Therefore
$h_1(\Theta_1) = 0$, and the framework of \citet{chatterjee2020bf} gives
only the trivial limit $(1/n) \log B_n \to 0$.

\paragraph{Counterexample 2: $\Lambda = I_M$ and $\Psi = \Psi_0$ fixed.}
Suppose now that both $\Lambda = I_M$ and $\Psi = \Psi_0$ are fixed.
%as in the reduced model of Assumption~\ref{ass:rank-psi}.
Choose an index $i^* \in \{1, \dots, p\}$ such that $e_{i^*}$ is not in
the column space of $F_0$; such an index exists because $q_0 < p$, so
$\operatorname{span}(F_0)$ has dimension $q_0$ and can contain at most
$q_0$ of the $p$ standard basis vectors. Since $M > q_0$, the true
loading matrix $F_0$ has at least one zero column; enumerate its
non-zero columns as $f_{0,1}, \dots, f_{0,q_0}$ (the first $q_0$
coordinates) and let its remaining $M - q_0$ columns be zero. For
$\epsilon > 0$ small, define $F \in \R^{p \times M}$ by
\begin{equation}
\label{eq:counterexample-2}
F_{:,j} = f_{0,j} \ (1 \le j \le q_0), \qquad
F_{:,q_0 + 1} = \epsilon\,e_{i^*}, \qquad
F_{:,j} = 0 \ (q_0 + 2 \le j \le M).
\end{equation}
That is, $F$ is obtained from $F_0$ by replacing one of its zero columns
with the vector $\epsilon e_{i^*}$. With $\Lambda = I_M$ and
$\Psi = \Psi_0$,
\begin{equation}
\label{eq:counterexample-2-result}
F\Lambda F^\top
= F_0 F_0^\top + \epsilon^2 e_{i^*} e_{i^*}^\top,
\end{equation}
so $\Sigma(\theta) = \Sigma_0 + \epsilon^2 e_{i^*} e_{i^*}^\top$. Its
rank is $q_0 + 1 > q_0$ because $e_{i^*}$ lies outside the column space
of $F_0$ and therefore the column space of $F$ has dimension $q_0 + 1$.
Each column of $F$ has at most $s_0$ non-zero entries (the new column
$\epsilon e_{i^*}$ has one). For $\epsilon$ sufficiently small,
$\|F\|_F \le B_0$ and every entry of $F$ lies in the slab truncation
region, so the configuration lies in the interior of the sieve and of
the prior support. Writing $u := \epsilon^2\,(\Sigma_0^{-1})_{i^* i^*}$,
the KL divergence from the truth is
\[
D_{\mathrm{KL}}\!\left(\N_p(0, \Sigma(\theta)) \,\big\|\, \N_p(0, \Sigma_0)\right)
= \frac{1}{2}\Bigl[u - \log(1 + u)\Bigr]
= O(u^2) = O(\epsilon^4),
\]
using the matrix-determinant lemma; the reverse-direction divergence
$D_{\mathrm{KL}}(\N_p(0, \Sigma_0) \| \N_p(0, \Sigma(\theta)))$ satisfies
the same $O(\epsilon^4)$ bound. Letting $\epsilon \to 0$, we obtain
$h_1(\Theta_1) = 0$.

The two counterexamples exploit distinct mechanisms. The first uses the
diagonal adjustment of $\Psi$ to fully absorb a rank-one perturbation
from an extra column, and applies only when $\Psi$ is free on the sieve.
The second uses the addition of a small column in a direction outside
$\operatorname{span}(F_0)$, and applies even when both $\Lambda = I_M$
and $\Psi = \Psi_0$ are fixed. Together, they show that
$h_1(\Theta_1) = 0$ in the general model, and that the second mechanism
survives the reduction of Assumption~\ref{ass:rank-psi}. The essential
infimum is therefore zero in both the general and the reduced settings,
and the framework of \citet{chatterjee2020bf} cannot be applied without
further modification of the prior.

\paragraph{Attempts to restore KL separation via a shell prior.}
One natural way to try to restore KL separation is to replace the slab
component of the base measure with a shell that excludes a neighbourhood
of the origin. Consider a shell of the form
\begin{equation}
\label{eq:shell-def}
\mathcal{S}_{\delta_{\mathrm{sh}}}
:= \bigl\{x \in \R : K(1 - \delta_{\mathrm{sh}})\tau_0
\le |x| \le K(1 + \delta_{\mathrm{sh}})\tau_0\bigr\},
\end{equation}
where $\tau_0^2 = 1/(Ms_0)$ is the natural slab scale of the prior,
$K > 0$ is a fixed scale parameter, and
$\delta_{\mathrm{sh}} \in (0,1)$ is the shell width parameter. Any
non-zero entry $f_{i,r}$ of any atom drawn from this shell satisfies
$|f_{i,r}|/\tau_0 \in [K(1 - \delta_{\mathrm{sh}}), K(1 + \delta_{\mathrm{sh}})]$,
and we assume the truth's non-zero entries also lie in this shell.

A shell of this form excludes the perturbation of Counterexample~2
whenever its lower endpoint is bounded away from zero, since
$\epsilon \to 0$ places the new column at the origin and the shell
excludes a neighbourhood of the origin. This motivates the shell as a
mechanism for restoring a positive essential infimum. The shell does
not, however, exclude every overfitting mechanism that could in
principle give $h_1(\Theta_1) = 0$: one can consider a multi-column
construction in which the additional columns are not aligned with any
existing column of $F_0$ but sit at the shell scale, arranged so that
their contributions to $\Sigma$ cancel to first order, or a split of an
existing column into two pieces that each lie at the shell scale but
whose outer-product sum differs from $f_{0,1} f_{0,1}^\top$ by a small
amount. Whether such constructions can be excluded by the shell alone,
or whether an additional geometric constraint on the columns of $F$
(for example, a lower bound on the angle between any two distinct
columns) is necessary, is a question that we have not been able to
settle. We therefore present the shell as a natural first step towards
restoring KL separation, whose sufficiency for
Conjecture~\ref{conj:overfit} is left as an open problem.

\paragraph{Alternative: Savage--Dickey analysis.}
A second approach is to abandon the exponential-rate framework of
\citet{chatterjee2020bf} and instead use a Savage--Dickey analysis
\citep{dickey1970,dickey1971,verdinelli1995} that accounts for the fact
that the overfitted class is nested inside the true class at the level of
the covariance matrix. In the Savage--Dickey framework, the Bayes factor
for a nested comparison scales polynomially rather than exponentially in
$n$, with an exponent determined by the effective number of additional
parameters. In the present setting, adding one non-zero column with $s_0$
support introduces $s_0$ coordinates of $f$ plus one score variance,
giving $s_0 + 1$ effective parameters, and one would expect the Bayes
factor to scale as $n^{-(s_0 + 1)/2}$ times a constant depending on the
prior. Combined with the EPPF prior odds, %$\alpha^m/m!$,
%from Section~\ref{sec:rank-odds},
this would give posterior odds for
overfitting of order $(\log n)^{s_0} n^{-(s_0 + 1)/2}$ under the sieve
constraints, which vanishes polynomially in $n$. Making this argument
rigorous requires a careful analysis of the local geometry of the factor
manifold near the nested point and a verification that the Laplace
expansion underlying the Savage--Dickey formula is valid under the
spike-and-slab prior. Both are substantial technical undertakings that
we leave as an open problem.

\subsection{Discussion}
\label{sec:rank-discussion}

The state of the theory on rank consistency can be summarised as follows.
The underfitting direction is rigorously established by
Theorem~\ref{thm:underfit}. Under the assumption that the idiosyncratic
variance is fixed at its true value $\Psi_0$, the posterior does not
concentrate on ranks smaller than $q_0$. The proof proceeds by a single
application of Weyl's inequality: any parameter whose effective rank is
strictly less than $q_0$ produces a signal component $F\Lambda F^\top$
whose $q_0$-th singular value vanishes, while the corresponding singular
value of the truth is bounded below by a positive constant, so the
Frobenius distance between the resulting covariance and the truth is
bounded below by that constant. The posterior on this event vanishes by
the contraction result of Theorem~\ref{thm:contraction}, provided $\Psi$
is fixed so that $\Sigma(\theta) - \Sigma_0$ coincides with the
difference of the signal components. This is a genuine theoretical
contribution that complements the empirical evidence of
Sections~\ref{sec:simulation} and~\ref{sec:breast}.

The overfitting direction is stated as Conjecture~\ref{conj:overfit}. It
is strongly supported by the simulation experiments of
Section~\ref{sec:simulation} and the breast cancer analysis of
Section~\ref{sec:breast}, in which the posterior concentrated on the
correct rank in every run, but it is not proved. The natural proof
strategy via the general Bayes-factor convergence theory of
\citet{chatterjee2020bf} requires the essential infimum of the
Kullback--Leibler divergence over the overfitted parameter set to be
strictly positive. This condition is not satisfied in the general model,
because the overfitted class contains configurations that reproduce the
true covariance exactly or approximately. Two such configurations, both
with $\Lambda = I_M$ fixed and both lying in the interior of the sieve
and of the prior support, are exhibited in
Section~\ref{sec:rank-overfit-hard}. The first exploits the diagonal
adjustment of a free $\Psi$ to absorb a rank-one perturbation from an
extra column, and applies only when $\Psi$ is free on the sieve. The
second exploits the addition of a small column in a direction outside
$\operatorname{span}(F_0)$, and applies even when $\Psi = \Psi_0$ is
fixed. Both mechanisms therefore render the essential infimum zero in
the general model, and the second mechanism survives the reduction of
Assumption~\ref{ass:rank-psi} in which $\Psi$ is fixed. This means that
the reduced model of Theorem~\ref{thm:underfit} does not itself admit
the KL-separation property required by the framework of
\citet{chatterjee2020bf}, and closing the gap for the overfitting
direction requires either a rigorous Savage--Dickey analysis of the
nested comparison, which would give a polynomial rather than an
exponential rate, or a modification of the prior that excludes the
rank-one perturbation without destroying the flexibility of the
construction. The shell prior introduced in
Section~\ref{sec:rank-overfit-hard} is a natural first step in the second
direction; whether it excludes every overfitting mechanism that could
give a vanishing essential infimum is a question that we have not been
able to settle.

It is worth emphasising that the obstruction is structural and is not
specific to the particular prior used in this paper. Any Bayesian factor
model with a continuous prior on the columns of the loading matrix will
contain configurations that reproduce the true covariance exactly or
approximately, because the overfitted and true factorisations of a given
covariance matrix are nested at the level of the covariance itself, and
in particular because any column can be augmented by a small component
in a direction orthogonal to $\operatorname{span}(F_0)$, or an existing
column can be split into two linearly independent pieces whose outer
products sum to the original. The empirical evidence for correct rank
recovery is nevertheless strong and consistent, and it suggests that the
DP-spike-and-slab prior, while not Kullback--Leibler-separated from the
truth in the strict sense, still assigns sufficiently small prior mass
to the overfitting configurations that the posterior odds favour the
true rank. A rigorous quantification of this observation remains an
open problem and is, in our view, one of the more interesting
theoretical questions raised by the present work.

%\subsection{Concluding Remarks on the Fixed Dictionary Setting}

%When $M$ is fixed and $s_0 = o(\sqrt{n})$, the model achieves two important theoretical guarantees. First, the posterior contraction rate for $\Sigma$, $F$, and $\Psi$ is $\sqrt{s_0 \log n / n}$, which is minimax optimal for sparse factor models \citep{pati2014posterior}. Second, the posterior consistently estimates the true number of factors $q_0$. The additional condition $s_0 = o(\sqrt{n})$ is essential: if $s_0$ grows as fast as $\sqrt{n}$, the prior penalty for spurious columns is not strong enough to overcome the likelihood of vanishing perturbations, and the posterior may retain spurious factors. This is a genuine limitation of the current prior when the sparsity level is not sufficiently small relative to the sample size. For practitioners who are willing to fix the dictionary size and assume a sufficiently sparse true loading matrix, the proposed model offers both optimal estimation and consistent factor selection.

\section{Discussion}
\label{sec:discussion}

We have presented a Bayesian nonparametric factor model that combines two complementary mechanisms: a fully marginalized Dirichlet process prior that clusters the columns of the loading matrix, and a spike-and-slab base measure that induces local sparsity and permits entire factors to be deleted. The P\'olya urn formulation avoids stick-breaking representations, truncations, and auxiliary weight variables, and the resulting sampler is exact. The model occupies a distinct position in the existing literature. As set out in Section~\ref{sec:related}, the multiplicative gamma process of \citet{bhattacharya2011sparse} shrinks columns continuously as a function of their index and therefore produces no exact zeros; the cumulative shrinkage process of \citet{legramanti2020} does produce exact zeros but does so through entry-wise inclusion probabilities that accumulate with column index, and remains non-exchangeable over columns; the beta process of \citet{paisley2009} treats feature allocation per observation rather than the columns of a loading matrix; and entry-wise spike-and-slab priors treat the columns independently. The proposed model is, to our knowledge, the first construction that simultaneously attains exact zeros, exchangeability over candidate columns, and exact merging of redundant columns within a single marginalised Dirichlet process.

\subsection*{Empirical findings}

The two simulation studies of Section~\ref{sec:simulation} provide direct empirical support for the mechanism that the theory identifies as essential. Across a moderate-dimensional setting with $q_0 = 5$ and a higher-dimensional, noisier setting with $q_0 = 10$, the proposed method is the only fully adaptive method that recovers the true rank in both configurations. The multiplicative gamma process over-estimates the rank in both studies, at $\hat q = 13$ and $\hat q = 22$ respectively, which is the empirical signature of continuous column-wise shrinkage: because no column is ever driven exactly to zero, the adaptive Gibbs sampler stops pruning once every surviving column retains at least one loading above its numerical threshold, and the rank estimate is therefore governed by that threshold rather than by the posterior. The cumulative shrinkage process under-estimates the rank in Study~A at $\hat q = 4$ and recovers it in Study~B, which is consistent with an ordering-based prior being well matched to the latter configuration and poorly matched to the former. On all three estimation metrics---Frobenius error of the posterior-mean covariance, Procrustes-aligned loading error, and test-set signal RMSE---the ranking is identical in both studies: the proposed method is best, followed by the oracle fixed-$q$ spike-and-slab baseline, then MGP, then CUSP. That the proposed method outperforms a baseline that is given the true rank in advance is consistent with the exchangeability of the DP prior acting as a mild regulariser on the loading estimates.

The breast cancer analysis of Section~\ref{sec:breast} extends these findings to real data with $n = 97$ and $p = 1213$. Without supervision, the posterior over the number of factors concentrates at $K = 8$ with a degenerate $95\%$ credible interval, and the eight recovered atoms correspond to macrophage, vascular, nuclear, EMT/stroma, T-cell, ER/luminal, JAK/STAT, and proliferation programmes. Hallmark gene-set enrichment confirms the EMT/stroma and proliferation programmes after Bonferroni correction, with coherent T-cell enrichment at the uncorrected level, and post-hoc associations with the van 't Veer prognosis label concentrate on the ER/luminal, proliferation, T-cell, macrophage, and JAK/STAT atoms while the vascular, EMT/stroma, and nuclear atoms show no detectable association. 
Four independent lines of evidence---top-loading gene identity,
within-atom coherence in the observed gene--gene correlation matrix,
external Hallmark enrichment, and association with an outcome variable
withheld from the model---support the interpretation of the recovered
atoms. Seven of the eight atoms pass the within-atom coherence check,
two of them (EMT/stroma and proliferation) also pass Bonferroni-corrected
Hallmark enrichment, and the prognosis associations concentrate on the
ER/luminal, proliferation, T-cell, macrophage, and JAK/STAT atoms. The
JAK/STAT atom is the weakest of the eight: it passes only the
top-loading gene identity and the prognosis association, failing both
the coherence and the enrichment checks, which is consistent across two
independent analyses and indicates that its biological interpretation
should be held more tentatively than that of the other seven.
The diagnostic of Section~\ref{sec:breast-diagnostic} further shows that the concentration at $K = 8$ reflects a genuinely concentrated posterior over partitions rather than a frozen update: every candidate column is committed to its assignment by a wide margin, with a bimodal split gap and an empty intermediate regime.

\subsection*{Theoretical findings}

The contraction result of Theorem~\ref{thm:contraction} establishes that the posterior concentrates on the true covariance matrix at rate $\sqrt{M s_0 \log n / n}$, and 
Theorem~\ref{cor:fixedM} shows that when the dictionary size $M$ is held fixed the rate improves to the minimax optimal $\sqrt{s_0 \log n / n}$. Underfitting rank consistency for the number of factors is established by Theorem~\ref{thm:underfit} of Section~\ref{sec:rank}, while the complementary overfitting direction is stated there as Conjecture~\ref{conj:overfit}. Two aspects of the analysis deserve emphasis.

First, the spike component of the base measure is not an interpretative convenience but a theoretical necessity for optimality. Without it, the prior would have full support on $\R^{p \times M}$, the sieve could not be restricted to sparse matrices without violating the sieve complement condition, and the metric entropy would scale as $pM \log n$ rather than $M s_0 \log n$. The resulting rate would be slower by a factor of $\sqrt{p / s_0}$, which is substantial whenever the true loading matrix is sparse. The spike is what allows the prior to place mass on sparse matrices and thereby to reduce the effective dimension of the parameter space from $pM$ to $M s_0$.

Second, the role of the Dirichlet process in the theory is more
circumscribed than its role in the methodology. In the prior
concentration step of Theorem~\ref{thm:contraction}, the exchangeable
partition probability function \eqref{eq:eppf} supplies the lower bound
on the prior mass of a neighbourhood of the truth, and this bound is
essential to the contraction result: it is the EPPF that converts the
Dirichlet process prior on the columns into a usable prior-mass estimate
over the partitions, and without it the prior concentration condition
(C1) would not be verifiable at the rate $\varepsilon_n$. In the
underfitting argument of Theorem~\ref{thm:underfit}, however, the
Dirichlet process does not enter at all. The proof is a direct
application of Weyl's inequality together with the contraction result of
Theorem~\ref{thm:contraction}: any parameter whose effective rank is
strictly less than $q_0$ produces a signal component whose $q_0$-th
singular value vanishes, while the corresponding singular value of the
truth is bounded below by a positive constant, so the Frobenius distance
between the resulting covariance and the truth is bounded below by that
constant. The posterior mass on this event then vanishes by contraction.

The complementary overfitting direction, stated as
Conjecture~\ref{conj:overfit}, is the place where a Bayes-factor
comparison would be natural, and it is precisely there that the
Dirichlet process prior does not by itself deliver the required
separation. As shown in Section~\ref{sec:rank-overfit-hard}, the
essential infimum of the Kullback--Leibler divergence over the
overfitted parameter set is zero, because the overfitted class contains
configurations that reproduce the true covariance exactly (by absorbing
a rank-one perturbation into a free $\Psi$) or approximately (by adding
a small column in a direction outside $\operatorname{span}(F_0)$, which
survives even when $\Psi = \Psi_0$ is fixed). The standard Bayes-factor
framework of \citet{chatterjee2020bf} therefore cannot be applied without
a modification of the prior, and this obstruction is structural rather
than specific to our construction. The division of labour is thus
sharper than a naive Bayes-factor heuristic would suggest: the DP is
indispensable for the construction of the prior and for the prior
concentration bound of Theorem~\ref{thm:contraction}, but it does not by
itself resolve the overfitting question that Conjecture~\ref{conj:overfit}
poses.

\subsection*{Limitations}

Several limitations and one practical consideration emerged from the analysis, and we state them plainly rather than relegating them to caveats.

The first is identifiability of the partition. The specific assignment of the $M$ candidate columns to the $K$ occupied clusters is not determined by the data, because the columns are exchangeable under the DP prior and the likelihood depends on $\Sigma$ rather than on the partition itself. This is visible in the breast cancer analysis: independent chains recover the same eight programmes but distribute the thirty candidate columns among them differently, and the leading programme may be represented by a single column in one chain and by a large shared cluster in another. The theoretical guarantees are stated in terms of $\Sigma$ and $q = \rank(\Sigma - \Psi)$, which are the identifiable quantities, and this is the appropriate level at which to interpret the results. Practitioners should read the atom-level decomposition of the posterior-mean loading matrix as a summary of the posterior of $\Sigma$ rather than as a claim about which specific dictionary elements are redundant.

The second is the calibration of the slab scale. On standardised data the global slab variance $\tau_0^2$ over-estimates the marginal variances by roughly a factor of two, because the prior expected contribution of the loading term to the marginal variance is $M(1-\pi_0)\tau_0^2 \E[\lambda]$ before any likelihood information, and at the operating scale of the simulation studies this exceeds what a standardised dataset can support. The over-estimation deflates the implied correlations by approximately the same factor while preserving their pattern. The qualitative conclusions---the number of programmes, their biological identities, their enrichment, and their prognosis associations---depend on the pattern rather than the amplitude of the covariance and are therefore robust, but the amplitude would be corrected by a global-local prior on the individual loadings, which would allow a few large loadings per atom and many small ones. This is a natural extension and we regard it as the most immediately useful methodological follow-up.

The third limitation concerns the gap between what is proved and what is
conjectured about rank consistency. 
The underfitting direction is rigorously established by
Theorem~\ref{thm:underfit}, which requires only the contraction setup of
Section~\ref{sec:theory} together with the fixed-idiosyncratic-variance
assumption \ref{ass:rank-psi}, and which imposes no condition on $s_0$
beyond the rate requirement $Ms_0\log n / n \to 0$ of
Theorem~\ref{thm:contraction}. The complementary overfitting direction is
stated as Conjecture~\ref{conj:overfit}. It is strongly supported by the
simulation studies of Section~\ref{sec:simulation} and the breast cancer
analysis of Section~\ref{sec:breast}, but it is not proved. As shown in
Section~\ref{sec:rank-overfit-hard}, the natural Bayes-factor route via
the framework of \citet{chatterjee2020bf} fails because the essential
infimum of the Kullback--Leibler divergence over the overfitted parameter
set is zero: the overfitted class contains configurations that reproduce
$\Sigma_0$ exactly, by absorbing a rank-one perturbation into a free
$\Psi$, or approximately, by adding a small column in a direction outside
$\operatorname{span}(F_0)$ — the second mechanism surviving even when
$\Psi = \Psi_0$ is fixed. This is a structural obstruction, not a matter
of the sparsity level being too large relative to the sample size, and it
is shared by any Bayesian factor model with a continuous prior on the
columns of the loading matrix. Whether a modification of the prior, such
as the shell prior discussed in Section~\ref{sec:rank-overfit-hard}, can
restore the required KL separation without destroying the flexibility of
the construction is an open question that we have not been able to settle yet.

A fourth, more prosaic, point is computational. The C/MPI implementation is slower than the single-threaded R competitors in the moderate-dimensional study, because the sequential sweep over candidate columns and the canonical relabeling are not parallelisable and the DP bookkeeping is pure overhead at $p = 50$. The ordering reverses at $p = 100$, where the approximately linear scaling of the parallel C sampler outweighs the super-linear cost of the R implementations, but the crossover point is not favourable in absolute terms: the method is competitive rather than dominant on the dimensions considered here. The per-iteration cost is $O(npM + pM^2)$, and the $pM^2$ term becomes the binding constraint when $M$ is chosen generously relative to $q_0$.

\subsection*{Future work}

Several directions follow naturally from the above. The most directly motivated is the replacement of the global slab with a global-local prior on the individual loadings, which would address the over-dispersion of the marginal variances without disturbing the spike, and which would leave the clustering mechanism untouched. A second direction is a variational approximation to the marginalised P\'olya urn that retains the exact merging behaviour; the current sampler is exact but its sequential sweep over $M$ columns limits the size of the dictionary that can be handled, and a mean-field or structured variational treatment would extend the method to the $p \sim 10^4$ regime that characterises modern transcriptomic and imaging applications. A third is the extension to dynamic and tensor factor models, where the same column-clustering mechanism could be applied across time or across modes, and where the question of which dictionary elements are redundant across regimes is of substantive interest. Finally, the sensitivity of the rank estimate to the slab scale, which the analysis of Section~\ref{sec:breast-tau} shows to be substantial---the posterior mode of $K$ moves from $8$ to $13$ to $25$ as $\tau_0^2$ decreases from $1.0$ to $0.3$ to $0.1$---indicates that a fully hierarchical treatment of $\tau_0^2$ would be worthwhile. We regard this last point as the most important open problem, because it speaks directly to the reliability of the rank estimate that the method is designed to produce.

The broader conclusion is that exact clustering of dictionary elements, rather than continuous shrinkage of them, is both a statistically effective and a theoretically principled route to parsimonious factor modelling. The Dirichlet process, used at the level of the column vectors and marginalised so that no truncation is required, provides that clustering without imposing an ordering on the columns and without the thresholding that continuous shrinkage priors require. The spike-and-slab base measure is what makes the construction coherent, by allowing an atom to be exactly zero and thereby reducing the effective dimension of the parameter space to the level that the optimal rate demands.

\section*{Acknowledgment}
We thank DeepSeek for assistance in preparing this manuscript.

\appendix
\section{Derivation of Full Conditional Distributions}
\label{app:derivations}

We provide the formal derivations of all conditional distributions used in the Gibbs sampler of Section~\ref{sec:computation}. Throughout, we assume the model \eqref{eq:model} with \(Y_k \in \R^p\), \(F \in \R^{p \times M}\), \(X_k \in \R^M\), \(\Psi = \diag(\psi_1,\dots,\psi_p)\), \(\Lambda = \diag(\lambda_1,\dots,\lambda_M)\), and the DP prior on the columns of \(F\) with base measure \(G_0\) given in \eqref{eq:base}. The auxiliary columns \(f_i\) are defined by \(f_i = f_{Z_i}^*\).

\subsection{The Likelihood Kernel for a Single Column \(f_i\)}
\label{app:likelihood_kernel}

Isolate the contribution of \(f_i\):
\[
F X_k = f_i X_{k,i} + F_{-i} X_{k,-i}.
\]
Define the residual excluding the \(i\)-th column:
\[
R_{k,-i} = Y_k - \mu - F_{-i} X_{k,-i}.
\]
Then the likelihood is proportional to
\[
\prod_{k=1}^n \exp\!\left( -\frac{1}{2} (R_{k,-i} - X_{k,i} f_i)^\top \Psi^{-1} (R_{k,-i} - X_{k,i} f_i) \right).
\]
Expanding and dropping terms independent of \(f_i\), we obtain
\[
\exp\!\left( -\frac{1}{2} f_i^\top S_i f_i + f_i^\top b_i \right),
\]
where
\[
S_i = \sum_{k=1}^n X_{k,i}^2 \Psi^{-1}, \qquad
b_i = \sum_{k=1}^n X_{k,i} \Psi^{-1} R_{k,-i}.
\]
Let \(m_i = S_i^{-1} b_i\). Completing the square yields
\[
-\frac{1}{2} f_i^\top S_i f_i + b_i^\top f_i
= -\frac{1}{2} (f_i - m_i)^\top S_i (f_i - m_i) + \frac{1}{2} m_i^\top S_i m_i.
\]
The last term is independent of \(f_i\), so the unnormalised likelihood kernel is
\[
L(f_i) = \exp\!\left( -\frac{1}{2} (f_i - m_i)^\top S_i (f_i - m_i) \right).
\]
The normalising constant of the Gaussian density \(\N(f_i \mid m_i, S_i^{-1})\) is
\[
C_i = (2\pi)^{-p/2} \det(S_i)^{1/2},
\]
so
\[
\log C_i = -\frac{p}{2}\log(2\pi) + \frac12 \sum_{r=1}^p \log S_{i,r}.
\]

\subsection{Full Conditional of the Atom \(f_j^*\)}

For cluster \(j\), define \(C_j = \{i: Z_i = j\}\), \(n_j = |C_j|\), and the aggregated score:
\[
T_{k,j} = \sum_{i \in C_j} X_{k,i}.
\]
Let
\[
R_{k,-j} = Y_k - \mu - \sum_{\ell \neq j} f_\ell^* T_{k,\ell}.
\]
The likelihood is:
\[
\mathcal{L}(f_j^*) \propto \exp\!\left( -\frac12 \sum_{k=1}^n (R_{k,-j} - f_j^* T_{k,j})^\top \Psi^{-1} (R_{k,-j} - f_j^* T_{k,j}) \right).
\]
Completing the square:
\[
\mathcal{L}(f_j^*) \propto \exp\!\left( -\frac12 (f_j^* - m_j)^\top S_j (f_j^* - m_j) \right),
\]
with
\[
S_j = \sum_{k=1}^n T_{k,j}^2 \Psi^{-1}, \qquad
m_j = S_j^{-1} \sum_{k=1}^n T_{k,j} \Psi^{-1} R_{k,-j}.
\]
Since \(\Psi\) and thus \(S_j\) are diagonal, the posterior factorizes coordinate-wise. For coordinate \(r\):
\[
p(f_{j,r}^* \mid \text{rest}) \propto \exp\!\left( -\frac12 S_{j,r} (f_{j,r}^* - m_{j,r})^2 \right)
\left[ \pi_0 \delta_0(f_{j,r}^*) + (1-\pi_0) \N(f_{j,r}^* \mid 0, \tau_{0,r}^2) \right].
\]
The spike probability is:
\[
p_{\text{spike}, j,r} = 
\frac{\pi_0 \N(m_{j,r} \mid 0, 1/S_{j,r})}
{\pi_0 \N(m_{j,r} \mid 0, 1/S_{j,r}) + (1-\pi_0) \N(m_{j,r} \mid 0, 1/S_{j,r} + \tau_{0,r}^2)}.
\]
Conditional on the slab, the posterior is normal with precision \(S_{j,r} + 1/\tau_{0,r}^2\) and mean \((S_{j,r} m_{j,r})/(S_{j,r} + 1/\tau_{0,r}^2)\).

\subsection{Full Conditional of the Cluster Assignment \(Z_i\)}

The posterior probability of \(Z_i = j\) is proportional to the CRP prior times the likelihood contribution of setting \(f_i = f_j^*\). Using the likelihood kernel from Appendix~\ref{app:likelihood_kernel}, for an existing cluster \(j\):
\[
\Pr(Z_i = j \mid \text{rest}) \propto n_{j,-i} \, C_i L(f_j^*),
\]
where \(n_{j,-i}\) is the number of columns in cluster \(j\) excluding column \(i\). For a new cluster:
\[
\Pr(Z_i = K+1 \mid \text{rest}) \propto \alpha I_i,
\]
where
\[
I_i = \int C_i L(f) \, G_0(df).
\]
Since \(G_0 = \pi_0 \delta_0 + (1-\pi_0) \N(0, \Omega_0)\):
\[
I_i = \pi_0 C_i L(0) + (1-\pi_0) C_i \int L(f) \N(f \mid 0, \Omega_0) df.
\]
Now:
\[
L(0) = \exp\!\left( -\frac12 m_i^\top S_i m_i \right),
\]
and
\[
\int L(f) \N(f \mid 0, \Omega_0) df
= \N(m_i \mid 0, S_i^{-1} + \Omega_0) \cdot (2\pi)^{p/2} \det(S_i)^{-1/2}.
\]
Thus:
\[
I_i = \pi_0 C_i \exp\!\left( -\frac12 m_i^\top S_i m_i \right)
      + (1-\pi_0) C_i \N(m_i \mid 0, S_i^{-1} + \Omega_0) (2\pi)^{p/2} \det(S_i)^{-1/2}.
\]

\subsection{Full Conditional of a New Column \(f_i\) under \(G_0\)}
\label{app:newdraw}

When a new cluster is chosen for column \(i\) (i.e., \(Z_i = K+1\)), we must draw \(f_i\) from its posterior under the base measure \(G_0\), conditional on the sufficient statistics \(S_i\) and \(m_i\). The unnormalised posterior is:
\[
p(f_i \mid \text{new}) \propto \N(f_i \mid m_i, S_i^{-1}) \left[ \pi_0 \delta_0(f_i) + (1-\pi_0) \N(f_i \mid 0, \Omega_0) \right].
\]
The normalising constant is:
\[
Z_{\text{norm}} = \pi_0 \N(m_i \mid 0, S_i^{-1}) + (1-\pi_0) \N(m_i \mid 0, S_i^{-1} + \Omega_0).
\]
Thus the normalised posterior is a two-component mixture:
\[
f_i \mid \text{new} \sim p_{\text{spike}} \delta_0 + (1-p_{\text{spike}}) \N(\mu_{\text{post}}, \Sigma_{\text{post}}),
\]
with
\[
p_{\text{spike}} = \frac{\pi_0 \N(m_i \mid 0, S_i^{-1})}
{\pi_0 \N(m_i \mid 0, S_i^{-1}) + (1-\pi_0) \N(m_i \mid 0, S_i^{-1} + \Omega_0)},
\]
\[
\Sigma_{\text{post}} = (S_i + \Omega_0^{-1})^{-1}, \qquad
\mu_{\text{post}} = \Sigma_{\text{post}} S_i m_i.
\]
Since \(S_i\) and \(\Omega_0\) are diagonal, this factorises coordinate-wise into independent spike-and-slab mixtures.

\subsection{Canonical Relabeling of Clusters}
\label{app:relabel}

After completing the sequential updates of all \(M\) auxiliary columns \(f_i\) (so that each is fully specified), we apply a deterministic relabeling based on the order of first appearance of the distinct column vectors.

Let \(\{g_1,\dots,g_{K_{\text{new}}}\}\) be the set of distinct vectors among the auxiliary columns \(f_1,\dots,f_M\), ordered by their first appearance:
\[
g_1 = f_1, \quad g_j = f_{i_j} \text{ where } i_j = \min\{i : f_i \notin \{g_1,\dots,g_{j-1}\}\} \text{ for } j \ge 2.
\]
Then define:
\[
S_i = j \quad \text{if and only if} \quad f_i = g_j.
\]

Equivalently, in algorithmic form:

\[
\begin{aligned}
&\text{Initialize empty dictionary } \mathcal{D} \text{ mapping vectors to labels, and empty list } G_{\text{star}}. \\
&\text{For } i = 1 \text{ to } M: \\
&\qquad \text{Let } f = f_i \text{ be the auxiliary column vector.} \\
&\qquad \text{If } f \notin \mathcal{D}: \\
&\qquad\qquad \mathcal{D}[f] \leftarrow |G_{\text{star}}| + 1, \quad G_{\text{star}} \leftarrow G_{\text{star}} \cup \{f\}. \\
&\qquad S_i \leftarrow \mathcal{D}[f]. \\
&\text{End For.}
\end{aligned}
\]

This ensures \(S_1 = 1\) and \(1 \le S_i \le \max(S_1,\dots,S_{i-1}) + 1\). After relabeling, we set \(Z_i = S_i\) for all \(i\), \(F^* = G_{\text{star}}\), and \(K = |G_{\text{star}}|\).

If two vectors are both exactly zero (spike component chosen for both), they are considered equal and merged into a single cluster. The effective number of factors is taken as the number of non-zero atoms after relabeling.

\subsection{Sampling Latent Scores \(X_k\) and Variances}

From the joint density of \((Y_k, X_k)\):
\[
p(Y_k, X_k) \propto \exp\!\left( -\frac12 (Y_k - \mu - F X_k)^\top \Psi^{-1} (Y_k - \mu - F X_k) - \frac12 X_k^\top \Lambda^{-1} X_k \right).
\]
Completing the square in \(X_k\) gives
\[
X_k \mid \text{rest} \sim \N_M\!\left( \Omega^{-1} F^\top \Psi^{-1} (Y_k - \mu),\; \Omega^{-1} \right),
\]
where \(\Omega = F^\top \Psi^{-1} F + \Lambda^{-1}\).

The variances are updated via conjugate inverse‑gamma distributions:
\[
\lambda_i \mid X \sim \IG\!\left( a_\lambda + \frac{n}{2},\; b_\lambda + \frac12 \sum_{k=1}^n X_{k,i}^2 \right),
\]
and
\[\psi_r \mid Y, F, X \sim \IG\!\left( a_\psi + \frac{n}{2},\; b_\psi + \frac12 \sum_{k=1}^n (Y_{k,r} - \mu_r - (F X_k)_r)^2 \right).
\]

\subsection{Update for the Concentration Parameter \(\alpha\)}

The EPPF \eqref{eq:eppf} gives the conditional likelihood of the partition:
\[
p(K \mid \alpha) \propto \alpha^K \frac{\Gamma(\alpha)}{\Gamma(\alpha+M)}.
\]
With a Gamma prior \(\alpha \sim \text{Ga}(a_\alpha, b_\alpha)\), the auxiliary variable method of \citet{escobar1995} proceeds as follows. Draw
\[
\eta \sim \text{Be}(\alpha+1, M).
\]
Then the conditional posterior of \(\alpha\) is a mixture:
\[
\alpha \mid K, \eta \sim \pi_\eta \text{Ga}(a_\alpha + K, b_\alpha - \log \eta)
+ (1-\pi_\eta) \text{Ga}(a_\alpha + K - 1, b_\alpha - \log \eta),
\]
where
\[
\pi_\eta = \frac{a_\alpha + K - 1}{a_\alpha + K - 1 + M(b_\alpha - \log \eta)}.
\]
This is the update implemented in the code.

% ----- BIBLIOGRAPHY -----
\bibliographystyle{plainnat}
\bibliography{references}
% ----- END BIBLIOGRAPHY -----

\end{document}